\documentclass[journal,draftcls,onecolumn,twoside]{IEEEtran}\IEEEoverridecommandlockouts
\usepackage{graphicx}
\usepackage{amsmath, amsfonts, amssymb, amsthm}
\usepackage{enumerate}
\usepackage{multirow}
\usepackage{multicol}
\usepackage{color}
\usepackage{makecell}
\usepackage{arydshln}
\usepackage{cases}
\usepackage{bm}
\usepackage{rotating}
\newcommand{\RNum}[1]{\uppercase\expandafter{\romannumeral #1\relax}}
\theoremstyle{definition} \newtheorem{theorem}{Theorem}
\theoremstyle{definition} 
\theoremstyle{definition} \newtheorem{proposition}[theorem]{Proposition}
\theoremstyle{definition} \newtheorem{definition}[theorem]{Definition}
\theoremstyle{definition} \newtheorem{lemma}[theorem]{Lemma}
\theoremstyle{definition} 
\theoremstyle{definition} \newtheorem{example}{Example}
\theoremstyle{definition} 

\theoremstyle{definition} 
\theoremstyle{definition} \newtheorem*{remark}{Remark}
\theoremstyle{definition} \newtheorem*{notation}{Notation}
{\end{list}}

\usepackage[hidelinks]{hyperref}
\begin{document}
\title{Construction of MRD Codes Based on Circular-Shift Operations}
\author{\IEEEauthorblockN{Zhe~Zhai{$^1$},~Sheng~Jin{$^{1}$},~Qifu~Tyler~Sun{$^{1*}$},~Zongpeng~Li{$^{2}$}} \\
\IEEEauthorblockA{{$^1$}
Department of Communication Engineering, University of Science and Technology Beijing, P. R. China \\
{$^2$}  Institute for Network Sciences and Cyberspace, Tsinghua University, P. R. China\\
}

\thanks{$^*~$ Q. T. Sun (Email: qfsun@ustb.edu.cn) is the corresponding author.}
}

\date{}
\maketitle
\sloppy

\begin{abstract}
Most well-known constructions of $(N \times n, q^{Nk}, d)$ maximum rank distance (MRD) codes rely on the arithmetic of $\mathbb{F}_{q^N}$, whose increasing complexity with larger $N$ hinders parameter selection and practical implementation. %
In this work, based on circular-shift operations, we present a construction of $(J \times n, q^{Jk}, d)$ MRD codes with efficient encoding, where $J$ equals to the Euler's totient function of a defined $L$ subject to $\gcd(q, L) = 1$. %
The proposed construction is performed entirely over $\mathbb{F}_q$ and avoids the arithmetic of $\mathbb{F}_{q^J}$. %
Furthermore, we characterize the constructed MRD codes, Gabidulin codes and twisted Gabidulin codes using a set of $q$-linearized polynomials over the row vector space $\mathbb{F}_{q}^N$, and clarify the inherent difference and connection among these code families. %
Specifically, for the case $J \neq m_L$, where $m_L$ denotes the multiplicative order of $q$ modulo $L$, we show that the proposed MRD codes, in a family of settings, are different from any Gabidulin code and any twisted Gabidulin code. %
For the case $J = m_L$, we prove that every constructed $(J \times n, q^{Jk}, d)$ MRD code coincides with a $(J \times n, q^{Jk}, d)$ Gabidulin code, so that an equivalent circular-shift-based construction of Gabidulin codes is obtained which operates directly over $\mathbb{F}_{q}$ and avoids the arithmetic of $\mathbb{F}_{q^J}$. %
In addition, we prove that under some parameter settings, the constructed MRD codes are equivalent to a generalization of Gabidulin codes obtained by summing and concatenating several $(m_L \times n, q^{m_Lk}, d)$ Gabidulin codes. %
When $q=2$, $L$ is prime and $n\leq m_L$, it is analyzed that generating a codeword of the proposed $((L-1) \times n, 2^{(L-1)k}, d)$ MRD codes requires $O(nkL)$ exclusive OR (XOR) operations, while generating a codeword of $((L-1) \times n, 2^{(L-1)k}, d)$ Gabidulin codes, based on customary construction, requires $O(nkL^2)$ XOR operations. %
\end{abstract}

\section{Introduction}
Rank-metric codes are a class of error-correcting codes initially introduced by Delsarte in \cite{Delsarte1978}.  An $(N \times n, q^{Nk}, d)$ rank-metric code $\mathcal{M}$ is a subset of $N \times n$ matrices over $\mathbb{F}_{q}$ with minimum rank distance $d$, %
where the rank distance between two matrices of the same size is defined by the rank of their difference. %
An $(N \times n, q^{Nk}, d)$ rank-metric code $\mathcal{M}$ that meets the Singleton bound $|\mathcal{M}| \leq q^{N(n-d+1)}$ is called a \emph{maximum rank distance} (MRD) code, which has found wide applications in cryptography \cite{cryptosystems}, distributed storage \cite{Roth}, \cite{Silberstein} and subspace codes \cite{Koetter_RNC}-\cite{Subspace Codes_Chen_2}. %
It is important to note that an MRD code is also a maximum distance separable code but not vice versa. 

The first family of MRD codes was independently constructed by Delsarte \cite{Delsarte1978} and Gabidulin \cite{Gabidulin}, with Kshevetskiy and Gabidulin later generalizing this construction in \cite{GG}. %
This family of MRD codes is constructed by means of $q$-linearized polynomials and is commonly referred to as (generalized) \emph{Gabidulin} codes. %
In 2016, Sheekey \cite{TG} introduced a new class of MRD codes by evaluating linearized polynomials, known as \emph{twisted Gabidulin} codes. %
Unlike Gabidulin codes, the linearized polynomials in twisted Gabidulin codes include an extra nonzero monomial with a $q$-degree greater than $k-1$. %
To ensure that the constructed codes satisfy the MRD property, the coefficient of the monomial must be carefully selected according to a defined condition. %
However, this construction is not applicable for the case $q = 2$. %
Twisted Gabidulin codes were further generalized and studied in \cite{TG} and \cite{GTG}, where they are referred to as (generalized) twisted Gabidulin codes. %
As an additional generalization, \emph{additive generalized twisted Gabidulin} codes were introduced in \cite{Otal K 1}. %
A comprehensive review of existing constructions of MRD codes and their applications was provided by Sheekey in \cite{Sheekey_2019}. %
Yildiz and Hassibi \cite{Yildiz} established conditions for the existence of Gabidulin codes and general MRD codes from the perspective of support constrained generator matrices, and characterized the largest rank distance under support constraints, showing that it can be achieved by subcodes of Gabidulin codes. %
More recently, a new class of rank-metric codes, called $\bm{\lambda}$-\emph{twisted Gabidulin} codes, was introduced in \cite{Li_ITW}. %
This class is constructed by multiplying each column of the generator matrix of a twisted Gabidulin code by the corresponding component of an $n$-dimensional vector $\boldsymbol{\lambda}$ over $\mathbb{F}_{q^m}$. %
More results on MRD codes can be found in \cite{non-Gabidulin}-\cite{end}. %

However, to the best of our knowledge, most constructions of $(N \times n, q^{Nk}, d)$ MRD codes with $n \leq N$ are performed over $\mathbb{F}_{q^N}$ and rely on the arithmetic over the extension field $\mathbb{F}_{q^N}$. %
In particular, in the conventional way to construct a Gabidulin MRD code, a linear code over $\mathbb{F}_{q^N}$ is first constructed and then transformed to a rank-metric code over $\mathbb{F}_q$ with respect to a defined basis of $\mathbb{F}_{q^N}$ over $\mathbb{F}_q$. %
As $n$ increases, the extension field $\mathbb{F}_{q^N}$ required for constructing MRD codes becomes significantly larger, resulting in much more complicated implementation. %
In addition, it is not well investigated how to properly select a basis of $\mathbb{F}_{q^N}$ over $\mathbb{F}_q$ so that the transformation of an $n$-dimensional row vector over $\mathbb{F}_{q^N}$ to an $N\times n$ matrix over $\mathbb{F}_q$ can be efficiently performed. It turns out that from a practical point of view, parameter $N$ cannot be selected too large. The inflexibility of parameter selection hinders the practicability of MRD codes.  %

In this work, motivated by circular-shift linear network coding (LNC) \cite{Tang_LNC_TIT}-\cite{Jin_TIT}, we present a construction of MRD codes based on circular-shift operations with efficient encoding, referred to as \emph{circular-shift-based} MRD codes. %
Let $q$ be a prime, $L$ be a positive integer satisfying  $\gcd(q, L) = 1$, $m_L$ denote the multiplicative order of $q$ modulo $L$, \emph{i.e.,} $q^{m_L}=1 \mod L$ and $J$ equal to the Euler's totient function of $L$. The main contributions of this paper are summarized as follows:
\begin{itemize}
\item Based on circular-shift operations, we present a construction of $(J \times n, q^{Jk}, d)$ MRD codes with $n \leq m_L (\leq J)$. The new codes can be expressed in the following form
    \begin{equation}
    \mathcal{C} = \{\Delta(\mathbf{m}(\mathbf{I}_k\otimes\mathbf{P})\mathbf{\Psi}_{k\times n}(\mathbf{I}_n\otimes\mathbf{Q})): \mathbf{m}\in \mathbb{F}_{q}^{Jk}\}.
    \end{equation}
    Herein, $\mathbf{P}$ and $\mathbf{Q}$ are respectively $J\times L$ and $L\times J$ matrices over $\mathbb{F}_{q}$ constructed by two general methods, $\mathbf{\Psi}_{k\times n}$ denotes a $k \times n$ block matrix with every block to be an $L\times L$ circulant matrix over $\mathbb{F}_q$, $\otimes$ denotes the Kronecker product, and $\Delta$ denotes a mapping from a $Jn$-dimensional row vector over $\mathbb{F}_q$ to a $J\times n$ matrix over $\mathbb{F}_q$. %
    In particular,  different from conventional construction of $(J \times n, q^{Jk}, d)$ Gabidulin codes, the proposed construction gets rid of the arithmetic of $\mathbb{F}_{q^{J}}$ (See Definition \ref{def:circular-shift-MRD-code}, Theorem \ref{the:Circular_MRD_GH} and \ref{the:Circular_MRD_GG} in Section \ref{sec:Circular-shift-based MRD codes}). %
\item We equivalently characterize the proposed MRD codes and Gabidulin codes based on a set of $q$-linearized polynomials over the row vector space $\mathbb{F}_{q}^N$, and prove that their polynomial evaluations are different (See Proposition \ref{prop:new charac of Gabidulin} and Lemma \ref{GCH neq VAV} in Section \ref{sec:Connection with Gabidulin codes}). %
   Based on this new characterization, we further show that for the case $J \neq m_L$, in a family of settings, the proposed MRD codes are different from any $(J \times n, q^{Jk}, d)$ Gabidulin code and any $(J \times n, q^{Jk}, d)$ twisted Gabidulin code under specific parameter settings (See Theorem \ref{theorem:connection with Gabidulin} and Proposition \ref{propo:connection with twisted_Gabidulin} in Section \ref{sec:Connection with Gabidulin codes}). %
   We also provide explicit examples demonstrating that outside these settings, the proposed MRD codes may either coincide with or differ from $(J\times n, q^{Jk}, d)$ Gabidulin codes. %
\item For the case  $J = m_L$, we prove that every proposed MRD code coincides with a Gabidulin code (See Theorem \ref{theorem:connection with Gabidulin} in Section \ref{sec:Connection with Gabidulin codes}). 
   Despite this, the proposed construction has its own merit because it provides a new approach to construct Gabidulin codes purely based on the arithmetic of $\mathbb{F}_{q}$, different from the customary constructions that rely on the arithmetic of $\mathbb{F}_{q^{J}}$.
\item In addition to the difference between the proposed MRD codes and Gabidulin codes we unveiled, we further show that the proposed MRD codes, under some particular selection of $\mathbf{P}$ and $\mathbf{Q}$,  are equivalent to a generalization of Gabidulin codes obtained by summing and concatenating several $(m_L \times n, q^{m_Lk}, d)$ Gabidulin codes (See Theorem \ref{generalization} in Section \ref{sec:Connection with Gabidulin codes}). %
    To the best of our knowledge, this generalization has not been explored in the literature. %
\item When $q=2$, $L$ is prime and $n \leq m_L$, theoretical analysis shows that generating a codeword of an $((L-1) \times n, 2^{(L-1)k}, d)$ circular-shift-based MRD code requires $O(nkL)$ exclusive OR (XOR) operations, while generating a codeword of an $((L-1) \times n, 2^{(L-1)k}, d)$ Gabidulin code, based on the customary construction, requires $O(nkL^2)$ XOR operations (See Section \ref{sec:Complexity Analysis}).
\end{itemize} 

The remainder of this paper is organized as follows. %
Section \ref{sec:Preliminaries} reviews the fundamentals on the vector space representation of $\mathbb{F}_{q^N}$ over $\mathbb{F}_q$, $q$-linearized polynomial and two equivalent representations of Gabidulin codes. %
Section \ref{sec:Circular-shift-based MRD codes} introduces the new construction of MRD codes based on circular-shift operations. %
Section \ref{sec:Connection with Gabidulin codes} first characterizes Gabidulin codes in terms of $q$-linearized polynomials over the row vector space $\mathbb{F}_{q}^N$ instead of over $\mathbb{F}_{q^N}$, and then clarifies the inherent difference and connection between the proposed circular-shift-based MRD codes and Gabidulin codes. %
Section \ref{sec:Complexity Analysis} provides a theoretical complexity analysis for generating a codeword of the constructed MRD codes and of Gabidulin codes for the case that $q=2$ and $L$ is prime. %
Section \ref{section:Concluding Remarks} summarizes the paper. 
Most technical proofs are placed in Appendices.

\section{Preliminaries}
\label{sec:Preliminaries}
\begin{notation}
In the remaining part of this paper, $\mathbb{F}_q$ denotes a finite field consisting of $q$ elements, and $\mathbb{F}_{q^N}$ denotes the extension field of $\mathbb{F}_q$ of degree $N$. %
The notations $\mathbb{F}_{q^N}^{N \times n}$ and $\mathbb{F}_{q^N}^{n}$ represent the set of all $N \times n$ matrices over $\mathbb{F}_{q^N}$ and the set of all $n$-dimensional row vectors over $\mathbb{F}_{q^N}$, respectively. %
The Kronecker product is denoted by $\otimes$. %
For $l \geq 1$, the $l \times l$ identity matrix is denoted by $\mathbf{I}_l$. %
In addition, $\mathbf{0}$ and $\mathbf{1}$ respectively represent an all-zero and all-one matrix, whose size, if not explicitly explained, can be inferred from the context. %
For an arbitrary matrix $\mathbf{A}$, let $\mathbf{A}^{\circ j}$ denote the Hadamard power, that is, at every same location, the entry $a$ in $\mathbf{A}$ becomes $a^j$ in $\mathbf{A}^{\circ j}$. %
For a row vector $\mathbf{v}\in \mathbb{F}_{q^N}^n$, let $\mathbf{M}_\mathbf{o}(\mathbf{v})$ denote the following matrix in $\mathbb{F}_{q^N}^{n \times N}$
\begin{equation}
\label{def:Mo}
\mathbf{M}_\mathbf{o}(\mathbf{v})=[\mathbf{v}^\mathrm{T}~(\mathbf{v}^\mathrm{T})^{\circ q}~\ldots~(\mathbf{v}^\mathrm{T})^{\circ q^{N-1}}].
\end{equation}
Frequently used notations are listed in Appendix-\ref{appendix:list of notation} for reference.
\end{notation}
\subsection{Vector space representation of extension field}
There are several equivalent ways to represent elements in the extension field $\mathbb{F}_{q^N}$ over the base field $\mathbb{F}_{q}$. The one that interests us most in the context of MRD codes is the vector space representation.

Let $\alpha$ be such a nonzero element in $\mathbb{F}_{q^N}$ that it is a root of a certain irreducible polynomial of degree $N$ over $\mathbb{F}_q$. Then, all $q^N$ elements in $\mathbb{F}_{q^N}$ can be represented as $\mathbb{F}_q$-linear combinations of $\{1, \alpha, \ldots, \alpha^{N-1}\}$
\begin{equation}
\mathbb{F}_{q^N} = \{a_0 + a_1\alpha +\ldots + a_{N-1}\alpha^{N-1}: a_j \in \mathbb{F}_q \}.
\end{equation}
Herein, $[1~ \alpha~ \ldots~\alpha^{N-1}]$ is called a \emph{polynomial basis} of $\mathbb{F}_{q^N}$ over $\mathbb{F}_{q}$.

Analogously, $N$ elements $\{\beta_0, \beta_1, \ldots, \beta_{N-1}\}$ in $\mathbb{F}_{q^N}$ are said to be $\mathbb{F}_q$-\emph{linearly independent} if for every $a_0, \ldots, a_{N-1} \in \mathbb{F}_q$, $a_0\beta_0 + a_1\beta_1 + \ldots + a_{N-1}\beta_{N-1}$ represents a different element in $\mathbb{F}_{q^N}$. Every set $\mathcal{B}$ of $N$ $\mathbb{F}_q$-linearly independent elements in $\mathbb{F}_{q^N}$ forms a \emph{basis} of $\mathbb{F}_{q^N}$, so that every element $\gamma \in \mathbb{F}_{q^N}$ can be \emph{uniquely} written as an $\mathbb{F}_q$-linear combination of elements in $\mathcal{B}$. %
Equivalently, $\gamma$ corresponds to a row vector in $\mathbb{F}_{q}^{N}$, that is, there is a bijection between $\mathbb{F}_{q^N}$ and the $N$-dimensional vector space $\mathbb{F}_q^N$ over $\mathbb{F}_q$ with respect to $\mathcal{B}$. %
Hereafter in this paper, a basis of $\mathbb{F}_{q^N}$ over $\mathbb{F}_{q}$ is represented as a row vector in $\mathbb{F}_{q^N}^N$.

Let $\mathcal{B} = [\omega_0~\omega_1~\ldots~\omega_{N-1}]$ denote a basis of $\mathbb{F}_{q^N}$ over $\mathbb{F}_{q}$. %
For an arbitrary row vector $\mathbf{v}=[v_0 ~ v_1 ~\ldots~v_{n-1}]\in \mathbb{F}_{q^N}^n$, each element $v_i, 0\leq i \leq n-1$, can be expressed as an $\mathbb{F}_q$-linear combination of $\{\omega_0, \omega_1, \ldots, \omega_{N-1}\}$
\begin{equation}
v_i = a_{i,0}\omega_0 + a_{i,1}\omega_1 +\ldots + a_{i,N-1}\omega_{N-1}, a_{i,j} \in \mathbb{F}_q .
\end{equation}
Corresponding to $\mathbf{v}$, the counterpart matrix $\mathbf{M}_{\mathcal{B}}(\mathbf{v})$ in $\mathbb{F}_q^{N\times n}$ (with respect to basis $\mathcal{B}$) is expressed as follows
\begin{equation}
\label{eqn:counterpart matrix}
\mathbf{M}_{\mathcal{B}}(\mathbf{v})=\begin{bmatrix}
a_{0,0}&a_{1,0}&\ldots&a_{n-1,0}\\
a_{0,1}&a_{1,1}&\ldots&a_{n-1,1}\\
\vdots&\vdots&\ldots&\vdots\\
a_{0,N-1}&a_{1,N-1}&\ldots&a_{n-1,N-1}\\
\end{bmatrix}.
\end{equation}
Consequently, we have 
\begin{equation}
\label{eqn:counterpart matrix eqn}
\mathcal{B}\mathbf{M}_{\mathcal{B}}(\mathbf{v})=\mathbf{v}.
\end{equation}
For a given basis $\mathcal{B}$ of $\mathbb{F}_{q^N}$ over $\mathbb{F}_q$, there exists a unique \emph{dual basis} $\mathcal{B}'$ such that (See e.g., \cite{rank codes})
\begin{equation}
\label{def:dual basis}
\mathbf{M}_\mathbf{o}(\mathcal{B})^\mathrm{T}\mathbf{M}_\mathbf{o}(\mathcal{B}')=\mathbf{M}_\mathbf{o}(\mathcal{B})\mathbf{M}_\mathbf{o}(\mathcal{B}')^\mathrm{T}=\mathbf{I}_N,
\end{equation}
where the matrices $\mathbf{M}_\mathbf{o}(\mathcal{B})$ and $\mathbf{M}_\mathbf{o}(\mathcal{B}')$ in $\mathbb{F}_{q^N}^{N \times N}$ can be obtained by \eqref{def:Mo}.

\subsection{$q$-linearized polynomial}
The $q$-linearized polynomial plays a central role in the construction of MRD codes.
\begin{definition}
A polynomial in the form $L(x) = \sum_{j=0}^{n-1} \beta_jx^{q^j}$, $\beta_j \in \mathbb{F}_{q^N}$, is called a $q$-\emph{linearized polynomial} over $\mathbb{F}_{q^N}$ .
\end{definition}

A $q$-linearized polynomial over $\mathbb{F}_{q^N}$ is relatively easy to handle because of the following nice  properties:
\begin{itemize}
\item $L(\beta_1 + \beta_2) = L(\beta_1)+L(\beta_2)$ for all $\beta_1, \beta_2 \in \mathbb{F}_{q^N}$;
\item $L(c\beta) = cL(\beta)$ for all $c \in \mathbb{F}_q, \beta \in \mathbb{F}_{q^N}$;
\end{itemize}
\begin{lemma}
\label{lemma:q-polynomial}
Given a $q$-linearized polynomial $L(x)$ over $\mathbb{F}_{q^N}$ and $n$ elements $\beta_0, \ldots, \beta_{n-1} \in \mathbb{F}_{q^N}$, if
\begin{equation}
c_0L(\beta_0) + c_1L(\beta_1) + \ldots + c_{n-1}L(\beta_{n-1}) = 0
\end{equation}
for some $c_0, \ldots, c_{n-1} \in \mathbb{F}_q$, then $\sum\nolimits_{0 \leq i \leq n-1} c_i\beta_i$ is a root of $L(x)$ and for all $j \geq 0$,
\begin{equation}
\label{eqn:q-polynomial-lemma}
c_0L(\beta_0)^{q^j} + c_0L(\beta_1)^{q^j} + \ldots + c_{n-1}L(\beta_{n-1})^{q^j} = 0.
\end{equation}
\begin{IEEEproof}
By the linearity property of $L(x)$, $L(c_0\beta_0 + c_1\beta_1 + \ldots + c_{n-1}\beta_{n-1}) = c_0L(\beta_0) + c_1L(\beta_1) + \ldots + c_{n-1}L(\beta_{n-1}) = 0$, that is, $\sum_{0 \leq i \leq n-1} c_i\beta_i$ is a root of $L(x)$. Thus,
\begin{equation}
0=L(\sum\nolimits_{0 \leq i \leq n-1} c_i\beta_i)^{q^j}= L(\sum\nolimits_{0 \leq i \leq n-1} c_i\beta_i^{q^j}) 
= \sum\nolimits_{0 \leq i \leq n-1} c_iL(\beta_i^{q^j}) = \sum\nolimits_{0 \leq i \leq n-1} c_iL(\beta_i)^{q^j}.
\end{equation}
\end{IEEEproof}
\end{lemma}

\subsection{Gabidulin codes}
\label{subsec:Gabidulin codes}
In the literature, there are two representations of codes in \emph{rank metric}: \emph{matrix} representation \cite{Delsarte1978} and \emph{vector} representation \cite{Gabidulin}. %
For two matrices $\mathbf{M}_1$, $\mathbf{M}_2$ in $\mathbb{F}_{q}^{N \times n}$, the \emph{rank distance} between them is defined as
\begin{equation}
\label{eqn:rank_metric_matrix_rep}
d_{Rk}(\mathbf{M}_1, \mathbf{M}_2) = \mathrm{rank}(\mathbf{M}_1 - \mathbf{M}_2).
\end{equation}
A set $\mathcal{M}$ of matrices in $\mathbb{F}_q^{N\times n}$ is called an $(N \times n, q^{Nk}, d)$ rank-metric code with the minimum distance
\begin{equation}
d= \min_{\scriptstyle\mathbf{M}_1, \mathbf{M}_2 \in \mathcal{M}\hfill\atop\scriptstyle \mathbf{M}_1 \neq \mathbf{M}_2 \hfill} d_{Rk}(\mathbf{M}_1, \mathbf{M}_2).
\end{equation}
For a row vector $\mathbf{v}\in \mathbb{F}_{q^N}^n$, by representing every element in $\mathbf{v}$ as its corresponding $N$-dimensional column vector over $\mathbb{F}_q$ with respect to a defined basis $\mathcal{B}$ of $\mathbb{F}_{q^N}$ over $\mathbb{F}_q$, we obtain a counterpart matrix $\mathbf{M}_{\mathcal{B}}(\mathbf{v})\in \mathbb{F}_q^{N\times n}$. The \emph{rank} of $\mathbf{v}$ just refers to the rank of $\mathbf{M}_{\mathcal{B}}(\mathbf{v})$.
Thus, for two row vectors $\mathbf{v}_1$, $\mathbf{v}_2 \in \mathbb{F}_{q^N}^n$, the \emph{rank distance} defined as
\begin{equation}
\label{eqn:d_R matrix}
d_{Rk}(\mathbf{v}_1, \mathbf{v}_2)=\mathrm{rank}(\mathbf{M}_{\mathcal{B}}(\mathbf{v}_1)-\mathbf{M}_{\mathcal{B}}(\mathbf{v}_2)).
\end{equation}
A set $\mathcal{V}$ of row vectors in $\mathbb{F}_{q^N}^n$ is called an $(N \times n, q^{Nk}, d)$ rank-metric code with the minimum distance
\begin{equation}
d= \min_{\scriptstyle\mathbf{v}_1, \mathbf{v}_2 \in \mathcal{V}\hfill\atop\scriptstyle \mathbf{v}_1 \neq \mathbf{v}_2 \hfill} d_{Rk}(\mathbf{M}_{\mathcal{B}}(\mathbf{v}_1), \mathbf{M}_{\mathcal{B}}(\mathbf{v}_2)).
\end{equation}

For an $(N \times n, q^{Nk}, d)$ rank-metric code $\mathcal{M}$, the number of codewords in $\mathcal{M}$ satisfies the Singleton bound $|\mathcal{M}| \leq q^{N(n-d+1)}$. %
A rank-metric code that meets the Singleton bound is called a \emph{maximum rank distance} (MRD) code.

Gabidulin codes, the most well-known MRD codes, have two equivalent descriptions: one from the perspective of matrix representation introduced by Delsarte \cite{Delsarte1978} and the other from the perspective of vector representation introduced by Gabidulin \cite{Gabidulin}. %
In both formulations, the MRD property of the constructed codes can be easily derived by making use of $q$-linearized polynomials over $\mathbb{F}_{q^N}$. %
Let $k=n-d+1$ and $d\leq n \leq N$. %
We begin by reviewing the two representations of Gabidulin codes and then explicitly characterize their inherent connection.

We first review the construction of Gabidulin codes from the perspective of matrix representation. %
Let $\mathcal{B} = [\omega_0~\omega_1~\ldots~\omega_{N-1}]$ denote an arbitrary basis of $\mathbb{F}_{q^N}$ over $\mathbb{F}_q$ used to define the matrix in \eqref{eqn:MO_B}, and $\beta_0, \beta_1, \ldots, \beta_{n-1}$ denote $n$ $\mathbb{F}_q$-linearly independent elements in $\mathbb{F}_{q^N}$ used to define the matrix in \eqref{eqn:L}. %
Note that $n\leq N$, and that even when $n = N$, the elements $\omega_0, \omega_1, \ldots,\omega_{N-1}$ are not necessarily same as $\beta_0, \beta_1, \ldots, \beta_{n-1}$. %
For a row vector $\mathbf{u}=[u_0, u_1, \ldots, u_{k-1}] \in \mathbb{F}_{q^N}^k$, let $L_\mathbf{u}(x)$ denote the following $q$-linearized polynomial over $\mathbb{F}_{q^N}$ 
\begin{equation}
\label{eqn:L_u(x)}
L_\mathbf{u}(x)=\sum\nolimits_{s=0}^{k-1} u_sx^{q^s}.
\end{equation}
Define the matrix $\mathbf{M}(\mathbf{u}) \in \mathbb{F}_q^{N\times n}$ as
\begin{equation}
\label{eqn:Mu}
\mathbf{M}(\mathbf{u})=\mathbf{M}_\mathbf{o}(\mathcal{B})\mathbf{L}_\mathbf{u},
\end{equation}
where
\begin{align}
\label{eqn:MO_B}
&\mathbf{M}_\mathbf{o}(\mathcal{B})=
\begin{bmatrix}
\omega_0 & \omega_0^q &\ldots&\omega_0^{q^{N-1}}\\
\omega_1  & \omega_1^q &\ldots&\omega_1^{q^{N-1}}\\
\vdots&\vdots&\vdots&\vdots\\
\omega_{N-1}  & \omega_{N-1}^q &\ldots&\omega_{N-1}^{q^{N-1}}\\
\end{bmatrix},\\
\label{eqn:L}
&\mathbf{L}_\mathbf{u}=
\begin{bmatrix}
L_\mathbf{u}(\beta_0) & L_\mathbf{u}(\beta_1) & \ldots & L_\mathbf{u}(\beta_{n-1}) \\
L_\mathbf{u}(\beta_0)^q & L_\mathbf{u}(\beta_1)^q & \ldots & L_\mathbf{u}(\beta_{n-1})^q \\
 \vdots & \vdots & & \vdots \\
L_\mathbf{u}(\beta_0)^{q^{N-1}} & L_\mathbf{u}(\beta_1)^{q^{N-1}} & \ldots & L_\mathbf{u}(\beta_{n-1})^{q^{N-1}} \\
\end{bmatrix}.
\end{align}

\begin{definition}
\label{def:matrix  representation}
An $(N \times n, q^{Nk}, d)$ rank-metric code $\mathcal{M}$ consists of the following $q^{Nk}$ matrices in $\mathbb{F}_q^{N\times n}$ 
\begin{equation}
\label{eqn:circular-shift-MRD-code-DEF}
\mathcal{M} = \{\mathbf{M}(\mathbf{u}): \mathbf{u}\in \mathbb{F}_{q^N}^k\}.
\end{equation}
\end{definition}
It is well-known (See, e.g., in \cite{Delsarte1978}) that the $(N \times n, q^{Nk}, d)$ rank-metric code $\mathcal{M} = \{\mathbf{M}(\mathbf{u}): \mathbf{u}\in \mathbb{F}_{q^N}^k\}$ is an MRD code.

We next review the construction of Gabidulin codes from the perspective of vector representation. %
Define the following matrix $\mathbf{G}\in \mathbb{F}_{q^N}^{k \times n}$
\begin{equation}
\label{eqn:MRD_generator_matrix}
\mathbf{G} = \begin{bmatrix}
\beta_0 & \beta_1 & \cdots & \beta_{n-1} \\
\beta_0^q & \beta_1^q & \cdots & \beta_{n-1}^q \\
 \vdots & \vdots & & \vdots \\
\beta_0^{q^{k-1}} & \beta_1^{q^{k-1}} & \cdots & \beta_{n-1}^{q^{k-1}} \\
\end{bmatrix}.
\end{equation}
Recall that with respect to a defined basis $\mathcal{B}$ of $\mathbb{F}_{q^N}$ over $\mathbb{F}_{q}$, every row vector $\mathbf{v}\in \mathbb{F}_{q^N}^{n}$ can be represented as a counterpart matrix $\mathbf{M}_{\mathcal{B}}(\mathbf{v})\in \mathbb{F}_q^{N\times n}$ according to \eqref{eqn:counterpart matrix}. %

\begin{definition}
\label{def:Gabidulin codes}
Let $k \leq n\leq N$. The $(n, k)$ linear code $\mathcal{V} = \{\mathbf{u}\mathbf{G}: \mathbf{u}\in \mathbb{F}_{q^N}^k\}$ over $\mathbb{F}_{q^N}$ is called a Gabidulin (linear) code over $\mathbb{F}_{q^N}$. %
With respect to a defined basis $\mathcal{B}$ of $\mathbb{F}_{q^N}$ over $\mathbb{F}_{q}$, the set $\mathcal{M}_\mathcal{B}(\mathcal{V}) = \{\mathbf{M}_{\mathcal{B}}(\mathbf{v}), \mathbf{v} \in \mathcal{V}\}$ of matrices in $\mathbb{F}_q^{N\times n}$ induced from $\mathcal{V}$ is called an $(N \times n, q^{Nk}, d)$ Gabidulin (rank-metric) code.
\end{definition}

It is well-known (See, e.g., in \cite{Gabidulin}) that for an arbitrary selection of the basis $\mathcal{B}$, the $(N \times n, q^{Nk}, d)$ Gabidulin code $\mathcal{M}_\mathcal{B}(\mathcal{V})$ is an MRD code. %
When a Gabidulin code $\mathcal{V}$ is regarded as a linear code over $\mathbb{F}_{q^N}$, it can be transformed into different MRD codes $\mathcal{M}_\mathcal{B}(\mathcal{V})$ over $\mathbb{F}_{q}$ by different selections of $\mathcal{B}$ of $\mathbb{F}_{q^N}$ over $\mathbb{F}_{q}$. %
Assume that $\mathcal{B}_1$ and $\mathcal{B}_2$ are two bases of $\mathbb{F}_{q^N}$ over $\mathbb{F}_q$. Let $\mathbf{V}$ be the invertible matrix in $\mathbb{F}_q^{N\times N}$ subject to $\mathcal{B}_1= \mathcal{B}_2\mathbf{V}$. For a Gabidulin linear code $\mathcal{V}$ over $\mathbb{F}_{q^N}$, the two different corresponding $(N \times n, q^{Nk}, d)$ Gabidulin MRD codes $\mathcal{M}_{\mathcal{B}_1}(\mathcal{V})$ and $\mathcal{M}_{\mathcal{B}_2}(\mathcal{V})$ satisfy $\mathbf{V}\mathcal{M}_{\mathcal{B}_1}(\mathcal{V}) = \mathcal{M}_{\mathcal{B}_1\mathbf{V}^{-1}}(\mathcal{V}) = \mathcal{M}_{\mathcal{B}_2}(\mathcal{V})$. %
However, to the best of our knowledge, from a practical point of view, proper selection of basis of $\mathbb{F}_{q^N}$ over $\mathbb{F}_q$ has not been thoroughly explored, particularly in terms of how to efficiently transform a row vector in $\mathbb{F}_{q^N}^{n}$ to a matrix in $\mathbb{F}_q^{N\times n}$. %

In the following proposition, we show that with particular choices of basis, the Gabidulin codes defined in Definition \ref{def:Gabidulin codes} are equivalent to the MRD codes defined in Definition \ref{def:matrix  representation}.

\begin{proposition}
\label{prop:Connections between M1 and M2}
Let $\mathcal{B} = [\omega_0~\omega_1~\ldots~\omega_{N-1}]$ be a basis of $\mathbb{F}_{q^N}$ over $\mathbb{F}_q$, $\mathcal{B}' = [\omega_0'~\omega_1'~\ldots~\omega_{N-1}']$ be the dual basis of $\mathcal{B}$. Assume that the $n$ $\mathbb{F}_q$-linearly independent elements $\beta_0, \ldots, \beta_{n-1}$ in $\mathbb{F}_{q^N}$ used in defining $\mathbf{L}_\mathbf{u}$ of \eqref{eqn:L} are identical to those adopted in defining $\mathbf{G}$ of \eqref{eqn:MRD_generator_matrix}. %
Denote by $\mathcal{M}$ the MRD code formulated in Definition \ref{def:matrix  representation}, and by $\mathcal{M}_{\mathcal{B}'}(\mathcal{V})$ the Gabidulin code  with respect to the basis $\mathcal{B}'$ formulated in Definition \ref{def:Gabidulin codes}. We have
\begin{equation}
\label{eqn:M1 and M2}
\mathcal{M} = \mathcal{M}_{\mathcal{B}'}(\mathcal{V}).
\end{equation}
\begin{IEEEproof}
Please refer to Appendix-\ref{appendix:proof proposition M1M2}
\end{IEEEproof}
\end{proposition}

\begin{example}
\label{example:example 1}
Consider $q=2$. %
Let $\alpha$ be a root of the primitive polynomial $p(x)=x^4+x+1$ over $\mathbb{F}_2$. It can be shown $\mathcal{B}=[1~ \alpha~ \alpha^2~\alpha^3]$ is a basis of $\mathbb{F}_{2^4}$ over $\mathbb{F}_2$ and $\mathcal{B}'=[\alpha^{14}~ \alpha^2~\alpha~1]$ is the dual basis of $\mathcal{B}$. %
Choose the $\mathbb{F}_2$-linearly independent elements $[\beta_0, \beta_1, \beta_2, \beta_3]=[\alpha, \alpha^2, \alpha^3, \alpha^4]$ to construct $\mathbf{L}_\mathbf{u}$ in \eqref{eqn:L}. %
Set 
\begin{equation*}
\mathbf{G}=[\alpha~\alpha^2~\alpha^3~\alpha^4]. 
\end{equation*}
For each 
\begin{equation*}
\mathbf{u} \in \{0, 1, \alpha, \alpha^2, \alpha^3, \alpha^4, \alpha^5, \alpha^6, \alpha^7, \alpha^8, \alpha^9, \alpha^{10}, \alpha^{11}, \alpha^{12},\alpha^{13}, \alpha^{14}\}=\mathbb{F}_{2^4},
\end{equation*}
the codewords in the $(4\times4, 2^4, 4)$ MRD code  $\mathcal{M} = \{\mathbf{M}(\mathbf{u}): \mathbf{u}\in \mathbb{F}_{2^4}\}$ as defined in Definition \ref{def:matrix  representation} are
\begin{align*}
&
\begin{bmatrix}\begin{smallmatrix}
   0&   0&   0 &0\\
   0&   0&   0 &0\\
   0&   0&   0 &0\\
   0&   0&   0 &0
   \end{smallmatrix}\end{bmatrix},~
\begin{bmatrix}\begin{smallmatrix}
   0&   0&   1&   0\\
   0&   1&   0&   0\\
   1&   0&   0&   1\\
   0&   0&   1&   1
   \end{smallmatrix}\end{bmatrix},~
\begin{bmatrix}\begin{smallmatrix}
   0&   1&  0&   0\\
   1&   0&   0&   1\\
   0&   0&   1&   1\\
   0&   1&   1&   0
   \end{smallmatrix}\end{bmatrix},~
\begin{bmatrix}\begin{smallmatrix}
   1&   0&   0&   1\\
   0&   0&   1&   1\\
   0&   1&   1&   0\\
   1&   1&   0&   1
   \end{smallmatrix}\end{bmatrix},~
\begin{bmatrix}\begin{smallmatrix}
   0&   0&   1&   1\\
   0&   1&   1&   0\\
   1&  1&   0&   1\\
   1&   0&   1&   0
   \end{smallmatrix}\end{bmatrix},~
\begin{bmatrix}\begin{smallmatrix}
   0&   1&   1&   0\\
   1&   1&   0&   1\\
   1&   0&   1&   0\\
   0&   1&   0&   1
   \end{smallmatrix}\end{bmatrix},~
\begin{bmatrix}\begin{smallmatrix}
   1&   1&   0&   1\\
   1&   0&   1&   0\\
   0&  1&   0&   1\\
   1&   0&   1&   1
   \end{smallmatrix}\end{bmatrix},~
\begin{bmatrix}\begin{smallmatrix}
   1&   0&   1&   0\\
   0&   1&   0&   1\\
   1&   0&   1&   1\\
   0&   1&   1&   1
   \end{smallmatrix}\end{bmatrix},\\
&
\begin{bmatrix}\begin{smallmatrix}
   0&   1&   0&   1\\
   1&   0&   1&   1\\
   0&   1&   1&   1\\
   1&   1&   1&   1
   \end{smallmatrix}\end{bmatrix},~
\begin{bmatrix}\begin{smallmatrix}
   1&   0&   1&   1\\
   0&   1&   1 &  1\\
   1&   1&   1 &  1\\
   1&   1&   1&   0
   \end{smallmatrix}\end{bmatrix},~
\begin{bmatrix}\begin{smallmatrix}
   0&   1&   1&   1\\
   1&   1&   1&   1\\
   1&   1&   1&   0\\
   1&   1&   0&   0
   \end{smallmatrix}\end{bmatrix},~
\begin{bmatrix}\begin{smallmatrix}
   1&   1&   1&   1\\
   1&   1 & 1&   0\\
   1&   1&   0 &  0\\
   1&   0&   0&   0
   \end{smallmatrix}\end{bmatrix},~
\begin{bmatrix}\begin{smallmatrix}
   1&   1&   1&   0\\
   1 &  1&   0&   0\\
   1&   0&   0&   0\\
   0&   0&   0&   1
   \end{smallmatrix}\end{bmatrix},~
\begin{bmatrix}\begin{smallmatrix}
   1&   1&   0&   0\\
   1&   0&   0&   0\\
   0&   0&   0&   1\\
   0&   0&  1&   0
  \end{smallmatrix} \end{bmatrix},~
\begin{bmatrix}\begin{smallmatrix}
   1&   0&   0&   0\\
   0&   0&   0&   1\\
   0&   0&   1&   0\\
   0&   1&   0&   0
   \end{smallmatrix}\end{bmatrix},~
\begin{bmatrix}\begin{smallmatrix}
   0&   0&   0&   1\\
   0&   0&   1&   0\\
   0&   1&   0&   0\\
   1&   0&   0&   1
   \end{smallmatrix}\end{bmatrix}.
\end{align*}
By Definition \ref{def:Gabidulin codes},  the codewords in the  $(4,1)$ Gabidulin code $\mathcal{V} = \{\mathbf{u}\mathbf{G}: \mathbf{u}\in \mathbb{F}_{2^4}\}$ are
\begin{align*}
\label{eqn:Gabidulin_code_example}
 &[0~ 0~ 0 ~ 0],[\alpha~\alpha^2~\alpha^3~\alpha^4], [\alpha^2~\alpha^3~\alpha^4~\alpha^5],[\alpha^3~\alpha^4~\alpha^5~\alpha^6], 
[\alpha^4~\alpha^5~\alpha^6~\alpha^7],[\alpha^5~\alpha^6~\alpha^7~\alpha^8],[\alpha^6~\alpha^7~\alpha^8~\alpha^9],\\
&[\alpha^7~\alpha^8~\alpha^9~\alpha^{10}],[\alpha^8 ~ \alpha^9~ \alpha^{10} ~ \alpha^{11}],[\alpha^9 ~ \alpha^{10}~ \alpha^{11} ~ \alpha^{12}],
 [\alpha^{10} ~ \alpha^{11}~ \alpha^{12} ~ \alpha^{13}],[\alpha^{11} ~ \alpha^{12}~ \alpha^{13} ~ \alpha^{14}],[\alpha^{12} ~ \alpha^{13}~ \alpha^{14} ~ 1],\\
 &[\alpha^{13} ~ \alpha^{14}~ 1 ~ \alpha],[\alpha^{14}~1~\alpha~\alpha^2],[1 ~\alpha~\alpha^2~\alpha^3].
\end{align*}
With respect to $\mathcal{B}'$, the codewords in the corresponding $(4 \times 4, 2^4, 4)$ Gabidulin MRD code $\mathcal{M}_{\mathcal{B}'}(\mathcal{V})$ as defined in Definition \ref{def:Gabidulin codes} are 
\begin{align*}
&\begin{bmatrix}\begin{smallmatrix}
   0&   0&   0 &0\\
   0&   0&   0 &0\\
   0&   0&   0 &0\\
   0&   0&   0 &0
   \end{smallmatrix}\end{bmatrix},~
\begin{bmatrix}\begin{smallmatrix}
   0&   0&   1&   0\\
   0&   1&   0&   0\\
   1&   0&   0&   1\\
   0&   0&   1&   1
   \end{smallmatrix}\end{bmatrix},~
\begin{bmatrix}\begin{smallmatrix}
   0&   1&  0&   0\\
   1&   0&   0&   1\\
   0&   0&   1&   1\\
   0&   1&   1&   0
   \end{smallmatrix}\end{bmatrix},~
\begin{bmatrix}\begin{smallmatrix}
   1&   0&   0&   1\\
   0&   0&   1&   1\\
   0&   1&   1&   0\\
   1&   1&   0&   1
   \end{smallmatrix}\end{bmatrix},~
\begin{bmatrix}\begin{smallmatrix}
   0&   0&   1&   1\\
   0&   1&   1&   0\\
   1&  1&   0&   1\\
   1&   0&   1&   0
   \end{smallmatrix}\end{bmatrix},~
\begin{bmatrix}\begin{smallmatrix}
   0&   1&   1&   0\\
   1&   1&   0&   1\\
   1&   0&   1&   0\\
   0&   1&   0&   1
   \end{smallmatrix}\end{bmatrix},~
\begin{bmatrix}\begin{smallmatrix}
   1&   1&   0&   1\\
   1&   0&   1&   0\\
   0&  1&   0&   1\\
   1&   0&   1&   1
   \end{smallmatrix}\end{bmatrix},~
\begin{bmatrix}\begin{smallmatrix}
   1&   0&   1&   0\\
   0&   1&   0&   1\\
   1&   0&   1&   1\\
   0&   1&   1&   1
   \end{smallmatrix}\end{bmatrix}, \\
&\begin{bmatrix}\begin{smallmatrix}
   0&   1&   0&   1\\
   1&   0&   1&   1\\
   0&   1&   1&   1\\
   1&   1&   1&   1
   \end{smallmatrix}\end{bmatrix},~
\begin{bmatrix}\begin{smallmatrix}
   1&   0&   1&   1\\
   0&   1&   1 &  1\\
   1&   1&   1 &  1\\
   1&   1&   1&   0
   \end{smallmatrix}\end{bmatrix},~
\begin{bmatrix}\begin{smallmatrix}
   0&   1&   1&   1\\
   1&   1&   1&   1\\
   1&   1&   1&   0\\
   1&   1&   0&   0
   \end{smallmatrix}\end{bmatrix},~
\begin{bmatrix}\begin{smallmatrix}
   1&   1&   1&   1\\
   1&   1 & 1&   0\\
   1&   1&   0 &  0\\
   1&   0&   0&   0
   \end{smallmatrix}\end{bmatrix},~
\begin{bmatrix}\begin{smallmatrix}
   1&   1&   1&   0\\
   1 &  1&   0&   0\\
   1&   0&   0&   0\\
   0&   0&   0&   1
   \end{smallmatrix}\end{bmatrix},~
\begin{bmatrix}\begin{smallmatrix}
   1&   1&   0&   0\\
   1&   0&   0&   0\\
   0&   0&   0&   1\\
   0&   0&  1&   0
  \end{smallmatrix} \end{bmatrix},~
\begin{bmatrix}\begin{smallmatrix}
   1&   0&   0&   0\\
   0&   0&   0&   1\\
   0&   0&   1&   0\\
   0&   1&   0&   0
   \end{smallmatrix}\end{bmatrix},~
\begin{bmatrix}\begin{smallmatrix}
   0&   0&   0&   1\\
   0&   0&   1&   0\\
   0&   1&   0&   0\\
   1&   0&   0&   1
   \end{smallmatrix}\end{bmatrix}.
\end{align*}
One may readily check that the $16$ matrices in the $(4 \times 4, 2^4, 4)$ MRD code $\mathcal{M}$ coincide with the $16$ matrices in the $(4 \times 4, 2^4, 4)$ Gabidulin code $\mathcal{M}_{\mathcal{B}'}(\mathcal{V})$. \hfill $\blacksquare$ 
\end{example}

Proposition \ref{prop:Connections between M1 and M2} explicitly characterizes the intrinsic equivalence between the matrix and vector representations of Gabidulin codes. %
In particular, for any given $(N \times n, q^{Nk}, d)$ MRD code $\mathcal{M} = \{\mathbf{M}(\mathbf{u}): \mathbf{u}\in \mathbb{F}_{q^N}^k\}$ (which is a collection of matrices over $\mathbb{F}_q$) with respect to a defined basis $\mathcal{B}$, 
a corresponding $(n, k)$ Gabidulin code $\mathcal{V} = \{\mathbf{u}\mathbf{G}: \mathbf{u}\in \mathbb{F}_{q^N}^k\}$ over $\mathbb{F}_{q^N}$ (which is a $\mathbb{F}_{q^N}$-vector subspace) can be constructed by utilizing the dual basis of $\mathcal{B}$. %
Furthermore, this inherent connection between the two representations lays a theoretical foundation for demonstrating that the MRD codes explored in this paper are equivalent to a generalization of Gabidulin codes, as to be shown in Section \ref{subsec:Generalization of Gabidulin codes}. 

\section{Circular-shift-based MRD codes}
\label{sec:Circular-shift-based MRD codes}
\subsection{New construction of rank-metric codes}
In this section, assume that $q$ is prime. Let $L$ be a positive integer subject to $\gcd(q, L) = 1$, $m_L$ be the multiplicative order of $q$ modulo $L$, \emph{i.e.,} $q^{m_L}=1 \mod L$ and $n$ be a positive integer no larger than $m_L$. %
Denote by $L_{\mathbb{F}_q}= L \mod q$. %
Define
\begin{equation}
\label{J}
{\cal J}=\{j: 1\leq j \leq L-1, \mathrm{gcd}(j,L)=1\},
\end{equation}
whose cardinality is $J$, that is, $J$ equals to the number of positive integers less than $L$ that are coprime to $L$. %
We shall first describe a new construction of $(J \times n, q^{Jk}, d)$ rank-metric codes, referred to as \emph{circular-shift-based} rank-metric codes, which avoids arithmetic over extension fields. %

For $k \leq n\leq m_L\leq J$, define the following $k\times n$ block matrix $\mathbf{\Psi}_{k\times n}$ with every block to be a matrix in $\mathbb{F}_q ^{L\times L}$
\begin{equation}
\label{eqn:Psi_def}
\mathbf{\Psi}_{k\times n}= \begin{bmatrix}
\mathbf{C}_L^{l_0}&\mathbf{C}_L^{l_1} & \ldots & \mathbf{C}_L^{l_{n-1}}\\
\mathbf{C}_L^{ql_0}&\mathbf{C}_L^{ql_1}& \ldots & \mathbf{C}_L^{ql_{n-1}}\\
 \vdots & \vdots & \vdots & \vdots \\
\mathbf{C}_L^{q^{k-1}l_0}&\mathbf{C}_L^{q^{k-1}l_1}& \ldots & \mathbf{C}_L^{q^{k-1}l_{n-1}}
\end{bmatrix},
\end{equation}
where $0 \leq l_0,l_1,\cdots,l_{n-1} \leq L-1$ are distinct elements, and $\mathbf{C}_L$ denotes the  $L\times L$ cyclic permutation matrix 
$\mathbf{C}_L = \begin{bmatrix} \mathbf{0} & \mathbf{I}_{L-1} \\ 1 & \mathbf{0} \end{bmatrix}$. %
Consider a row vector $\mathbf{m}=[\mathbf{m}_0, \mathbf{m}_1, \ldots, \mathbf{m}_{k-1}] \in \mathbb{F}_q^{Jk}$, where $\mathbf{m}_s \in \mathbb{F}_q^{J}$ for $0 \leq s \leq k-1$. Let 
\begin{equation}
\label{def:C=m_psi}
\mathbf{c}=[
  \mathbf{c}_0 ~ \mathbf{c}_1 ~ \ldots ~ \mathbf{c}_{n-1}]=\mathbf{m}(\mathbf{I}_k\otimes\mathbf{P})\mathbf{\Psi}_{k\times n}(\mathbf{I}_n\otimes\mathbf{Q}) \in \mathbb{F}_q^{Jn},
\end{equation}
where $\mathbf{c}_i \in \mathbb{F}_q^{J}$ for $0 \leq i \leq n-1$, $\mathbf{P} \in \mathbb{F}_q^{J\times L}$ and $\mathbf{Q} \in \mathbb{F}_q^{L\times J}$ are arbitrary matrices, and $\otimes$ denotes the Kronecker product. %
Notice that $\mathbf{c}_i = \sum_{s=0}^{k-1}\mathbf{m}_s\mathbf{P}\mathbf{C}_L^{q^sl_{i}}\mathbf{Q}$ in \eqref{def:C=m_psi}. Let $\Delta$ represent the following 
\begin{equation}
\label{eqn:C_bar}
\Delta: \mathbb{F}_q^{Jn} \rightarrow \mathbb{F}_q^{J \times n},~ \mathrm{with}~
\Delta(\mathbf{c})=[\mathbf{c}_0^\mathrm{T}~ \mathbf{c}_1^\mathrm{T} ~ \ldots ~ \mathbf{c}_{n-1}^\mathrm{T}].
\end{equation}

\begin{definition}
\label{def:circular-shift-MRD-code}
For $k \leq n\leq m_L\leq J$, a $(J \times n, q^{Jk}, d)$ \emph{circular-shift-based rank-metric code} $\mathcal{C}$ is defined as
\begin{equation}
\label{eqn:circular-shift-MRD-code-DEF}
\mathcal{C} = \{\Delta(\mathbf{m}(\mathbf{I}_k\otimes\mathbf{P})\mathbf{\Psi}_{k\times n}(\mathbf{I}_n\otimes\mathbf{Q})): \mathbf{m}\in \mathbb{F}_{q}^{Jk}\}.
\end{equation}
Herein, $\mathbf{\Psi}_{k\times n}$ refers to the $k\times n$ block matrix defined in \eqref{eqn:Psi_def}, and $\mathbf{P} \in \mathbb{F}_q^{J\times L}$, $\mathbf{Q} \in \mathbb{F}_q^{L\times J}$ are arbitrarily given matrices.
\end{definition}

Notice that a circular-shift-based rank-metric code defined above is closed under addition and scalar multiplication by elements in $\mathbb{F}_q$. %
Although this property holds for any $\mathbf{P} \in \mathbb{F}_q^{J\times L}$ and $\mathbf{Q} \in \mathbb{F}_q^{L\times J}$, not all choices of $\mathbf{P}$ and $\mathbf{Q}$ ensure that the 
$(J \times n, q^{Jk}, d)$ circular-shift-based rank-metric code $\mathcal{C}$ in Definition \ref{def:circular-shift-MRD-code} is an MRD code. %
In the following, we provide two proper designs of $\mathbf{P}$ and $\mathbf{Q}$ that give circular-shift-based MRD codes.

\subsection{Explicit construction of $\mathbf{P}$ and $\mathbf{Q}$ for circular-shift-based MRD codes}
\label{subsec:GH}
In this subsection, we introduce a general method to construct matrices $\mathbf{G}_L \in \mathbb{F}_q^{J \times L}$ and $\mathbf{H}_L \in \mathbb{F}_q^{L\times J}$, so that the circular-shift-based rank-metric code $\mathcal{C}$ given in Definition \ref{def:circular-shift-MRD-code} satisfies the MRD property by setting $\mathbf{P}=\mathbf{G}_L$ and $\mathbf{Q}=\mathbf{H}_L$. %
Consider two matrices $\mathbf{U} = [\mathbf{u}_j]_{0\leq j \leq L-1} \in \mathbb{F}_{q^{m_L}}^{J \times L}$ and $\mathbf{U}' = [\mathbf{u}_j']_{0\leq j \leq L-1}\in \mathbb{F}_{q^{m_L}}^{J \times L}$ satisfying the following conditions
\begin{align}
\label{U_def}
\mathrm{rank}([\mathbf{u}_j]_{j\in {\cal J}}) = J, \quad [\mathbf{u}'_j]_{j\in \mathcal{ J}}^{\mathrm{T}} = L_{\mathbb{F}_{q}}^{-1}[\mathbf{u}_j]_{j\in \mathcal{ J}}^{-1}, \quad
[\mathbf{u}_j]_{j\notin \cal J} = &\mathbf{0}~\mathrm{or}~[\mathbf{u}_j']_{j \notin \cal J} = \mathbf{0},
\end{align}
where $L_{\mathbb{F}_q}= L \mod q$. %
Let $\mathbf{V}_L$ and $\widetilde{\mathbf{V}}_L$ respectively denote the following Vandermonde matrices in $\mathbb{F}_{q^{m_L}}^{L \times L}$ with $\beta \in \mathbb{F}_{q^{m_L}}$ being a primitive $L^{th}$ root of unity
\begin{align}
\label{eqn:VL}
\mathbf{V}_L =& \begin{bmatrix}
1 & 1 & 1 & \ldots & 1 \\
1 & \beta & \beta^2 & \ldots & \beta^{L-1} \\
\vdots & \vdots & \vdots & \ldots & \vdots \\
1 & \beta^{L-1} & \beta^{2(L-1)} & \ldots & \beta^{(L-1)(L-1)}
\end{bmatrix}, \\
\label{eqn:VL tilde}
\widetilde{\mathbf{V}}_L =& \begin{bmatrix}
1 & 1 & 1 & \ldots & 1 \\
1 & \beta^{-1} & \beta^{-2} & \ldots & \beta^{-(L-1)} \\
\vdots & \vdots & \vdots & \ldots & \vdots \\
1 & \beta^{-(L-1)} & \beta^{-(L-1)2} & \ldots & \beta^{-(L-1)(L-1)}
\end{bmatrix}.
\end{align}
The desired matrices $\mathbf{G}_L \in \mathbb{F}_q^{J \times L}$ and $\mathbf{H}_L \in \mathbb{F}_q^{L\times J}$ are given by
\begin{equation}
\label{eqn:GH_def}
\mathbf{G}_L = \mathbf{U}\widetilde{\mathbf{V}}_L, \quad \mathbf{H}_L=\mathbf{V}_L\mathbf{U}'^{\mathrm{T}}.
\end{equation}
The conditions imposed on the matrices $\mathbf{U}$ and $\mathbf{U}'$ not only ensure the MRD property of the circular-shift-based rank-metric codes defined in this section, but also guarantee that $(\mathbf{G}_L\mathbf{C}_L^i\mathbf{H}_L)(\mathbf{G}_L\mathbf{C}_L^j\mathbf{H}_L)=\mathbf{G}_L\mathbf{C}_L^{i+j}\mathbf{H}_L$ for $0 \leq i, j \leq L-1$ (See, e.g., \cite{Jin-arXiv}), providing a theoretical foundation for the subsequent discussion on the differences between the proposed circular-shift-based MRD codes and Gabidulin codes.

Notice that the definitions above do not guarantee that the matrices $\mathbf{G}_L$ and $\mathbf{H}_L$ are necessarily over $\mathbb{F}_{q}$. %
However, certain techniques (See, e.g., \cite{Jin-arXiv}) for designing $\mathbf{U}$ and $\mathbf{U}'$ ensure that $\mathbf{G}_L$ and $\mathbf{H}_L$ are indeed defined over $\mathbb{F}_{q}$. %
For instance, let $\bar{\mathbf{V}}_L$ denote the $J\times L$ full rank matrix obtained by deleting the last $L-J$ rows of $\mathbf{V}_L$. %
Let $\mathbf{U}= [\mathbf{u}_j]_{0 \leq j \leq L-1}=L_{\mathbb{F}_q}^{-1}\bar{\mathbf{V}}_L$, $[\mathbf{u}'_j]_{j\in \mathcal{J}} = L_{\mathbb{F}_{q}}^{-1}([\mathbf{u}_j]^{-1}_{j\in \mathcal{J}})^{\mathrm{T}}$ and $[\mathbf{u}'_j]_{j\notin \mathcal{J}} = \mathbf{0}$. %
Based on \eqref{eqn:GH_def} and $\mathbf{V}_L\widetilde{\mathbf{V}}_L = \widetilde{\mathbf{V}}_L\mathbf{V}_L = L_{\mathbb{F}_q}\mathbf{I}_L$ (See, e.g., \cite{tang2020circular}), we have 
\begin{equation}
\label{eqn:instance 1}
\mathbf{G}_L = \mathbf{U}\widetilde{\mathbf{V}}_L=[\mathbf{I}_J~\mathbf{0}], \quad \mathbf{H}_L = \mathbf{V}_L\mathbf{U}'^{\mathrm{T}}=[\mathbf{I}_J~\mathbf{A}]^{\mathrm{T}},
\end{equation}
where $\mathbf{G}_L$ is over $\mathbb{F}_{q}$ and $\mathbf{A}$ is a $J\times(L-J)$ matrix. %
Proposition $2$ in \cite{Jin-arXiv} further guarantees that $\mathbf{H}_L$ is also over $\mathbb{F}_{q}$. %
As an illustration, let $q = 2$, $L = 9$ and $\mathcal{J} = \{1, 2, 4, 5, 7 ,8\}$. %
One can readily obtain that $\mathbf{G}_9 = [\mathbf{I}_6~\mathbf{0}]$ and $\mathbf{H}_9 = [\mathbf{I}_6~\mathbf{A}]^{\mathrm{T}}$  with $\mathbf{A} = [\mathbf{I}_3~\mathbf{I}_3]^{\mathrm{T}}$. %
As another instance, let $\widetilde{\mathbf{V}}_L'$  denote the $J\times L$ full rank matrix obtained by deleting the first $L-J$ rows of $\widetilde{\mathbf{V}}_L$.  Let $\mathbf{U}'= [\mathbf{u}'_j]_{0 \leq j \leq L-1}=L_{\mathbb{F}_q}^{-1}\widetilde{\mathbf{V}}_L'$, $[\mathbf{u}_j]_{j\in \mathcal{J}} = L_{\mathbb{F}_{q}}^{-1}([\mathbf{u}'_j]^{-1}_{j\in \mathcal{J}})^{\mathrm{T}}$ and $[\mathbf{u}_j]_{j\notin \mathcal{J}} = \mathbf{0}$ . We have 
\begin{equation}
\mathbf{G}_L = \mathbf{U}\widetilde{\mathbf{V}}_L=[\mathbf{A}~\mathbf{I}_J], \quad \mathbf{H}_L = \mathbf{V}_L\mathbf{U}'^{\mathrm{T}}=[\mathbf{0}~\mathbf{I}_J]^{\mathrm{T}},
\end{equation}
where $\mathbf{A} \in \mathbb{F}_{q}^{J\times(L-J)}$ is the  matrix  same as the one in \eqref{eqn:instance 1}. As an illustration, let $q = 2$, $L = 9$ and $\mathcal{J}= \{1, 2, 4, 5, 7, 8\}$. One can readily obtain that $\mathbf{G}_9 = [\mathbf{A}~\mathbf{I}_6]$ and $\mathbf{H}_9 = [\mathbf{0}~\mathbf{I}_6]^{\mathrm{T}}$, with $\mathbf{A} = [\mathbf{I}_3~\mathbf{I}_3]^{\mathrm{T}}$. %
Now assume that $\mathbf{G}_L$ and $\mathbf{H}_L$ are constructed subject to \eqref{eqn:GH_def} and  \eqref{U_def}, and they are both matrices over $\mathbb{F}_{q}$. %

\begin{theorem}
\label{the:Circular_MRD_GH}
Let $\mathbf{P}=\mathbf{G}_L$ and $\mathbf{Q}=\mathbf{H}_L$. For $n \leq m_L\leq J$, if  
\begin{equation}
\label{eqn:independent C_GH}
\mathrm{rank}(\sum\nolimits_{i = 0}^{n-1} a_i\mathbf{P}\mathbf{C}_L^{l_i}\mathbf{Q}) =J, ~\forall [a_0, \ldots, a_{n-1}] \in \mathbb{F}_q^n\backslash \{\mathbf{0}\},
\end{equation}
then the $(J \times n, q^{Jk}, d)$ circular-shift-based rank-metric code $\mathcal{C}$ in Definition \ref{def:circular-shift-MRD-code} is an MRD code, that is, $d = n - k + 1$. %
In particular, for every constant $0 \leq c \leq L-1$, $l_j = j+c \mod L$, $0 \leq j \leq n-1$, is a valid selection to satisfy \eqref{eqn:independent C_GH}. %
\begin{IEEEproof}
We first present two preliminary statements before proving the MRD property. %
\begin{enumerate}
  \item For $0 \leq l_0, l_1, \ldots, l_{n-1} \leq L-1$, \eqref{eqn:independent C_GH}  holds if and only if
\begin{equation}
\label{eqn:independent b}
\sum\nolimits_{i=0}^{n-1} a_i\beta^{jl_i} \neq 0 \quad \forall j \in \mathcal{J},~[a_0, \ldots, a_{n-1}] \in \mathbb{F}_q^n \backslash \{\mathbf{0}\},
\end{equation}
that is, $\beta^{jl_0}, \beta^{jl_1}, \ldots, \beta^{jl_{n-1}}$ are $\mathbb{F}_q$-linearly independent elements in $\mathbb{F}_{q^{m_L}}$ for all $j \in \mathcal{J}$. %
  \item For all $0\leq l \leq L-1$ and $j \in \mathcal{J}$, we have
\begin{align}
\label{eqn:GH_mu_C}
\beta^{jl}\mathbf{u}_j =\mathbf{G}_L\mathbf{C}_L^l\mathbf{H}_L\mathbf{u}_j,
\end{align}
where $\mathbf{u}_j$ denotes the $(j+1)^{st}$ column of $\mathbf{U}$ defined in \eqref{U_def}. %
\end{enumerate}

According to \cite{Jin-arXiv} and \cite{tang2020circular}, 
\begin{equation}
\label{eqn:VL VL'}
\mathbf{V}_L\widetilde{\mathbf{V}}_L = \widetilde{\mathbf{V}}_L\mathbf{V}_L = L_{\mathbb{F}_q}\mathbf{I}_L,\quad
\mathbf{C}_L^l=L_{\mathbb{F}_q}^{-1}\mathbf{V}_L\mathrm{Diag}([\beta^{lj}]_{0 \leq j \leq L-1})\widetilde{\mathbf{V}}_L, \quad  \forall 0\leq l \leq L-1.
\end{equation}
Based on \eqref{U_def}, \eqref{eqn:GH_def} and \eqref{eqn:VL VL'}, $\mathrm{rank}(\sum\nolimits_{i = 0}^{n-1} a_i\mathbf{G}_L\mathbf{C}_L^{l_i}\mathbf{H}_L)$ can be written as
\begin{equation*}
\begin{split}
&\mathrm{rank}(\sum\nolimits_{i = 0}^{n-1} a_i\mathbf{G}_L\mathbf{C}_L^{l_i}\mathbf{H}_L)
=\mathrm{rank}(\sum\nolimits_{i = 0}^{n-1} a_i\mathbf{U}\mathrm{Diag}([\beta^{l_ij}]_{0 \leq j \leq L-1})\mathbf{U}'^{\mathrm{T}})\\
=&\mathrm{rank}(\sum\nolimits_{i = 0}^{n-1} a_i[\mathbf{u}_j]_{j\in {\cal J}}\mathrm{Diag}([\beta^{l_ij}]_{j\in \mathcal{J}})[\mathbf{u}'_j]^\mathrm{T}_{j\in {\cal J}}).
\end{split}
\end{equation*}
The full rank of $[\mathbf{u}_j]_{j\in {\cal J}}$ and $[\mathbf{u}'_j]^\mathrm{T}_{j\in {\cal J}}$ implies
\begin{equation}
\mathrm{rank}(\sum\nolimits_{i = 0}^{n-1} a_i\mathbf{G}_L\mathbf{C}_L^{l_i}\mathbf{H}_L)
=\mathrm{rank}(\mathrm{Diag}([f(\beta^{j})]_{j\in \mathcal{J}})),
\end{equation}
where $f(x)=\sum_{i = 0}^{n-1} a_ix^{l_i}$ with $a_i \in \mathbb{F}_{q}$. %
Consequently, \eqref{eqn:independent C_GH} holds if and only if \eqref{eqn:independent b} holds. %
This statement also provides a clear explanation for why the parameter $n$ must be restricted to be no larger than $m_L$. %

The right-hand side of \eqref{eqn:GH_mu_C} can be written as
\begin{align*}
\mathbf{G}_L\mathbf{C}_L^l\mathbf{H}_L\mathbf{u}_j=L_{\mathbb{F}_q}\mathbf{U}\mathrm{Diag}([\beta^{lj}]_{0 \leq j \leq L-1})\mathbf{U}'^{\mathrm{T}}\mathbf{u}_j=L_{\mathbb{F}_q}\sum\nolimits_{i\in \mathcal{J}}\beta^{il}\mathbf{u}_i{\mathbf{u}'}_i^\mathrm{T}\mathbf{u}_j=\mathbf{u}_j\beta^{jl}, 
\quad \forall 0\leq l \leq L-1, ~j \in \cal{J},
\end{align*}
where the second equality holds because for $j_1, j_2 \in \mathcal{J}$, ${\mathbf{u}'}_{j_1}^\mathrm{T}\mathbf{u}_{j_2}=L_{\mathbb{F}_{q}}^{-1}$ when $j_1=j_2$, and ${\mathbf{u}'}_{j_1}^\mathrm{T}\mathbf{u}_{j_2}=0$ otherwise. %
Therefore, \eqref{eqn:GH_mu_C} is proved.

Because $\mathcal{C}$ is closed under addition and scalar multiplication by elements in $\mathbb{F}_q$, it suffices to prove that every nonzero codeword $\Delta(\mathbf{m}(\mathbf{I}_k\otimes\mathbf{G}_L)\mathbf{\Psi}_{k\times n}(\mathbf{I}_n\otimes\mathbf{H}_L))$ has rank at least $n-k+1$. %
Assume $\mathrm{rank}(\Delta(\mathbf{m}(\mathbf{I}_k\otimes\mathbf{G}_L)\mathbf{\Psi}_{k\times n}(\mathbf{I}_n\otimes\mathbf{H}_L))=r$. %
There are $n - r$ $n$-dimensional linearly independent column vectors over $\mathbb{F}_q$, say, $\mathbf{l}_1, \ldots, \mathbf{l}_{n-r}$ such that $\Delta(\mathbf{m}(\mathbf{I}_k\otimes\mathbf{G}_L)\mathbf{\Psi}_{k\times n}(\mathbf{I}_n\otimes\mathbf{H}_L)\mathbf{l}_j = \mathbf{0}$ for all $1 \leq j \leq n - r$. %
Write $\mathbf{l}_{j}$ as $[l_{j_0}~l_{j_1}~\ldots~l_{j_{n-1}}]^\mathrm{T}$ with $l_{js} \in \mathbb{F}_q$, which further implies
\begin{equation}
\label{eqn:zi1}
\sum\nolimits_{i=0}^{n-1}\mathbf{c}_i^\mathrm{T}l_{ji}= \mathbf{0} \quad  \forall 1 \leq j \leq n - r,
\end{equation}
where $\mathbf{c}_i=\sum_{s=0}^{k-1}\mathbf{m}_s\mathbf{G}_L\mathbf{C}_L^{q^sl_{i}}\mathbf{H}_L$. %
Let $\mathbf{t}_j$ denote the following row vector in $\mathbb{F}_{q^{m_L}}^k$
\begin{equation}
\label{eqn:t_j}
\mathbf{t}_j=[t_{0,j}~t_{1,j}~\ldots~t_{k-1,j}], \quad j \in \mathcal{J},
\end{equation}
where $t_{s,j}=\mathbf{m}_s\mathbf{u}_{j}$ for $0 \leq s \leq k-1$. The $q$-linearized polynomial $L_{\mathbf{t}_j}(x)$ over $\mathbb{F}_{q^{m_L}}$ is defined as 
\begin{equation}
L_{\mathbf{t}_j}(x)=\sum\nolimits_{s=0}^{k-1} t_{s,j}x^{q^s}, \quad j \in \mathcal{J}.
\end{equation}
For any nonzero row vector $\mathbf{m}=[\mathbf{m}_0~\mathbf{m}_1~\ldots~\mathbf{m}_{k-1}]$, assume that $\mathbf{t}_j=\mathbf{0}$ for all $j \in \mathcal{J}$. This assumption implies that there exists a nonzero row vector $\mathbf{m}_s$, $0 \leq s \leq k-1$, in $\mathbb{F}_q^J$ such that $\mathbf{m}_s[\mathbf{u}_j]_{j \in \mathcal{J}}=\mathbf{0}$. The full rank of the $J\times J$ matrix $[\mathbf{u}_j]_{j\in {\cal J}}$ contradicts the existence of such a nonzero row vector $\mathbf{m}_s$. Consequently, there exists an index $j'\in \mathcal{J}$ subject to $\mathbf{t}_{j'} \neq \mathbf{0}$.

Using \eqref{eqn:GH_mu_C}, Eq. \eqref{eqn:zi1} implies
\begin{equation}
\begin{split}
&\mathbf{u}_{j'}^\mathrm{T}\sum\nolimits_{i=0}^{n-1}\mathbf{c}_i^\mathrm{T}l_{ji}=
\sum\nolimits_{i=0}^{n-1}l_{ji} \left(\sum\nolimits_{s=0}^{k-1}\mathbf{m}_s\mathbf{u}_{j'} \beta^{j'l_{i}q^s} \right)
=\sum\nolimits_{i=0}^{n-1}l_{ji} \left(\sum\nolimits_{s=0}^{k-1}t_{s,j'}\beta^{j'l_{i}q^s}\right)\\
=&\sum\nolimits_{i=0}^{n-1}l_{ji}L_{\mathbf{t}_{j'}}(\beta^{j'l_i})= 0, \quad j'\in \mathcal{J}, ~ \forall 1 \leq j \leq n - r.
\end{split}
\end{equation}
Furthermore, by making use of Lemma \ref{lemma:q-polynomial}, the nonzero element $\sum\nolimits_{i=0}^{n-1} l_{ji}\beta^{j'l_i}$ is a root of $L_{\mathbf{t}_{j'}}(x)$. %
Note that any $\mathbb{F}_q$-linear combination of $\sum\nolimits_{i=0}^{n-1} l_{ji}\beta^{j'l_i}$ for all $1 \leq j\leq n-r$ is also a root of the $q$-linearized polynomial  $L_{\mathbf{t}_{j'}}(x)$. On the other hand, since $L_{\mathbf{t}_{j'}}(x)$ has degree at most $q^{k-1}$, it has at most $q^{k-1}$ roots in $\mathbb{F}_{q^{m_L}}$. This implies that $n - r \leq k - 1$ or equivalently $r \geq n - k + 1 = d$.

It remains to prove that for every constant $0 \leq c \leq L-1$, $l_j = j+c \mod L$ is a valid selection to satisfy \eqref{eqn:independent C_GH}. %
As discussed in the first paragraph of this proof, it is equivalent to show that $\beta^{cj}, \beta^{(1+c)j}, \ldots, \beta^{(n-1+c)j}$ are $\mathbb{F}_q$-linearly independent elements in $\mathbb{F}_{q^{m_L}}$ for all $j \in \mathcal{J}$. %
According to Lemma 3.51 in \cite{Lidl:Finite_Field}, this holds if and only if 
\begin{equation}
\label{eqn:independent i}
\det \left (
\begin{bmatrix}
\beta^{cj} &\beta^{cjq}&\ldots&\beta^{cjq^{n-1}}\\
\beta^{(1+c)j} &\beta^{(1+c)jq}&\ldots&\beta^{(1+c)jq^{n-1}}\\
\vdots&\vdots&\ldots&\vdots\\
\beta^{(n-1+c)j} &\beta^{(n-1+c)jq}&\ldots&\beta^{(n-1+c)jq^{n-1}}\\
\end{bmatrix}
\right)\neq 0, \quad \forall j \in \mathcal{J}.
\end{equation}
Observe that
\begin{equation*}
\begin{bmatrix}
\beta^{cj} &\beta^{cjq}&\ldots&\beta^{cjq^{n-1}}\\
\beta^{(1+c)j} &\beta^{(1+c)jq}&\ldots&\beta^{(1+c)jq^{n-1}}\\
\vdots&\vdots&\ldots&\vdots\\
\beta^{(n-1+c)j} &\beta^{(n-1+c)jq}&\ldots&\beta^{(n-1+c)jq^{n-1}}\\
\end{bmatrix}
=
\begin{bmatrix}
1 &1&\ldots&1\\
\beta^j &\beta^{jq}&\ldots&\beta^{jq^{n-1}}\\
\vdots&\vdots&\ldots&\vdots\\
\beta^{(n-1)j} &\beta^{(n-1)jq}&\ldots&\beta^{(n-1)jq^{n-1}}\\
\end{bmatrix}
\mathrm{Diag}([\beta^{cjq^i}]_{0 \leq i \leq n-1}).
\end{equation*}
Consequently,
\begin{equation*}
\det \left (
\begin{bmatrix}\begin{smallmatrix}
\beta^{cj} &\beta^{cjq}&\ldots&\beta^{cjq^{n-1}}\\
\beta^{(1+c)j} &\beta^{(1+c)jq}&\ldots&\beta^{(1+c)jq^{n-1}}\\
\vdots&\vdots&\ldots&\vdots\\
\beta^{(n-1+c)j} &\beta^{(n-1+c)jq}&\ldots&\beta^{(n-1+c)jq^{n-1}}\\
\end{smallmatrix}\end{bmatrix}
\right)=\left (\prod\nolimits_{0 \leq i_1<i_2 \leq n-1}(\beta^{jq^{i_2}}-\beta^{jq^{i_1}})\right)\left (\prod\nolimits_{0 \leq i \leq n-1}\beta^{cjq^i}\right),  \forall j \in \mathcal{J}.
\end{equation*}
Since $n \leq m_L$, the set $\{jq^i ~\mathrm{mod}~L: 0\leq i \leq n-1\}$ is a subset of $\mathcal{J}$ defined in \eqref{J}, which implies that for all $j\in \mathcal{J}$, the elements $\beta^{cj}, \beta^{cjq}, \ldots, \beta^{cjq^{n-1}}$ are nonzero, and the elements $\beta^{j}, \beta^{jq}, \ldots, \beta^{jq^{n-1}}$ are distinct. %
Consequently, \eqref{eqn:independent i} holds. 
\end{IEEEproof}
\end{theorem}

\begin{remark}
Every codeword of a $(J \times n, q^{Jk}, d)$ Gabidulin code is a matrix in $\mathbb{F}_q^{J\times n}$ with $n\leq J$. In comparison, every codeword of a $(J \times n, q^{Jk}, d)$ circular-shift-based MRD code $\mathcal{C}$ is a matrix in $\mathbb{F}_q^{J\times n}$ with $n \leq m_L$. In general, $m_L \leq J$, so we cannot set $m_L < n \leq J$ in constructing $\mathcal{C}$.  %
\end{remark}

\begin{example}
\label{example 2}
Consider $q=2$ and $L=7$. Set
\begin{equation*}
\mathbf{G}_7=[\mathbf{I}_6~\mathbf{1}], \mathbf{H}_7=[\mathbf{I}_6~\mathbf{0}]^\mathrm{T},
\mathbf{\Psi}_{1\times 3} = 
\begin{bmatrix}
\mathbf{I}_7&\mathbf{C}_7 & \mathbf{C}_7^2
\end{bmatrix},
\end{equation*}
so that \eqref{eqn:independent C_GH} holds. %
Sequentially set $\mathbf{m}$ as a $\mathbb{F}_2$-linear combination of the following $6$ row vectors in $\mathbb{F}_2^6$
\begin{equation}
\label{def:Example 2 m}
\begin{split}
\{&[0~0~0~0~0~1], [0~0~0~0~1~0], [0~0~0~1~0~0], \\
&[0~0~1~0~0~0],[0~1~0~0~0~0], [1~0~0~0~0~0]\}.
\end{split}
\end{equation} 
Then, the codewords in the $(6 \times 3, 2^6, 3)$ circular-shift-based MRD code $\mathcal{C} = \{\Delta(\mathbf{m}(\mathbf{I}_1\otimes\mathbf{G}_7)\mathbf{\Psi}_{1\times 3}(\mathbf{I}_3\otimes\mathbf{H}_7), \mathbf{m} \in \mathbb{F}_2^6\}$ can be written as a $\mathbb{F}_2$-linear combination of the following $6$ matrices in $\mathbb{F}_2^{6 \times 3}$
\begin{equation}
\label{eqn:example GH 3}
\begin{bmatrix}\begin{smallmatrix}
   0 &  1 &  1\\
   0 &  0 &  1\\
   0 &  0 &  0\\
   0 &  0  & 0\\
   0 &  0  & 0\\
   1  & 0  & 0
   \end{smallmatrix}\end{bmatrix},
\begin{bmatrix}\begin{smallmatrix}
   0  & 1 &  0\\
   0 &  0  & 1\\
   0  & 0 &  0\\
   0  & 0 &  0\\
   1 &  0 &  0\\
   0  & 1 &  0
   \end{smallmatrix}\end{bmatrix},
\begin{bmatrix}\begin{smallmatrix}
   0 &  1 &  0\\
   0  & 0 &  1\\
   0 &  0 &  0\\
   1  & 0 &  0\\
   0 &  1 &  0\\
   0  & 0 &  1
   \end{smallmatrix}\end{bmatrix},
\begin{bmatrix}\begin{smallmatrix}
   0 &  1 &  0\\
   0 &  0  & 1\\
   1 &  0 &  0\\
   0  & 1 &  0\\
   0 &  0  & 1\\
   0  & 0  & 0
   \end{smallmatrix}\end{bmatrix},
\begin{bmatrix}\begin{smallmatrix}
   0 &  1  & 0\\
   1  & 0 &  1\\
   0  & 1 &  0\\
   0 &  0 &  1\\
   0  & 0  & 0\\
   0  & 0  & 0
   \end{smallmatrix}\end{bmatrix},
\begin{bmatrix}\begin{smallmatrix}
   1 &  1 &  0\\
   0 &  1  & 1\\
   0  & 0 &  1\\
   0  & 0 &  0\\
   0  & 0 &  0\\
   0  & 0  & 0
   \end{smallmatrix}\end{bmatrix}.
\end{equation}
One may check (by computer enumeration) that each of the nonzero codewords in $\mathcal{C}$ has full rank $3$, thus satisfying the MRD property. 
\hspace*{\fill} $\blacksquare$
\end{example}

Notice that constructing $(J \times n, q^{Jk}, d)$ Gabidulin codes with $n \leq J$ over $\mathbb{F}_q$ relies on arithmetic over the extension field $\mathbb{F}_{q^J}$ of $\mathbb{F}_q$. %
As $n$ increases, the extension field $\mathbb{F}_{q^J}$ required for constructing Gabidulin codes becomes significantly larger, resulting in much more complicated implementation. %
In contrast, for circular-shift-based MRD codes introduced in this paper, generating a codeword requires only entry-wise additions and multiplications over $\mathbb{F}_q$, along with column-wise circular-shift operations. %
By avoiding arithmetic over the extension field $\mathbb{F}_{q^{J}}$, circular-shift-based MRD codes enable more  parameter selection, that is, $J$ and $k$ can be selected as large as desired. %
Furthermore, compared to the vector representation of Gabidulin codes $\mathcal{V} = \{\mathbf{u}\mathbf{G}: \mathbf{u}\in \mathbb{F}_{q^J}^k\}$ reviewed in Section \ref{subsec:Gabidulin codes}, the selection of a proper basis of $\mathbb{F}_{q^J}$ over $\mathbb{F}_q$ to transform a codeword in $\mathbb{F}_{q^J}^n$ into a matrix codeword in $\mathbb{F}_{q}^{J\times n}$ is no longer required. 

In coding theory, due to the low computational complexity, circular-shift operations have been widely adopted in the design of various codes, such as circular-shift linear network codes \cite{Tang_LNC_TIT}-\cite{Jin_TIT}, \cite{Su_20_TCOM}, quasi-cyclic codes \cite{1988_TIT}-\cite{Huang quasi-cyclic}, array codes \cite{Blaum-Evenodd-ToC95}-\cite{Zhai_TCom} and regenerating codes \cite{Hou_Basic}.  %
Despite this, to the best of our knowledge, circular-shift operations have not been investigated to construct MRD codes.

\subsection{An alternative construction of $\mathbf{P}$ and $\mathbf{Q}$}
\label{subsec:GG}
In this subsection, we introduce an alternative way to design the matrices $\mathbf{P}$ and $\mathbf{Q}$ over $\mathbb{F}_q$ based on the matrix $\mathbf{G}_L$ defined in \eqref{eqn:GH_def}. %
This design ensures that the constructed circular-shift-based rank-metric code $\mathcal{C}$ given in Definition \ref{def:circular-shift-MRD-code} not only satisfies the MRD property, but also achieves lower encoding complexity under certain choices of $\mathbf{G}_L$. %
Let $\tau(x)$ denote the following polynomial of degree $L - J$
\begin{equation}
\label{def: tau(x)}
\tau(x)=\prod\nolimits_{j \notin \cal J}(x-\beta^j),
\end{equation}
where $\beta \in \mathbb{F}_{q^{m_L}}$ represents a primitive $L^{th}$ root of unity and $\mathcal{J}$ is defined in \eqref{J}. %
Since the polynomial $\prod\nolimits_{j \in \cal J}(x-\beta^j)$ is known as the $L^{th}$ cyclotomic polynomial over $\mathbb{F}_q$ (See e.g. \cite{Lidl:Finite_Field}) and $\prod\nolimits_{0 \leq j \leq L-1}(x-\beta^j)=x^L-1$, it follows that $\tau(x)$ is also a polynomial over $\mathbb{F}_q$. %
$\tau(\mathbf{C}_L)$ represents the $L\times L$ cyclic permutation matrix obtained by evaluation of $\tau(x)$ under the setting $x = \mathbf{C}_L$. %
Specifically, when $q=2$ and $L$ is a prime, $\tau(x) = 1 + x$ and $\tau(\mathbf{C}_L) = \mathbf{I}_L+\mathbf{C}_L$. %
Assume that the matrix $\mathbf{U} = [\mathbf{u}_j]_{0\leq j \leq L-1} \in \mathbb{F}_{q^{m_L}}^{J\times L}$ is constructed according to \eqref{U_def}, and that the $J \times L$ matrix $\mathbf{G}_L$ and the $L \times J$ matrix $\mathbf{H}_L$ are constructed according to \eqref{eqn:GH_def} and defined over $\mathbb{F}_{q}$. %
We next establish a connection between $\mathbf{G}_L\mathbf{C}_L^l\mathbf{H}_L$ and $\mathbf{G}_L\mathbf{C}_L^l\tau(\mathbf{C}_L)\mathbf{G}_L^\mathrm{T}$, $0 \leq l\leq L-1$, based on the following $J\times J$ matrix $\mathbf{T}$
\begin{equation}
\label{def:T}
\mathbf{T}=L_{\mathbb{F}_{q}}[\mathbf{u}_{L-j}]_{j \in \mathcal{J}} \mathrm{Diag}([\tau(\beta^j)]_{j\in \mathcal{J}}) [\mathbf{u}_{j}]_{j \in \mathcal{J}}^\mathrm{T},
\end{equation}
where $\mathrm{Diag}([\tau(\beta^j)]_{j\in \mathcal{J}})$ refers to the $J \times J$ diagonal matrix with diagonal entries respectively equal to $\tau(\beta^j)$, $j\in \mathcal{J}$.

\begin{lemma}
\label{lemma:C2=TC1}
The $J \times J$ matrix $\mathbf{T}$ defined in \eqref{def:T} is full rank. For every $0 \leq l \leq L-1$, we have
\begin{equation}
\mathbf{G}_L\mathbf{C}_L^l\tau(\mathbf{C}_L)\mathbf{G}_L^\mathrm{T} = \mathbf{G}_L\mathbf{C}_L^l\mathbf{H}_L\mathbf{T}^\mathrm{T}.
\end{equation}
\begin{IEEEproof}
Please refer to Appendix-\ref{appendix:C2=TC1}.
\end{IEEEproof}
\end{lemma} 

Let $\mathcal{C}_1$ and $\mathcal{C}_2$ denote the $(J \times n, q^{Jk}, d)$ circular-shift-based rank-metric code $\mathcal{C}$ in Definition \ref{def:circular-shift-MRD-code} under the respective setting of $\mathbf{P}=\mathbf{G}_L$, $\mathbf{Q}=\mathbf{H}_L$ and of $\mathbf{P}=\mathbf{G}_L$, $\mathbf{Q}=\tau(\mathbf{C}_L)\mathbf{G}_L^\mathrm{T}$. Lemma \ref{lemma:C2=TC1} proves the full rank of $\mathbf{T}$ and implies that $\mathcal{C}_2 = \mathbf{T}\mathcal{C}_1$. Consequently, the MRD property of $\mathcal{C}_1$ established in Theorem \ref{the:Circular_MRD_GH} directly implies the MRD property of $\mathcal{C}_2$ summarized in the next theorem. 

\begin{theorem}
\label{the:Circular_MRD_GG}
Let $\mathbf{P}=\mathbf{G}_L$ and $\mathbf{Q}=\tau(\mathbf{C}_L)\mathbf{G}_L^\mathrm{T}$. For $n \leq m_L\leq J$, if  
\begin{equation}
\label{eqn:independent C_GG}
\mathrm{rank}(\sum\nolimits_{i = 0}^{n-1} a_i\mathbf{P}\mathbf{C}_L^{l_i}\mathbf{Q}) =J, ~\forall [a_0, \ldots, a_{n-1}] \in \mathbb{F}_q^n\backslash \{\mathbf{0}\},
\end{equation}
then the $(J \times n, q^{Jk}, d)$ circular-shift-based rank-metric code $\mathcal{C}$ in Definition \ref{def:circular-shift-MRD-code} is an MRD code, that is, $d = n - k + 1$. %
In particular, for every constant $0 \leq c \leq L-1$, $l_j = j+c \mod L$, $0 \leq j \leq n-1$, is a valid selection to satisfy \eqref{eqn:independent C_GG}. %
\end{theorem}

\begin{example}
\label{example 3}
Consider $q=2$ and $L=7$, so that $\tau(\mathbf{C}_{7})=\mathbf{I}_{7}+\mathbf{C}_{7}$. Set
\begin{equation*}
\mathbf{G}_7=[\mathbf{I}_6~\mathbf{0}], \quad
\mathbf{\Psi}_{1\times 3} = 
[\begin{matrix}
\mathbf{I}_7&\mathbf{C}_7 & \mathbf{C}_7^2
\end{matrix}],
\end{equation*}
so that \eqref{eqn:independent C_GG} holds. %
Sequentially set $\mathbf{m}$ as a $\mathbb{F}_2$-linear combination of the $6$ row vectors in $\mathbb{F}_2^6$ in \eqref{def:Example 2 m}. %
Then, the codewords in the $(6 \times 3, 2^6, 3)$ circular-shift-based MRD code $\mathcal{C} = \{\Delta(\mathbf{m}(\mathbf{I}_1\otimes\mathbf{G}_7)\mathbf{\Psi}_{1\times 3}(\mathbf{I}_3\otimes(\tau(\mathbf{C}_7)\mathbf{G}_7^\mathrm{T})), \mathbf{m} \in \mathbb{F}_2^6\}$ can be written as a $\mathbb{F}_2$-linear combination of the following $6$ matrices in $\mathbb{F}_2^{6 \times 3}$
\begin{equation}
\label{eqn:example GG 5}
\begin{bmatrix}\begin{smallmatrix}
   0 &  1 &  1\\
   0  & 0 &  1\\
   0 &  0 &  0\\
   0 &  0 &  0\\
   0 &  0 &  0\\
   1 &  0 &  0
   \end{smallmatrix}\end{bmatrix},
\begin{bmatrix}\begin{smallmatrix}
   0 &  0 &  1\\
   0 &  0  & 0\\
   0 &  0 &  0\\
   0 &  0 &  0\\
   1 &  0 &  0\\
   1 &  1 &  0
   \end{smallmatrix}\end{bmatrix},
\begin{bmatrix}\begin{smallmatrix}
   0  & 0 &  0\\
   0 &  0 &  0\\
   0 &  0  & 0\\
   1 &  0 &  0\\
   1 &  1 &  0\\
   0 &  1 &  1
   \end{smallmatrix}\end{bmatrix},
\begin{bmatrix}\begin{smallmatrix}
   0 &  0 &  0\\
   0 &  0 &  0\\
   1 &  0 &  0\\
   1 &  1 &  0\\
   0 &  1 &  1\\
   0 &  0 &  1
   \end{smallmatrix}\end{bmatrix},
\begin{bmatrix}\begin{smallmatrix}
   0  & 0 &  0\\
   1 &  0 &  0\\
   1 &  1 &  0\\
   0 &  1 &  1\\
   0  & 0 &  1\\
   0 &  0 &  0
   \end{smallmatrix}\end{bmatrix},
\begin{bmatrix}\begin{smallmatrix}
   1 &  0  & 0\\
   1  & 1 &  0\\
   0 &  1 &  1\\
   0 &  0 &  1\\
   0 &  0 &  0\\
   0 &  0 &  0
   \end{smallmatrix}\end{bmatrix}.
\end{equation}
One may check (by computer enumeration) that each of the nonzero codewords in $\mathcal{C}$ has full rank $3$, thus satisfying the MRD property.  
\hspace*{\fill} $\blacksquare$
\end{example}

It is worth noting that under particular selection of $\mathbf{G}_L$, the encoding complexity of the code $\mathcal{C}_2 = \{\Delta(\mathbf{m}(\mathbf{I}_k\otimes\mathbf{G}_L)\mathbf{\Psi}_{k\times n}(\mathbf{I}_n\otimes(\tau(\mathbf{C}_L)\mathbf{G}_L^\mathrm{T}))): \mathbf{m}\in \mathbb{F}_{q}^{Jk}\}$ is lower than that of the code $\mathcal{C}_1 = \{\Delta(\mathbf{m}(\mathbf{I}_k\otimes\mathbf{G}_L)\mathbf{\Psi}_{k\times n}(\mathbf{I}_n\otimes\mathbf{H}_L)): \mathbf{m}\in \mathbb{F}_{q}^{Jk}\}$. %
This is because $\mathbf{G}_L$ can be chosen as $\mathbf{G}_L=[\mathbf{I}_{J}~\mathbf{0}]$, and every entry in the block generator matrix $(\mathbf{I}_k\otimes\mathbf{G}_L)\mathbf{\Psi}_{k\times n}(\mathbf{I}_n\otimes(\tau(\mathbf{C}_L)\mathbf{G}_L^\mathrm{T})$ contains the common term $\tau(\mathbf{C}_L)$, which can be reused during the encoding process to reduce complexity. %
Denote by $\Delta(\mathbf{c})=[\mathbf{c}_0^\mathrm{T}~ \mathbf{c}_1^\mathrm{T} ~ \ldots ~ \mathbf{c}_{n-1}^\mathrm{T}]$ the codeword. %
Consider the case that $q=2$, $\mathbf{G}_L=[\mathbf{I}_J~\mathbf{0}]$ and $\mathbf{H}_L=[\mathbf{I}_J~\mathbf{A}]^\mathrm{T}$, where $\mathbf{A}$ is a matrix in $\mathbb{F}_2^{J\times (L-J)}$. Let $h$ denote the number of nonzero elements in $\mathbf{A}$, and let $\delta$ denote the number of terms in $\tau(x)$. %
For $\mathcal{C}_1$, it takes $(k-1)L+h$ XOR operations to compute $\mathbf{c}_i^\mathrm{T}= (\sum_{s=0}^{k-1}\mathbf{m}_s\mathbf{G}_L\mathbf{C}_L^{2^sl_{i}}\mathbf{H}_L)^\mathrm{T}$ for each $0 \leq i \leq n-1$. In total, it requires $(k-1)nL+nh$ XOR operations to generate a codeword. %
For $\mathcal{C}_2$, it takes $(\delta-1)L$ XOR operations to compute $\bar{\mathbf{m}}_s=\mathbf{m}_s\mathbf{G}_L\tau(\mathbf{C}_L)$ for each $0 \leq s \leq k-1$, and $(k-1)J$ XOR operations to compute $\mathbf{c}_i^\mathrm{T}= (\sum_{s=0}^{k-1}\bar{\mathbf{m}}_s\mathbf{C}_L^{2^sl_{i}}\mathbf{G}_L^\mathrm{T})^\mathrm{T}$ for each $0 \leq i \leq n-1$. In total, it requires $(\delta-1)kL+(k-1)nJ$ XOR operations to generate a codeword. %
One can readily verify that a sufficient condition for $\mathcal{C}_2$ to achieve lower encoding complexity is that $L$ is prime and $k>2$. %
For instance, assume $q=2$, $L=5$, $\mathbf{G}_5=[\mathbf{I}_{4}~\mathbf{0}]$. In this case, $J=4$, $\tau(\mathbf{C}_5) = \mathbf{I}_5+\mathbf{C}_5$ and $\mathbf{H}_5=[\mathbf{I}_4~\mathbf{1}]^\mathrm{T}$. %
Let $k=3$, $n=4$, and let $\mathbf{m}_0, \mathbf{m}_1,\mathbf{m}_2$ denote three nonzero row vectors in $\mathbb{F}_{2}^4$. %
For $\mathcal{C}_2$, the computation can be simplified as follows
\begin{equation}
\label{eqn:reuse}
\begin{split}
&\bar{\mathbf{m}}_s =\mathbf{m}_s\mathbf{G}_5\tau(\mathbf{C}_5), \quad 0\leq s \leq 2 ,\\
&\mathbf{c}_i^\mathrm{T} =(\sum\nolimits_{s=0}^{2}\bar{\mathbf{m}}_s\mathbf{C}_5^{2^sl_{i}}\mathbf{G}_5^\mathrm{T})^\mathrm{T}, \quad 0 \leq i \leq 3.
\end{split}
\end{equation}
It respectively takes $15$ and $32$ XOR operations to compute $\bar{\mathbf{m}}_s$ for all $0 \leq s \leq 2$ and $\mathbf{c}_i^\mathrm{T}$ for all $0 \leq i \leq 3$. %
In all, generating a codeword of $\mathcal{C}_2$ requires $47$ XOR operations. %
In comparison, generating a codeword of $\mathcal{C}_1$ requires $56$ XOR operations. The encoding  complexity will be analyzed in detail in Section \ref{sec:Complexity Analysis}.  %

\section{Inherent relation and difference between the new codes and Gabidulin codes}
\label{sec:Connection with Gabidulin codes}
In this section, we clarify the inherent difference and connection between the proposed circular-shift-based MRD codes and Gabidulin codes, as well as their distinction from twisted Gabidulin codes. The key is to characterize both Gabidulin codes and circular-shift-based MRD codes from the perspective of a set of $q$-linearized polynomials defined over the row vector space $\mathbb{F}_{q}^N$ instead of over $\mathbb{F}_{q^N}$. %

\subsection{An alternative characterization of Gabidulin codes}
\label{subsec:characterization of Gabidulin codes}
Consider an arbitrary irreducible polynomial $p(x)=p_0+p_1x+\ldots+p_{N-1}x^{N-1}+x^N$ of degree $N$ over $\mathbb{F}_q$, and let $\gamma$ be a root of $p(x)$, so that we have a basis $\mathcal{B}_{\gamma}=[\gamma^{N-1}~\gamma^{N-2}~\ldots~\gamma~1]$ of $\mathbb{F}_{q^N}$ over $\mathbb{F}_q$. Denote by $\mathbf{A}(\gamma)$ the following matrix in $\mathbb{F}_q^{N\times N}$:
\begin{equation}
\label{def:companion matrix}
\mathbf{A}(\gamma)=\begin{bmatrix}
-p_{N-1} &-p_{N-2}&\ldots&-p_{0}\\
1 &0&\ldots&0\\
\vdots &\ddots&\ldots&\vdots\\
0 &\ldots&1&0\\
\end{bmatrix}.
\end{equation}
For every $\sum\nolimits_{i=0}^{N-1} a_i\gamma^i \in \mathbb{F}_{q^N}$, $a_i \in \mathbb{F}_q$, denote by $\mathbf{A}(\sum\nolimits_{i=0}^{N-1} a_i\gamma^i) = \sum\nolimits_{i=0}^{N-1} a_i\mathbf{A}(\gamma)^i$. It is well known that the set $\{\mathbf{A}(\sum\nolimits_{i=0}^{N-1} a_i\gamma^i): a_0, \ldots, a_{N-1} \in \mathbb{F}_q \}$ forms a matrix representation of $\mathbb{F}_{q^N}$ over $\mathbb{F}_q$, so that the following equations hold for all $\beta_1, \beta_2 \in \mathbb{F}_{q^N}$ (See, e.g., \cite{matrix-rep})
\begin{equation}
\label{eqn:property of matrix representation}
\mathbf{A}(a_1\beta_1+a_2\beta_2) = a_1\mathbf{A}(\beta_1)+a_2\mathbf{A}(\beta_2) ~ \forall a_1, a_2 \in \mathbb{F}_{q}, \quad
\mathbf{A}(\beta_1\beta_2) = \mathbf{A}(\beta_1)\mathbf{A}(\beta_2), \quad
\mathbf{A}(\beta_1^{-1})=\mathbf{A}(\beta_1)^{-1}.
\end{equation}
In addition, since $\mathbf{A}(\gamma)\mathcal{B}_{\gamma}^\mathrm{T}=\mathcal{B}_{\gamma}^\mathrm{T}\gamma$, it follows that $\mathbf{A}(\beta_1)\mathcal{B}_{\gamma}^\mathrm{T}=\mathcal{B}_{\gamma}^\mathrm{T}\beta_1$ for any $\beta_1 \in \mathbb{F}_{q^N}$. %
This property is a necessary and sufficient condition for the following proposition
\begin{equation}
\label{eqn:useful proposition}
\mathbf{v}_{\mathcal{B}_{\gamma}}(\beta_1)\mathbf{A}(\beta_2) = \mathbf{v}_{\mathcal{B}_{\gamma}}(\beta_1\beta_2), \quad \forall \beta_1, \beta_2 \in \mathbb{F}_{q^N},
\end{equation}
where $\mathbf{v}_{\mathcal{B}_{\gamma}}(\beta_i)$ represents the $q$-ary representation of $\beta_i \in \mathbb{F}_{q^N}$ with respect to the basis $\mathcal{B}_{\gamma}=[\gamma^{N-1}~\gamma^{N-2}~\ldots~\gamma~1]$ of $\mathbb{F}_{q^N}$ over $\mathbb{F}_q$. %
This proposition is crucial for the subsequent discussion of the difference and connection between circular-shift-based MRD codes proposed in Section \ref{sec:Circular-shift-based MRD codes} and Gabidulin codes.

Define a set of $q$-linearized polynomials over the row vector space $\mathbb{F}_{q}^N$ as
\begin{equation}
\label{def:L_dag}
\mathcal{L}_N^\dag=\{\mathbf{e}_0x+\mathbf{e}_1x^{q}+\ldots+\mathbf{e}_{k-1}x^{q^{k-1}}:~\mathbf{e}_0, \mathbf{e}_1, \ldots,\mathbf{e}_{k-1} \in \mathbb{F}_{q}^N\}.
\end{equation}
Based on the matrix representation $\{\mathbf{A}(\sum\nolimits_{i=0}^{N-1} a_i\gamma^i): a_0, \ldots, a_{N-1} \in \mathbb{F}_q \}$ of $\mathbb{F}_{q^N}$ and $\mathcal{L}_N^\dag$, we can obtain the following equivalent characterization for every $(N \times n, q^{Nk}, d)$ Gabidulin code $\mathcal{M}_{\mathcal{B}}(\mathcal{V})$. %
Recall (from Definition \ref{def:Gabidulin codes} and \eqref{eqn:MRD_generator_matrix}) that an $(n, k)$ Gabidulin (linear) code $\mathcal{V} = \{\mathbf{u}\mathbf{G}: \mathbf{u}\in \mathbb{F}_{q^N}^k\}$ over $\mathbb{F}_{q^N}$ is determined by $n$ $\mathbb{F}_q$-linearly independent elements $\beta_0, \beta_1, \ldots, \beta_{n-1}$ in $\mathbb{F}_{q^N}$ which prescribe the generator matrix $\mathbf{G}$. Equivalently, $\mathcal{V}$ can be expressed as 
\begin{equation}
\label{eqn:V=L}
\mathcal{V}=\{[L(\beta_0)~L(\beta_1)~\ldots~L(\beta_{n-1})]: L(x) \in \mathcal{L}\},
\end{equation}
where $\mathcal{L}$ denotes the following set of $q$-linearized polynomials over $\mathbb{F}_{q^N}$
\begin{equation}
\label{def:Gabidulin set L}
\mathcal{L} = \{u_0x+u_1x^q+\ldots+u_{k-1}x^{q^{k-1}}:~u_0, u_1, \ldots,u_{k-1} \in \mathbb{F}_{q^N}\}.
\end{equation}

\begin{proposition}
\label{prop:new charac of Gabidulin}
Let $\beta_0, \beta_1, \ldots, \beta_{n-1}$ be arbitrary $n$ $\mathbb{F}_q$-linearly independent elements in $\mathbb{F}_{q^N}$ that determine the Gabidulin (linear) code $\mathcal{V}$ over $\mathbb{F}_{q^N}$, and $\mathcal{B}$ be an arbitrary basis of $\mathbb{F}_{q^N}$ over $\mathbb{F}_q$. With respect to $\mathcal{B}$, the $(N \times n, q^{Nk}, d)$ Gabidulin code $\mathcal{M}_{\mathcal{B}}(\mathcal{V})$ induced from $\mathcal{V}$ can be equivalently characterized as 
\begin{equation}
\label{expression Gabidulin}
\mathcal{M}_{\mathcal{B}}(\mathcal{V}) = \{[L^\dag(\mathbf{V}\mathbf{A}(\beta_0)\mathbf{V}^{-1})^\mathrm{T}~L^\dag(\mathbf{V}\mathbf{A}(\beta_1)\mathbf{V}^{-1})^\mathrm{T}~\ldots~L^\dag(\mathbf{V}\mathbf{A}(\beta_{n-1})\mathbf{V}^{-1})^\mathrm{T}]: L^\dag(x) \in \mathcal{L}_N^\dag\},
\end{equation}
where $\mathbf{V}$ is the invertible matrix in $\mathbb{F}_q^{N\times N}$ satisfying $\mathcal{B}^\mathrm{T} = \mathbf{V}\mathcal{B}_{\gamma}^\mathrm{T}$ with $\mathcal{B}_{\gamma} = [\gamma^{N-1}~\gamma^{N-2}~\ldots~\gamma~1]$.
\begin{IEEEproof}
Please refer to Appendix-\ref{appendix:prop:new charac of Gabidulin}.
\end{IEEEproof}
\end{proposition}

\subsection{Relation and difference between the new codes and Gabidulin codes}
\label{subsec:Relation}
Let $\mathcal{C}_1$ and $\mathcal{C}_2$ denote the $(J \times n, q^{Jk}, d)$ circular-shift-based MRD codes $\mathcal{C}_1 = \{\Delta(\mathbf{m}(\mathbf{I}_k\otimes\mathbf{G}_L)\mathbf{\Psi}_{k\times n}(\mathbf{I}_n\otimes\mathbf{H}_L)): \mathbf{m}\in \mathbb{F}_{q}^{Jk}\}$ proposed in Section \ref{subsec:GH} and $\mathcal{C}_2= \{\Delta(\mathbf{m}(\mathbf{I}_k\otimes\mathbf{G}_L)\mathbf{\Psi}_{k\times n}(\mathbf{I}_n\otimes(\tau(\mathbf{C}_L)\mathbf{G}_L^\mathrm{T}))): \mathbf{m}\in \mathbb{F}_{q}^{Jk}\} = \mathbf{T}\mathcal{C}_1$ proposed in Section \ref{subsec:GG}, respectively. %
Based on the set $\mathcal{L}^\dag_N$ of $q$-linearized polynomials  over the row vector space $\mathbb{F}_q^{N}$ defined in \eqref{def:L_dag}, under the setting $N = J$, both $\mathcal{C}_1$ and $\mathcal{C}_2$ can be alternatively characterized as 
\begin{equation}
\label{expression C1}
\mathcal{C}_1 = \{[L^\dag(\mathbf{G}_L\mathbf{C}_L^{l_0}\mathbf{H}_L)^\mathrm{T}~L^\dag(\mathbf{G}_L\mathbf{C}_L^{l_1}\mathbf{H}_L)^\mathrm{T}~\ldots~L^\dag(\mathbf{G}_L\mathbf{C}_L^{l_{n-1}}\mathbf{H}_L)^\mathrm{T}]: L^\dag(x) \in \mathcal{L}_J^\dag\},
\end{equation}
\begin{equation}
\label{expression C2}
\mathcal{C}_2 = \{\mathbf{T}[L^\dag(\mathbf{G}_L\mathbf{C}_L^{l_0}\mathbf{H}_L)^\mathrm{T}~L^\dag(\mathbf{G}_L\mathbf{C}_L^{l_1}\mathbf{H}_L)^\mathrm{T}~\ldots~L^\dag(\mathbf{G}_L\mathbf{C}_L^{l_{n-1}}\mathbf{H}_L)^\mathrm{T}]: L^\dag(x) \in \mathcal{L}_J^\dag\}.
\end{equation} 

Based on the unified expression \eqref{expression Gabidulin}, \eqref{expression C1} and \eqref{expression C2} from the perspective of $\mathcal{L}^\dag_J$, we next unveil the inherent difference and connection between the circular-shift-based MRD codes and Gabidulin codes. Recall that in this paper, $J$ is the Euler's totient function of an integer $L$ coprime to $q$, and $m_L$ refers to the multiplicative order of $q$ modulo $L$. %
The following lemma plays a key role in the proof of the difference between $\mathcal{C}_1, \mathcal{C}_2$ and Gabidulin codes. %
Recall that $\mathcal{J}$ is defined in \eqref{J} as the set of all integers between $1$ and $L-1$ that are coprime to $L$.  

\begin{lemma}
\label{GCH neq VAV}
Consider the case $J \neq m_L$. For every $ l \in \mathcal{J}$, $\beta' \in \mathbb{F}_{q^J}$ and full rank matrix $\mathbf{V} \in \mathbb{F}_q^{J\times J}$, 
\begin{equation}
\mathbf{G}_L\mathbf{C}_L^{l}\mathbf{H}_L \neq \mathbf{V}\mathbf{A}(\beta')\mathbf{V}^{-1}.
\end{equation}
\begin{IEEEproof}
Please refer to Appendix-\ref{appendix:proof lemma:GCH neq VAV}.
\end{IEEEproof}
\end{lemma}

Notice that although circular-shift-based MRD codes and Gabidulin codes can be characterized in a unified expression from the perspective of $\mathcal{L}^\dag_J$, Lemma \ref{GCH neq VAV} shows that, in general, the evaluations of $L^\dag(x) \in \mathcal{L}^\dag_J$ are performed under different settings. %
In addition, the following theorem further unveils that, in a family of settings, both $\mathcal{C}_1$ and $\mathcal{C}_2$ are different from any Gabidulin MRD code $\mathcal{M}_{\mathcal{B}}(\mathcal{V})$ over $\mathbb{F}_{q}$, regardless of the choices of basis $\mathcal{B}$ and Gabidulin (linear) code $\mathcal{V}$ over $\mathbb{F}_{q^J}$. %
When $k = n$, all $(J \times n, q^{Jk}, d)$ MRD codes degenerate to the set of all $q^{Jn}$ matrices in $\mathbb{F}_q ^{J \times n}$, so we only focus the case $k<n$. %

\begin{theorem}
\label{theorem:connection with Gabidulin}
Consider the $(J \times n, q^{Jk}, d)$ circular-shift-based MRD codes $\mathcal{C}_1$ and $\mathcal{C}_2$ with $k<n$. The followings hold. 
\begin{enumerate}
  \item Consider the case that $J \neq m_L$, and the parameters $l_0, l_1, \ldots, l_{n-1}$ of $\mathcal{C}_1$, $\mathcal{C}_2$ satisfy
  \begin{equation}
  \label{eqn:lj setting}
  l_j = c'j + c \mod L, ~\forall 0 \leq j \leq n-1,
  \end{equation}
  where $0 \leq c \leq L-1$ and $1 \leq c' \leq L-1$ are two fixed constants with $\sum_{j=0}^{k-1} c'q^{j}$ coprime to $L$. %
  Both $\mathcal{C}_1$ and $\mathcal{C}_2$ are different from any $(J \times n, q^{Jk}, d)$ Gabidulin MRD code $\mathcal{M}_{\mathcal{B}}(\mathcal{V})$, regardless of the choices of basis $\mathcal{B}$ and Gabidulin (linear) code $\mathcal{V}$ over $\mathbb{F}_{q^J}$. 
  \item Consider the case $J=m_L$. Let $\mathcal{V}$ be the Gabidulin (linear) code over $\mathbb{F}_{q^J}$ determined by the $n$ $\mathbb{F}_q$-linearly independent elements $\beta^{l_0}, \beta^{l_1},\ldots,\beta^{l_{n-1}}$ in $\mathbb{F}_{q^J}$, where $\beta$ refers to a primitive $L^{th}$ root of unity in $\mathbb{F}_{q^{m_L}}$, and $l_0, l_1, \ldots, l_{n-1}$ are the exponents adopted in \eqref{expression C1} and \eqref{expression C2} to construct $\mathcal{C}_1$ and $\mathcal{C}_2$. Denote by $\mathcal{B}$ the basis of $\mathbb{F}_{q^J}$ over $\mathbb{F}_q$ equal to $\mathbf{u}_1^\mathrm{T}$, where $\mathbf{u}_1$ is the second column of $\mathbf{U}$ defined in \eqref{U_def}. We have
      \begin{equation}
      \mathcal{C}_1 = \mathcal{M}_{\mathcal{B}}(\mathcal{V}), \quad \mathcal{C}_2 = \mathbf{T}\mathcal{M}_{\mathcal{B}}(\mathcal{V}) = \mathcal{M}_{\mathcal{B}\mathbf{T}^{-1}}(\mathcal{V}),
      \end{equation}
      where $\mathbf{T}$ is the $J\times J$ invertible matrix defined in \eqref{def:T}.
\end{enumerate}   
\begin{IEEEproof}
Recall that Lemma \ref{lemma:C2=TC1} implies that $\mathcal{C}_2 = \mathbf{T}\mathcal{C}_1$, where $\mathbf{T} \in \mathbb{F}_q^{J\times J}$ is the  invertible matrix defined in \eqref{def:T}. %
Thus, it suffices to discuss the  difference and connection between $\mathcal{C}_1$ and Gabidulin codes. %
According to Proposition \ref{prop:new charac of Gabidulin}, the set $\{[L^\dag(\mathbf{V}\mathbf{A}(\beta_0)\mathbf{V}^{-1})^\mathrm{T}~L^\dag(\mathbf{V}\mathbf{A}(\beta_1)\mathbf{V}^{-1})^\mathrm{T}~\ldots~L^\dag(\mathbf{V}\mathbf{A}(\beta_{n-1})\mathbf{V}^{-1})^\mathrm{T}]: L^\dag(x) \in \mathcal{L}_N^\dag\}$ characterizes all $(N \times n, q^{Nk}, d)$ Gabidulin MRD codes $\mathcal{M}_{\mathcal{B}}(\mathcal{V})$ by choosing a full rank matrix $\mathbf{V} \in \mathbb{F}_q^{N \times N}$ determined by the basis $\mathcal{B}$, and choosing $n$ $\mathbb{F}_q$-linearly independent elements $\beta_0, \beta_1, \ldots, \beta_{n-1}$ in  $\mathbb{F}_{q^N}$ associated with the Gabidulin linear code $\mathcal{V}$. %
Consequently, it remains to discuss the difference and connection between 
the two sets in \eqref{expression C1} and \eqref{expression Gabidulin}.

Assume that there exist $n$ $\mathbb{F}_q$-linearly independent elements $\beta_0, \ldots, \beta_{n-1} \in \mathbb{F}_{q^N}$ and a  full rank matrix $\mathbf{V} \in\mathbb{F}_{q}^{N\times N}$ subject to 
\begin{equation}
\label{eqn:MBV=C}
\begin{split}
&\{[L^\dag(\mathbf{V}\mathbf{A}(\beta_0)\mathbf{V}^{-1})^\mathrm{T}~L^\dag(\mathbf{V}\mathbf{A}(\beta_1)\mathbf{V}^{-1})^\mathrm{T}~\ldots~L^\dag(\mathbf{V}\mathbf{A}(\beta_{n-1})\mathbf{V}^{-1})^\mathrm{T}]: L^\dag(x) \in \mathcal{L}_J^\dag\}\\
=&\{[L^\dag(\mathbf{G}_L\mathbf{C}_L^{l_0}\mathbf{H}_L)^\mathrm{T}~L^\dag(\mathbf{G}_L\mathbf{C}_L^{l_1}\mathbf{H}_L)^\mathrm{T}~\ldots~L^\dag(\mathbf{G}_L\mathbf{C}_L^{l_{n-1}}\mathbf{H}_L)^\mathrm{T}]: L^\dag(x) \in \mathcal{L}_J^\dag\}.
\end{split}
\end{equation}
For brevity, for every $0 \leq i \leq k-1$ and $0 \leq j \leq n-1$, write $\bm{\Phi}_{i,j} = \mathbf{V}\mathbf{A}(\beta_j^{q^i})\mathbf{V}^{-1}$ and $\bm{\Phi}_{i,j}' = \mathbf{G}_L\mathbf{C}_L^{q^il_j}\mathbf{H}_L$, so that we have %
  \begin{align}
  \label{eqn:psi_ij}
   [\bm{\Phi}_{i,j}]_{0 \leq i \leq k-1, 0 \leq j \leq n-1}=
   \begin{bmatrix}
\begin{smallmatrix}
\mathbf{V}\mathbf{A}(\beta_0)\mathbf{V}^{-1}&\mathbf{V}\mathbf{A}(\beta_1)\mathbf{V}^{-1}&\ldots&\mathbf{V}\mathbf{A}(\beta_{n-1})\mathbf{V}^{-1}\\
\mathbf{V}\mathbf{A}(\beta_0^q)\mathbf{V}^{-1}&\mathbf{V}\mathbf{A}(\beta_1^q)\mathbf{V}^{-1}&\ldots&\mathbf{V}\mathbf{A}(\beta_{n-1}^q)\mathbf{V}^{-1}\\
\vdots&\vdots&\ldots&\vdots\\
\mathbf{V}\mathbf{A}(\beta_0^{q^{k-1}})\mathbf{V}^{-1}&\mathbf{V}\mathbf{A}(\beta_1^{q^{k-1}})\mathbf{V}^{-1}&\ldots&\mathbf{V}\mathbf{A}(\beta_{n-1}^{q^{k-1}})\mathbf{V}^{-1}
\end{smallmatrix}
\end{bmatrix},\\
[\bm{\Phi}_{i,j}']_{0 \leq i \leq k-1, 0 \leq j \leq n-1}=
\begin{bmatrix}
\begin{smallmatrix}
\mathbf{G}_L\mathbf{C}_L^{l_0}\mathbf{H}_L&\mathbf{G}_L\mathbf{C}_L^{l_1}\mathbf{H}_L&\ldots&\mathbf{G}_L\mathbf{C}_L^{l_{n-1}}\mathbf{H}_L\\
\mathbf{G}_L\mathbf{C}_L^{ql_0}\mathbf{H}_L&\mathbf{G}_L\mathbf{C}_L^{ql_1}\mathbf{H}_L&\ldots&\mathbf{G}_L\mathbf{C}_L^{ql_{n-1}}\mathbf{H}_L\\
\vdots&\vdots&\ldots&\vdots\\
\mathbf{G}_L\mathbf{C}_L^{q^{k-1}l_0}\mathbf{H}_L&\mathbf{G}_L\mathbf{C}_L^{q^{k-1}l_1}\mathbf{H}_L&\ldots&\mathbf{G}_L\mathbf{C}_L^{q^{k-1}l_{n-1}}\mathbf{H}_L\\
\end{smallmatrix}
\end{bmatrix}.
  \end{align} %
Notice that \eqref{eqn:MBV=C} implies $\{\mathbf{e}[\bm{\Phi}_{i,j}]_{0 \leq i \leq k-1, 0 \leq j \leq n-1}: \mathbf{e} \in \mathbb{F}_{q}^{kJ}\}=
\{\mathbf{e}'[\bm{\Phi}'_{i,j}]_{0 \leq i \leq k-1, 0 \leq j \leq n-1}: \mathbf{e}' \in \mathbb{F}_{q}^{kJ}\}$. %
For every unit row vector $\mathbf{e}_{l} \in \mathbb{F}_{q}^{kJ}$, in which the $l^{th}$ entry is equal to $1$, let $\mathbf{e}_l' \in \mathbb{F}_{q}^{kJ}$ denote the row vector satisfying
\begin{equation}
\label{eqn:with e0}
\mathbf{e}_{l}[\bm{\Phi}_{i,j}]_{0 \leq i \leq k-1, 0 \leq j \leq n-1}=\mathbf{e}_l'[\bm{\Phi}_{i,j}']_{0 \leq i \leq k-1, 0 \leq j \leq n-1}.
\end{equation}
Thus, for every $\mathbf{e}= \sum_{l=1}^{kJ} a_l\mathbf{e}_{l} $ with $a_l \in \mathbb{F}_{q}$, we have $\mathbf{e}'= \sum_{l=1}^{kJ} a_l\mathbf{e}_l'$ satisfying
\begin{equation}
\label{eqn:with e0_e_1}
\mathbf{e}
[\bm{\Phi}_{i,j}]_{0 \leq i \leq k-1, 0 \leq j \leq n-1}
=
\mathbf{e}'
[\bm{\Phi}_{i,j}']_{0 \leq i \leq k-1, 0 \leq j \leq n-1}.
\end{equation}
It turns out that \eqref{eqn:MBV=C} holds if and only if
\begin{equation}
\label{eqn:without e0_e1}
[\bm{\Phi}_{i,j}]_{0 \leq i \leq k-1, 0 \leq j \leq n-1}
=\mathbf{N}
[\bm{\Phi}_{i,j}']_{0 \leq i \leq k-1, 0 \leq j \leq n-1},
\end{equation}
where $\mathbf{N} \in \mathbb{F}_q^{kJ \times kJ}$ denotes the full rank matrix in which the $l^{th}$ row is equal to $\mathbf{e}_l'$. %

We can establish the following $n-k$ matrices in $\mathbb{F}_q^{nJ \times J}$
\begin{equation}
\mathbf{D}_t= [(\mathbf{V}\mathbf{A}(\beta_{t, 0}')\mathbf{V}^{-1})^\mathrm{T}, \ldots, (\mathbf{V}\mathbf{A}(\beta_{t, n-t-2}')\mathbf{V}^{-1})^\mathrm{T}, -\mathbf{I}_J,\underbrace{\mathbf{0},\ldots,\mathbf{0}}_{t}]^\mathrm{T},~0 \leq t \leq n-k-1
\end{equation}
 with $\beta_{t, 0}', \ldots, \beta_{t, n-t-2}' \in \mathbb{F}_{q^J}$ such that $[\bm{\Phi}_{i,j}]_{0 \leq i \leq k-1, 0 \leq j \leq n-1}[\mathbf{D}_t]_{0 \leq t \leq n-k-1} = \mathbf{0}$, and $\mathrm{rank}([\mathbf{D}_t]_{0 \leq t \leq n-k-1}) = (n-k)J$. %
 In addition, \eqref{eqn:without e0_e1} further implies that 
 \begin{equation}
 \label{eqn:zero vector 1}
 [\bm{\Phi}_{i,j}']_{0 \leq i \leq k-1, 0 \leq j \leq n-1}[\mathbf{D}_t]_{0 \leq t \leq n-k-1} = \mathbf{0}.
 \end{equation} 
  Define a new matrix $[\bm{\Phi}_{i,j}'']_{0 \leq i \leq k-1, 0 \leq j \leq n-1}$ with every $\bm{\Phi}_{i,j}'' = \mathbf{G}_L\mathbf{C}_L^{q^ijc'}\mathbf{H}_L$
\begin{equation}
[\bm{\Phi}_{i,j}'']_{0 \leq i \leq k-1, 0 \leq j \leq n-1} = \begin{bmatrix}
\begin{smallmatrix}
\mathbf{I}_J & \mathbf{G}_L\mathbf{C}_L^{c'}\mathbf{H}_L&\ldots&\mathbf{G}_L\mathbf{C}_L^{(n-1)c'}\mathbf{H}_L\\
\mathbf{I}_J &\mathbf{G}_L\mathbf{C}_L^{qc'}\mathbf{H}_L&\ldots&\mathbf{G}_L\mathbf{C}_L^{q(n-1)c'}\mathbf{H}_L\\
\vdots&\vdots&\ldots&\vdots\\
\mathbf{I}_J &\mathbf{G}_L\mathbf{C}_L^{q^{k-1}c'}\mathbf{H}_L&\ldots&\mathbf{G}_L\mathbf{C}_L^{q^{k-1}(n-1)c'}\mathbf{H}_L\\
\end{smallmatrix}
\end{bmatrix}.
\end{equation}
As $l_j = c'j + c \mod L$ for every $0 \leq j \leq n-1$, we have
\begin{equation}
[\bm{\Phi}_{i,j}']_{0 \leq i \leq k-1, 0 \leq j \leq n-1} = \mathrm{Diag}([\mathbf{G}_L\mathbf{C}_L^{q^jc}\mathbf{H}_L]_{0 \leq j \leq k-1})[\bm{\Phi}_{i,j}'']_{0 \leq i \leq k-1, 0 \leq j \leq n-1}.
\end{equation}
Let $\mathbf{F}_{0}, \ldots, \mathbf{F}_{k-1}$ denote such $k$ matrices in $\mathbb{F}_q^{J\times J}$ subject to
\[
\begin{bmatrix}
\mathbf{F}_{0} \\
\vdots \\
\mathbf{F}_{k-1}
\end{bmatrix}
= [\bm{\Phi}_{i,j}'']_{0 \leq i,j \leq k-1}^{-1}[\bm{\Phi}_{i,j}'']_{0 \leq i \leq k-1, j = k} = 
\begin{bmatrix}
\begin{smallmatrix}
\mathbf{I}_J & \mathbf{G}_L\mathbf{C}_L^{c'}\mathbf{H}_L&\ldots&\mathbf{G}_L\mathbf{C}_L^{(k-1)c'}\mathbf{H}_L\\
\mathbf{I}_J &\mathbf{G}_L\mathbf{C}_L^{qc'}\mathbf{H}_L&\ldots&\mathbf{G}_L\mathbf{C}_L^{q(k-1)c'}\mathbf{H}_L\\
\vdots&\vdots&\ldots&\vdots\\
\mathbf{I}_J &\mathbf{G}_L\mathbf{C}_L^{q^{k-1}c'}\mathbf{H}_L&\ldots&\mathbf{G}_L\mathbf{C}_L^{q^{k-1}(k-1)c'}\mathbf{H}_L\\
\end{smallmatrix}
\end{bmatrix}^{-1}
\begin{bmatrix}
\begin{smallmatrix}
\mathbf{G}_L\mathbf{C}_L^{kc'}\mathbf{H}_L \\
\mathbf{G}_L\mathbf{C}_L^{qkc'}\mathbf{H}_L \\
\vdots \\
\mathbf{G}_L\mathbf{C}_L^{q^{k-1}kc'}\mathbf{H}_L \\
\end{smallmatrix}
\end{bmatrix}.
\]
Consequently, 
\begin{equation}
\label{eqn:zero vector 2}
\begin{split}
&[\bm{\Phi}_{i,j}']_{0 \leq i \leq k-1, 0 \leq j \leq n-1}
[\mathbf{F}_{0}^\mathrm{T}, \ldots, \mathbf{F}_{k-1}^\mathrm{T}, -\mathbf{I}_J, \underbrace{\mathbf{0},\ldots,\mathbf{0}}_{n-k-1}]^\mathrm{T}\\
= &[\bm{\Phi}_{i,j}'']_{0 \leq i \leq k-1, 0 \leq j \leq n-1}
[\mathbf{F}_{0}^\mathrm{T}, \ldots, \mathbf{F}_{k-1}^\mathrm{T}, -\mathbf{I}_J, \underbrace{\mathbf{0},\ldots,\mathbf{0}}_{n-k-1}]^\mathrm{T}
= \mathbf{0}.
\end{split}
\end{equation}
Let $\mathbf{\Delta}$ and $\mathbf{\Delta}'$ respectively denote the matrices in $\mathbb{F}_q^{J\times J}$ obtained by computing the determinant (in the block-wise manner) of the following two $k\times k$ block matrices with every block entry a matrix in $\mathbb{F}_q^{J\times J}$ 
\[
\mathbf{\Delta} = \det([\bm{\Phi}_{i,j}'']_{0 \leq i, j \leq k-1}) = \det\left(\begin{bmatrix}
\begin{smallmatrix}
\mathbf{I}_J & \mathbf{G}_L\mathbf{C}_L^{c'}\mathbf{H}_L&\ldots&\mathbf{G}_L\mathbf{C}_L^{(k-1)c'}\mathbf{H}_L\\
\mathbf{I}_J &\mathbf{G}_L\mathbf{C}_L^{qc'}\mathbf{H}_L&\ldots&\mathbf{G}_L\mathbf{C}_L^{q(k-1)c'}\mathbf{H}_L\\
\vdots&\vdots&\ldots&\vdots\\
\mathbf{I}_J &\mathbf{G}_L\mathbf{C}_L^{q^{k-1}c'}\mathbf{H}_L&\ldots&\mathbf{G}_L\mathbf{C}_L^{q^{k-1}(k-1)c'}\mathbf{H}_L\\
\end{smallmatrix}
\end{bmatrix}\right),
\]
\[ 
\mathbf{\Delta}' = \det\left( \begin{bmatrix}
\begin{smallmatrix}
\mathbf{G}_L\mathbf{C}_L^{kc'}\mathbf{H}_L & \mathbf{G}_L\mathbf{C}_L^{c'}\mathbf{H}_L&\ldots&\mathbf{G}_L\mathbf{C}_L^{(k-1)c'}\mathbf{H}_L\\
\mathbf{G}_L\mathbf{C}_L^{qkc'}\mathbf{H}_L &\mathbf{G}_L\mathbf{C}_L^{qc'}\mathbf{H}_L&\ldots&\mathbf{G}_L\mathbf{C}_L^{q(k-1)c'}\mathbf{H}_L\\
\vdots&\vdots&\ldots&\vdots\\
\mathbf{G}_L\mathbf{C}_L^{q^{k-1}kc'}\mathbf{H}_L &\mathbf{G}_L\mathbf{C}_L^{q^{k-1}c'}\mathbf{H}_L&\ldots&\mathbf{G}_L\mathbf{C}_L^{q^{k-1}(k-1)c'}\mathbf{H}_L\\
\end{smallmatrix}
\end{bmatrix}
\right) =  (-1)^{k+1}\det([\bm{\Phi}_{i,j+1}'']_{0 \leq i,j \leq k-1}). 
\]
Based on the properties of determinant, we have $\mathbf{F}_{0} = \mathbf{\Delta}'/\mathbf{\Delta}$. %
Notice that  $[\bm{\Phi}_{i,j+1}'']_{0 \leq i,j \leq k-1} = \mathrm{Diag}([\mathbf{G}_L\mathbf{C}_L^{c'q^j}\mathbf{H}_L]_{0 \leq j \leq k-1})[\bm{\Phi}_{i,j}'']_{0 \leq i,j \leq k-1}$. Subsequently,
\[
\mathbf{F}_{0} = \mathbf{\Delta}'/\mathbf{\Delta} = (-1)^{k+1}\mathbf{G}_L\mathbf{C}_L^{\sum_{j=0}^{k-1} c'q^{j}}\mathbf{H}_L.
\]
For brevity, write $\mathbf{D}' = [\mathbf{F}_{0}^\mathrm{T}, \ldots, \mathbf{F}_{k-1}^\mathrm{T}, -\mathbf{I}_J, \underbrace{\mathbf{0},\ldots,\mathbf{0}}_{n-k-1}]^\mathrm{T}$. %
Since $[\bm{\Phi}_{i,j}']_{0 \leq i \leq k-1, 0 \leq j \leq n-1}[\mathbf{D}_t]_{0 \leq t \leq n-k-1} = \mathbf{0}$ and $\mathrm{rank}([\mathbf{D}_t]_{0 \leq t \leq n-k-1}) = (n-k)J$, it follows from \eqref{eqn:zero vector 2} that all column vectors in $\mathbf{D}'$ are $\mathbb{F}_q$-linear combinations among column vectors in $[\mathbf{D}_t]_{0 \leq t \leq n-k-1}$. %
Notice that the last $n-k-1$ block entries of $\mathbf{D}'$ are zero and the last $n-k-1$ block rows of $[\mathbf{D}_t]_{0 \leq t \leq n-k-2}$ are full rank, so the columns in $\mathbf{D}'$ are $\mathbb{F}_q$-linear combinations among columns in $\mathbf{D}_{n-k-1}$. %
Since the $(k+1)^{st}$ block entry of both $\mathbf{D}'$ and $\mathbf{D}_{n-k-1}$ are $-\mathbf{I}_J$, we have $\mathbf{D}' = \mathbf{D}_{n-k-1}$, which further implies $\mathbf{F}_{0} = \mathbf{V}\mathbf{A}(\beta_{t, 0}')\mathbf{V}^{-1}$, that is, $\mathbf{G}_L\mathbf{C}_L^{\sum_{j=0}^{k-1} c'q^{j}}\mathbf{H}_L = \mathbf{V}\mathbf{A}((-1)^{k+1}\beta_{t, 0}')\mathbf{V}^{-1}$. %
According to Lemma \ref{GCH neq VAV}, when $\sum_{j=0}^{k-1} c'q^{j}$ is coprime to $L$, $\mathbf{G}_L\mathbf{C}_L^{\sum_{j=0}^{k-1} c'q^{j}}\mathbf{H}_L\neq\mathbf{V}\mathbf{A}(\beta')\mathbf{V}^{-1}$ for any $\beta' \in \mathbb{F}_{q^J}$ and full rank matrix $\mathbf{V} \in \mathbb{F}_q^{J\times J}$, which yields a contradiction.
We can now conclude that $\mathcal{C}_1$ and $\mathcal{C}_2$ are different from any $(J \times n, q^{Jk}, d)$ Gabidulin code with $k<n$.

Consider the case $J=m_L$. We first show that $\mathcal{B}=\mathbf{u}_1^\mathrm{T}$ forms a basis of $\mathbb{F}_{q^{J}}$ over $\mathbb{F}_q$. %
According to Theorem 2.37 and Corollary 2.38 in \cite{Lidl:Finite_Field}, it suffices to show that
\begin{equation}
\label{eqn:u1}
\det([\mathbf{u}_{1}~(\mathbf{u}_{1})^{\circ q}~ \ldots ~(\mathbf{u}_{1})^{\circ q^{J-1}}]) \neq 0.
\end{equation}
Let $\mathbf{v}_{1}$ denote the second column of $\mathbf{V}_L$ defined in \eqref{eqn:VL}. %
Since $\mathbf{G}_L$ is over $\mathbb{F}_{q}$, we have
\begin{equation}
\mathbf{u}_{q^j\mathrm{mod}~ L}=L_{\mathbb{F}_q}^{-1}\mathbf{G}_L\mathbf{v}_{q^j\mathrm{mod}~ L}=L_{\mathbb{F}_q}^{-1}\mathbf{G}_L(\mathbf{v}_{1})^{\circ q^j}=(L_{\mathbb{F}_q}^{-1}\mathbf{G}_L\mathbf{v}_1)^{\circ q^j}=(\mathbf{u}_{1})^{\circ q^j},
  ~\forall 0\leq j \leq J-1.
\end{equation}
Because $\{q^j ~\mathrm{mod}~L: 0\leq j \leq J-1\}=\mathcal{J}$ for the case $J=m_L$, the full rank of $[\mathbf{u}_j]_{j \in \mathcal{J}}$, as defined in \eqref{U_def}, implies \eqref{eqn:u1}. %
Let $\mathbf{V}$ be the invertible matrix in $\mathbb{F}_q^{J \times J}$ such that $\mathcal{B}^\mathrm{T}=\mathbf{V}\mathcal{B}_{\gamma}^\mathrm{T}$. %
In order to prove $\mathcal{C}_1=\mathcal{M}_{\mathcal{B}}(\mathcal{V})$, based on \eqref{expression Gabidulin} in Proposition \ref{prop:new charac of Gabidulin} and \eqref{expression C1}, it remains to show that 
\begin{equation}
\label{eqn:J=m_L}
\mathbf{G}_L\mathbf{C}_L^l\mathbf{H}_L=\mathbf{V}\mathbf{A}(\beta^l)\mathbf{V}^{-1},\quad \forall 0 \leq l \leq L-1.
\end{equation}
From \eqref{eqn:GH_mu_C}, we have $\mathbf{G}_L\mathbf{C}_L^l\mathbf{H}_L\mathcal{B}^\mathrm{T}=\mathcal{B}^\mathrm{T}\beta^l$, which implies
\begin{equation}
\mathbf{G}_L\mathbf{C}_L^l\mathbf{H}_L\mathbf{V}\mathcal{B}_{\gamma}^\mathrm{T}=\mathbf{V}\mathcal{B}_{\gamma}^\mathrm{T}\beta^l, \quad \forall  0 \leq l \leq L-1.
\end{equation}
Since $\mathbf{A}(\beta^l)\mathcal{B}_{\gamma}^\mathrm{T}=\mathcal{B}_{\gamma}^\mathrm{T}\beta^l$, it follows that $\mathbf{G}_L\mathbf{C}_L^l\mathbf{H}_L\mathbf{V}\mathcal{B}_{\gamma}^\mathrm{T}=\mathbf{V}\mathbf{A}(\beta^l)\mathcal{B}_{\gamma}^\mathrm{T}$, which proves \eqref{eqn:J=m_L}. %
\end{IEEEproof}
\end{theorem}

It is worth remarking that when $J \neq m_L$ and the parameters $l_0, l_1, \ldots, l_{n-1}$ of $\mathcal{C}_1$, $\mathcal{C}_2$ do not satisfy \eqref{eqn:lj setting}, the codes $\mathcal{C}_1$ and $\mathcal{C}_2$ may either coincide with or differ from Gabidulin codes. %
For instance, assume $L=15$ and $q=2$. Let $\mathcal{C}_1$ denote the $(8 \times 3, 2^{16}, 2)$ circular-shift-based MRD code with the $2\times 3$ block generator matrix $[\bm{\Phi}_{i,j}']_{0 \leq i \leq 1, 0 \leq j \leq 2}$, where $\bm{\Phi}_{i,j}'= \mathbf{G}_{15}\mathbf{C}_{15}^{2^il_j}\mathbf{H}_{15}$, $l_0=0, l_1=1, l_2=3$. %
Let $\mathbf{D}_1= [(\mathbf{G}_{15}(\mathbf{C}_{15}^4+\mathbf{C}_{15}^5)\mathbf{H}_{15})^\mathrm{T}, (\mathbf{G}_{15}(\mathbf{C}_{15}^2+\mathbf{C}_{15}^3+\mathbf{C}_{15}^4)\mathbf{H}_{15})^\mathrm{T}, \mathbf{I}_8]^\mathrm{T} \in \mathbb{F}_2^{24 \times 8}$ denote the matrix subject to 
$[\bm{\Phi}_{i,j}']_{0 \leq i \leq 1, 0 \leq j \leq 2}\mathbf{D}_1=\mathbf{0}$. %
Using the same technique as in Lemma \ref{GCH neq VAV}, one can readily verify that $\mathbf{G}_{15}(\mathbf{C}_{15}^4+\mathbf{C}_{15}^5)\mathbf{H}_{15}\neq \mathbf{V}\mathbf{A}(\beta')\mathbf{V}^{-1}$ for any  
$\beta' \in \mathbb{F}_{2^8}$ and  full rank matrix $\mathbf{V} \in \mathbb{F}_2^{8\times 8}$, which further implies that $\mathcal{C}_1$ is different from any $(8 \times 3, 2^{16}, 2)$ Gabidulin MRD code. %
As another instance, assume $L=7$ and $q=2$. Let $\mathcal{C}_1$ denote the $(6 \times 3, 2^{12}, 2)$ circular-shift-based MRD code with the $2 \times 3$ block generator matrix $[\bm{\Phi}_{i,j}']_{0 \leq i \leq 1, 0 \leq j \leq 2}$, where $\bm{\Phi}_{i,j}'= \mathbf{G}_7\mathbf{C}_7^{2^il_j}\mathbf{H}_7$, $l_0=0, l_1=1, l_2=4$. %
Let $\mathbf{D}_2= [(\mathbf{G}_7(\mathbf{I}_7+\mathbf{C}_7^5+\mathbf{C}_7^6)\mathbf{H}_7)^\mathrm{T}, (\mathbf{G}_7(\mathbf{I}_7+\mathbf{C}_7+\mathbf{C}_7^2)\mathbf{H}_7)^\mathrm{T}, \mathbf{I}_6]^\mathrm{T} \in \mathbb{F}_2^{18\times 6}$ denote the matrix subject to 
$[\bm{\Phi}_{i,j}']_{0 \leq i \leq 1, 0 \leq j \leq 2}\mathbf{D}_2=\mathbf{0}$. %
In this case, both $\mathbf{G}_7(\mathbf{I}_7+\mathbf{C}_7^5+\mathbf{C}_7^6)\mathbf{H}_7$ and $\mathbf{G}_7(\mathbf{I}_7+\mathbf{C}_7+\mathbf{C}_7^2)\mathbf{H}_7$ can be expressed in the form $\mathbf{V}\mathbf{A}(\beta')\mathbf{V}^{-1}$, where $\beta' \in \mathbb{F}_{2^6}$ and $\mathbf{V}$ is a full rank matrix in $\mathbb{F}_2^{6\times 6}$. %
This implies that $\mathcal{C}_1$ coincides with a $(6 \times 3, 2^{12}, 2)$ Gabidulin MRD code. %
As a potential future work, it is interesting to further characterize the relationship between circular-shift-based MRD codes, whose parameters $l_0, l_1, \ldots, l_{n-1}$ do not satisfy \eqref{eqn:lj setting}, and Gabidulin codes.

We would also like to remark that in the special case $J = m_L$, even though a $(J \times n, q^{Jk}, d)$ circular-shift-based MRD code coincides with a $(J \times n, q^{Jk}, d)$ Gabidulin MRD code, the construction of $\mathcal{C}_1$ and $\mathcal{C}_2$ in the form of \eqref{eqn:circular-shift-MRD-code-DEF} has its own merit because it provides a new approach to construct Gabidulin codes purely based on the arithmetic in $\mathbb{F}_{q}$ and avoids the arithmetic in $\mathbb{F}_{q^{J}}$. In comparison, the classical constructions of Gabidulin codes, from the perspective of both the matrix representation and the vector representation as reviewed in Section \ref{subsec:Gabidulin codes}, involve the arithmetic of $\mathbb{F}_{q^{J}}$.

\begin{example}
Consider $q=2$ and $L=5$, so that $\mathcal{J}=\{1, 2, 3, 4\}$ and $\tau(\mathbf{C}_{5})=\mathbf{I}_{5}+\mathbf{C}_{5}$. %
Set
\begin{equation*}
\mathbf{G}_5=[\mathbf{I}_4~\mathbf{1}], \mathbf{H}_5=[\mathbf{I}_4~\mathbf{0}]^\mathrm{T},
\mathbf{\Psi}_{1\times 4} = 
[\mathbf{I}_5~\mathbf{C}_5 ~\mathbf{C}_5^2~\mathbf{C}_5^3], 
\mathbf{G}=[1~\beta~\beta^2~\beta^3],
\end{equation*}
where $\beta \in \mathbb{F}_{2^4}$ denotes a primitive $5^{th}$ root of unity, so that \eqref{eqn:independent C_GH} and \eqref{eqn:independent C_GG} hold. %
The $4 \times 5$ matrices $\mathbf{U}$ and $\mathbf{U}'$ can be obtained from \eqref{eqn:GH_def}, and thus $\mathbf{u}_1$ is given by
\begin{equation*}
\mathbf{u}_1=[\alpha^{11}~\alpha^{10}~\alpha^4~\alpha^8]^\mathrm{T}, 
\end{equation*}
where $\alpha$ denotes a root of the primitive polynomial $p(x)=x^4+x+1$ over $\mathbb{F}_2$. %
Based on \eqref{def:T}, the $4\times 4$ full rank matrix $\mathbf{T}$ is $\mathbf{T}=[\mathbf{u}_{5-j}]_{j \in \mathcal{J}} \mathrm{Diag}([\tau(\beta^j)]_{j\in \mathcal{J}}) [\mathbf{u}_{j}]_{j \in \mathcal{J}}^\mathrm{T}=\begin{bmatrix}\begin{smallmatrix}1&0&0&1\\0&0&1&0\\1&0&0&0\\1&1&0&1\end{smallmatrix}\end{bmatrix}$. %
Define
\begin{equation}
\mathcal{B}' =\mathbf{u}_1^\mathrm{T}\mathbf{T}^{-1}=[\alpha~\alpha^{4}~\alpha^7~\alpha^{10}].
\end{equation}
It can be verified that $\mathbf{u}_1^\mathrm{T}$ and $\mathcal{B}'$ are indeed qualified as bases of $\mathbb{F}_{2^4}$ over $\mathbb{F}_2$. %

Sequentially set $\mathbf{m}$ as a row vector in $\mathbb{F}_2^4$
\begin{equation}
\begin{split}
\mathbf{m} \in \{&[0~0~0~0], [0~0~0~1], [0~0~1~0], [0~0~1~1], [0~1~0~0],[0~1~0~1],[0~1~1~0], [0~1~1~1],\\
&[1~0~0~0],[1~0~0~1], [1~0~1~0],[1~0~1~1], [1~1~0~0],[1~1~0~1], [1~1~1~0], [1~1~1~1]\}.
\end{split}
\end{equation} 
Then, the  codewords in  the $(4 \times 4, 2^4, 4)$ circular-shift-based MRD code $\mathcal{C}_1 = \{\Delta(\mathbf{m}(\mathbf{I}_1\otimes\mathbf{G}_5)\mathbf{\Psi}_{1\times 4}(\mathbf{I}_4\otimes\mathbf{H}_5), \mathbf{m} \in \mathbb{F}_2^4\}$ are
\begin{equation}
\label{eqn:C1}
\begin{split}
&
\begin{bmatrix}\begin{smallmatrix}
   0&   0&   0 &0\\
   0&   0&   0 &0\\
   0&   0&   0 &0\\
   0&   0&   0 &0
   \end{smallmatrix}\end{bmatrix},~
\begin{bmatrix}\begin{smallmatrix}
   0  & 1  & 1  & 0\\
   0  & 0  & 1 &  1\\
   0  & 0  & 0 &  1\\
   1 &  0  & 0 &  0
   \end{smallmatrix}\end{bmatrix},~
\begin{bmatrix}\begin{smallmatrix}
   0 &  1 &  0 &  1\\
   0 &  0  & 1 &  0\\
   1  & 0  & 0 &  1\\
   0  & 1  & 0  & 0
   \end{smallmatrix}\end{bmatrix},~
\begin{bmatrix}\begin{smallmatrix}
   0  & 0  & 1 &  1\\
   0  & 0  & 0 &  1\\
   1  & 0  & 0 &  0\\
   1  & 1  & 0 &  0
   \end{smallmatrix}\end{bmatrix},~
\begin{bmatrix}\begin{smallmatrix}
   0 &  1 &  0  & 0\\
   1  & 0 &  1  & 0\\
   0  & 1 &  0 &  1\\
   0  & 0 &  1 &  0
   \end{smallmatrix}\end{bmatrix},~
\begin{bmatrix}\begin{smallmatrix}
   0 &  0 &  1 &  0\\
   1 &  0  & 0 &  1\\
   0  & 1  & 0 &  0\\
   1  & 0  & 1 &  0
   \end{smallmatrix}\end{bmatrix},~
\begin{bmatrix}\begin{smallmatrix}
   0 &  0 &  0 &  1\\
   1 &  0 &  0 &  0\\
   1 &  1 &  0 &  0\\
   0 &  1 &  1 &  0
   \end{smallmatrix}\end{bmatrix},~
\begin{bmatrix}\begin{smallmatrix}
   0  & 1  & 1  & 1\\
   1  & 0  & 1  & 1\\
   1  & 1  & 0  & 1\\
   1 &  1 &  1 &  0
   \end{smallmatrix}\end{bmatrix},\\
&
\begin{bmatrix}\begin{smallmatrix}
   1 &  1 &  0  & 0\\
   0 &  1  & 1  & 0\\
   0  & 0  & 1 &  1\\
   0 &  0  & 0 &  1
   \end{smallmatrix}\end{bmatrix},~
\begin{bmatrix}\begin{smallmatrix}
   1 &  0 &  1 &  0\\
   0 &  1 &  0 &  1\\
   0 &  0 &  1 &  0\\
   1 &  0 &  0 &  1
   \end{smallmatrix}\end{bmatrix},~
\begin{bmatrix}\begin{smallmatrix}
   1 &  0 &  0  & 1\\
   0 &  1 &  0 &  0\\
   1 &  0 &  1 &  0\\
   0  & 1 &  0  & 1
   \end{smallmatrix}\end{bmatrix},~
\begin{bmatrix}\begin{smallmatrix}
   1  & 1 &  1 &  1\\
   0  & 1 &  1 &  1\\
   1 &  0 &  1 &  1\\
   1 &  1 &  0 &  1
   \end{smallmatrix}\end{bmatrix},~
\begin{bmatrix}\begin{smallmatrix}
   1 &  0 &  0 &  0\\
   1  & 1 &  0 &  0\\
   0 &  1 &  1 &  0\\
   0 &  0 &  1 &  1
   \end{smallmatrix}\end{bmatrix},~
\begin{bmatrix}\begin{smallmatrix}
   1 &  1  & 1 &  0\\
   1 &  1  & 1 &  1\\
   0 &  1  & 1 &  1\\
   1 &  0 &  1 &  1
  \end{smallmatrix} \end{bmatrix},~
\begin{bmatrix}\begin{smallmatrix}
   1  & 1 &  0 &  1\\
   1  & 1 &  1 &  0\\
   1  & 1 &  1 &  1\\
   0  & 1 &  1 &  1
   \end{smallmatrix}\end{bmatrix},~
\begin{bmatrix}\begin{smallmatrix}
   1  & 0 &  1 &  1\\
   1  & 1  & 0 &  1\\
   1 &  1 &  1 &  0\\
   1 &  1 &  1 &  1
   \end{smallmatrix}\end{bmatrix},
   \end{split}
\end{equation}
and the codewords in the $(4 \times 4, 2^4, 4)$ circular-shift-based MRD code $\mathcal{C}_2 = \{\Delta(\mathbf{m}(\mathbf{I}_1\otimes\mathbf{G}_5)\mathbf{\Psi}_{1\times 4}(\mathbf{I}_4\otimes(\tau(\mathbf{C}_5)\mathbf{G}_5^\mathrm{T})), \mathbf{m} \in \mathbb{F}_2^4\}$ are
\begin{equation}
\label{eqn:C2}
\begin{split}
&
\begin{bmatrix}\begin{smallmatrix}
   0&   0&   0 &0\\
   0&   0&   0 &0\\
   0&   0&   0 &0\\
   0&   0&   0 &0
   \end{smallmatrix}\end{bmatrix},~
\begin{bmatrix}\begin{smallmatrix}
   1  & 1  & 1  & 0\\
   0  & 0  & 0  & 1\\
   0  & 1 &  1  & 0\\
   1  & 1 &  0 &  1
   \end{smallmatrix}\end{bmatrix},~
\begin{bmatrix}\begin{smallmatrix}
   0 &  0 &  0  & 1\\
   1 &  0 &  0  & 1\\
   0 &  1  & 0 &  1\\
   0  & 0  & 1  & 1
   \end{smallmatrix}\end{bmatrix},~
\begin{bmatrix}\begin{smallmatrix}
   1 &  1 &  1  & 1\\
   1 &  0 &  0  & 0\\
   0 &  0 &  1 &  1\\
   1 &  1 &  1 &  0
   \end{smallmatrix}\end{bmatrix},~
\begin{bmatrix}\begin{smallmatrix}
   0 &  1  & 1  & 0\\
   0 &  1 &  0  & 1\\
   0  & 1 &  0 &  0\\
   1 &  1 &  0 &  0
   \end{smallmatrix}\end{bmatrix},~
\begin{bmatrix}\begin{smallmatrix}
   1  & 0  & 0  & 0\\
   0  & 1 &  0 &  0\\
   0  & 0 &  1 &  0\\
   0  & 0  & 0 &  1
   \end{smallmatrix}\end{bmatrix},~
\begin{bmatrix}\begin{smallmatrix}
   0  & 1 &  1   &1\\
   1  & 1 &  0  & 0\\
   0  & 0 &  0 &  1\\
   1  & 1 &  1 &  1
   \end{smallmatrix}\end{bmatrix},~
\begin{bmatrix}\begin{smallmatrix}
   1 &  0  & 0 &  1\\
   1 &  1  & 0 &  1\\
   0 &  1 &  1 &  1\\
   0  & 0 &  1 &  0
   \end{smallmatrix}\end{bmatrix},\\
&
\begin{bmatrix}\begin{smallmatrix}
   1 &  1 &  0 &  1\\
   0 &  0 &  1  & 1\\
   1 &  1 &  0 &  0\\
   1 &  0 &  1 &  1
   \end{smallmatrix}\end{bmatrix},~
\begin{bmatrix}\begin{smallmatrix}
   0  & 0 &  1  & 1\\
   0  & 0 &  1  & 0\\
   1  & 0 &  1 &  0\\
   0 &  1 &  1 &  0
   \end{smallmatrix}\end{bmatrix},~
\begin{bmatrix}\begin{smallmatrix}
   1  & 1  & 0 &  0\\
   1  & 0  & 1 &  0\\
   1  & 0  & 0 &  1\\
   1  & 0  & 0 &  0
   \end{smallmatrix}\end{bmatrix},~
\begin{bmatrix}\begin{smallmatrix}
   0 &  0  & 1 &  0\\
   1 &  0  & 1 &  1\\
   1  & 1  & 1 &  1\\
   0  & 1 &  0 &  1
   \end{smallmatrix}\end{bmatrix},~
\begin{bmatrix}\begin{smallmatrix}
   1  & 0  & 1  & 1\\
   0  & 1 &  1  & 0\\
   1  & 0 &  0  & 0\\
   0  & 1 &  1 &  1
   \end{smallmatrix}\end{bmatrix},~
\begin{bmatrix}\begin{smallmatrix}
   0 &  1 &  0  & 1\\
   0  & 1 &  1 &  1\\
   1  & 1 &  1 &  0\\
   1 &  0  & 1  & 0
  \end{smallmatrix} \end{bmatrix},~
\begin{bmatrix}\begin{smallmatrix}
   1  & 0  & 1 &  0\\
   1  & 1 &  1 &  1\\
   1  & 1 &  0 &  1\\
   0  & 1 &  0  & 0
   \end{smallmatrix}\end{bmatrix},~
\begin{bmatrix}\begin{smallmatrix}
   0  & 1 &  0 &  0\\
   1  & 1 &  1 &  0\\
   1  & 0 &  1 &  1\\
   1  & 0 &  0 &  1
   \end{smallmatrix}\end{bmatrix}.
   \end{split}
\end{equation}
For each $\mathbf{u} \in \{0, 1, \alpha, \alpha^2, \alpha^3, \alpha^4, \alpha^5, \alpha^6, \alpha^7, \alpha^8, \alpha^9, \alpha^{10}, \alpha^{11}, \alpha^{12},\alpha^{13}, \alpha^{14}\}=\mathbb{F}_{2^4}$, the codewords in the $(4,1)$ Gabidulin code $\mathcal{V} = \{\mathbf{u}\mathbf{G}: \mathbf{u}\in \mathbb{F}_{2^4}\}$ are
\begin{align*}
\label{eqn:Gabidulin_code_example}
&[0~ 0~ 0 ~ 0],[1~\alpha^3~\alpha^6~\alpha^9], [\alpha~\alpha^4~\alpha^7~\alpha^{10}],[\alpha^2~\alpha^5~\alpha^8~\alpha^{11}], 
[\alpha^3  ~\alpha^6~\alpha^9 ~\alpha^{12}],[\alpha^4~\alpha^7 ~\alpha^{10}~\alpha^{13}],[\alpha^5~\alpha^8~\alpha^{11}~\alpha^{14}],\\
&[\alpha^6~\alpha^9~\alpha^{12}~    1],[\alpha^7~\alpha^{10} ~\alpha^{13}~\alpha],[\alpha^8~\alpha^{11}~\alpha^{14}~\alpha^2],
 [\alpha^{9}~\alpha^{12} ~ 1   ~ \alpha^3],[\alpha^{10}~\alpha^{13}~\alpha ~\alpha^4],[\alpha^{11} ~\alpha^{14} ~ \alpha^2~\alpha^5],\\
 &[\alpha^{12}~ 1 ~ \alpha^3 ~\alpha^6],[\alpha^{13}~ \alpha ~ \alpha^4 ~\alpha^7],[\alpha^{14}~ \alpha^2~\alpha^5 ~\alpha^8].
\end{align*}
With respect to $\mathbf{u}_1^\mathrm{T}$, the codewords in the $(4 \times 4, 2^4, 4)$ Gabidulin MRD code $\mathcal{M}_{\mathbf{u}_1^\mathrm{T}}(\mathcal{V})$  are 
\begin{equation}
\label{eqn:M1}
\begin{split}
&
\begin{bmatrix}\begin{smallmatrix}
   0&   0&   0 &0\\
   0&   0&   0 &0\\
   0&   0&   0 &0\\
   0&   0&   0 &0
   \end{smallmatrix}\end{bmatrix},~
\begin{bmatrix}\begin{smallmatrix}
     0&     1&     1&     1\\
     1&     0&     1&     1\\
     1&     1&     0&     1\\
     1&     1&     1&     0
   \end{smallmatrix}\end{bmatrix},~
   \begin{bmatrix}\begin{smallmatrix}
     0&     0&     1&     0\\
     1&     0&     0&     1\\
     0&     1&     0&     0\\
     1&     0&     1&     0
   \end{smallmatrix}\end{bmatrix},~
   \begin{bmatrix}\begin{smallmatrix}
     0&     0&     0&     1\\
     1&     0&     0&     0\\
     1&     1&     0&     0\\
     0&     1&     1&     0
   \end{smallmatrix}\end{bmatrix},~
   \begin{bmatrix}\begin{smallmatrix}
    1&     1 &    1&     1\\
    0 &    1&     1&     1\\
    1&     0&     1&     1\\
    1&     1&     0&     1
   \end{smallmatrix}\end{bmatrix},~
   \begin{bmatrix}\begin{smallmatrix}
     0&      1 &    0&     1\\
     0&     0&     1&     0\\
     1 &    0 &    0&     1\\
     0 &    1&     0&     0
   \end{smallmatrix}\end{bmatrix},~
   \begin{bmatrix}\begin{smallmatrix}
     0&     0&     1&     1\\
     0&     0&     0&     1\\
     1&     0&     0&     0\\
     1&     1&     0&     0
  \end{smallmatrix} \end{bmatrix},~
\begin{bmatrix}\begin{smallmatrix}
     1&     1 &    1&     0\\
     1&     1&     1&     1\\
     0&     1&     1&     1\\
     1&     0&     1&     1 \\
   \end{smallmatrix}\end{bmatrix},\\
   &  
   \begin{bmatrix}\begin{smallmatrix}
     1&     0&     1&     0\\
     0&     1&     0&     1\\
     0&     0&     1&     0\\
     1&    0&     0&     1\\
   \end{smallmatrix}\end{bmatrix},~
\begin{bmatrix}\begin{smallmatrix}
    0 &    1 &      1&     0\\
     0 &    0&     1&     1\\
     0 &    0&     0&     1\\
     1 &    0&     0&     0
   \end{smallmatrix}\end{bmatrix},~
   \begin{bmatrix}\begin{smallmatrix}
   1&     1 &    0 &    1\\
     1&     1&     1 &    0\\
     1&     1&     1&     1\\
     0 &    1&     1 &    1
   \end{smallmatrix}\end{bmatrix},~
   \begin{bmatrix}\begin{smallmatrix}
     0&     1 &    0&     0\\
     1&     0 &    1 &    0\\
     0&     1&     0&     1\\
     0&     0&     1&     0\\
   \end{smallmatrix}\end{bmatrix},~
   \begin{bmatrix}\begin{smallmatrix}
     1 &     1 &     0&     0\\
     0&     1&     1&     0\\
     0&     0 &    1&     1\\
     0 &    0&     0&     1
   \end{smallmatrix}\end{bmatrix},~
   \begin{bmatrix}\begin{smallmatrix}
   1 &    0 &    1 &    1\\
     1&     1 &    0 &    1\\
     1 &    1&     1&     0\\
     1 &    1 &    1 &    1
   \end{smallmatrix}\end{bmatrix},~
\begin{bmatrix}\begin{smallmatrix}
   1 &    0 &    0&     1\\
     0 &    1&     0 &    0\\
     1 &    0&     1&     0\\
     0 &    1 &    0 &    1
   \end{smallmatrix}\end{bmatrix},~
   \begin{bmatrix}\begin{smallmatrix}
   1  &   0 &    0&     0\\
     1 &    1&     0 &    0\\
     0&     1 &    1 &    0\\
     0 &    0 &    1&     1
   \end{smallmatrix}\end{bmatrix}.
\end{split}
\end{equation}
With respect to $\mathcal{B}'$, the codewords in the $(4 \times 4, 2^4, 4)$ Gabidulin MRD code $\mathcal{M}_{\mathcal{B}'}(\mathcal{V})$  are
\begin{equation}
\label{eqn:M2}
\begin{split}
&
\begin{bmatrix}\begin{smallmatrix}
   0&   0&   0 &0\\
   0&   0&   0 &0\\
   0&   0&   0 &0\\
   0&   0&   0 &0
   \end{smallmatrix}\end{bmatrix},~
   \begin{bmatrix}\begin{smallmatrix}
   1 &  0  & 0 &  1\\
   1 &  1  & 0 &  1\\
   0 &  1 &  1 &  1\\
   0  & 0 &  1 &  0
   \end{smallmatrix}\end{bmatrix},~
   \begin{bmatrix}\begin{smallmatrix}
   1  & 0  & 0  & 0\\
   0  & 1 &  0 &  0\\
   0  & 0 &  1 &  0\\
   0  & 0  & 0 &  1
   \end{smallmatrix}\end{bmatrix},~
   \begin{bmatrix}\begin{smallmatrix}
   0  & 1 &  1   &1\\
   1  & 1 &  0  & 0\\
   0  & 0 &  0 &  1\\
   1  & 1 &  1 &  1
   \end{smallmatrix}\end{bmatrix},~
\begin{bmatrix}\begin{smallmatrix}
   0 &  0  & 1 &  0\\
   1 &  0  & 1 &  1\\
   1  & 1  & 1 &  1\\
   0  & 1 &  0 &  1
   \end{smallmatrix}\end{bmatrix},~
 \begin{bmatrix}\begin{smallmatrix}
   0 &  0 &  0  & 1\\
   1 &  0 &  0  & 1\\
   0 &  1  & 0 &  1\\
   0  & 0  & 1  & 1
   \end{smallmatrix}\end{bmatrix},~
 \begin{bmatrix}\begin{smallmatrix}
   1 &  1 &  1  & 1\\
   1 &  0 &  0  & 0\\
   0 &  0 &  1 &  1\\
   1 &  1 &  1 &  0
   \end{smallmatrix}\end{bmatrix},~  
\begin{bmatrix}\begin{smallmatrix}
   0 &  1 &  0  & 1\\
   0  & 1 &  1 &  1\\
   1  & 1 &  1 &  0\\
   1 &  0  & 1  & 0
  \end{smallmatrix} \end{bmatrix},\\
  &
 \begin{bmatrix}\begin{smallmatrix}
   0  & 0 &  1  & 1\\
   0  & 0 &  1  & 0\\
   1  & 0 &  1 &  0\\
   0 &  1 &  1 &  0
   \end{smallmatrix}\end{bmatrix},~
\begin{bmatrix}\begin{smallmatrix}
   1  & 1  & 1  & 0\\
   0  & 0  & 0  & 1\\
   0  & 1 &  1  & 0\\
   1  & 1 &  0 &  1
   \end{smallmatrix}\end{bmatrix},~
   \begin{bmatrix}\begin{smallmatrix}
   1  & 0  & 1 &  0\\
   1  & 1 &  1 &  1\\
   1  & 1 &  0 &  1\\
   0  & 1 &  0  & 0
   \end{smallmatrix}\end{bmatrix},~
\begin{bmatrix}\begin{smallmatrix}
   0 &  1 &  1 & 0\\
   0 &  1 &  0 & 1\\
   0 &  1 &  0 &  0\\
   1 &  1 &  0 &  0
   \end{smallmatrix}\end{bmatrix},~
\begin{bmatrix}\begin{smallmatrix}
   1 &  1 &  0 &  1\\
   0 &  0 &  1 &  1\\
   1 &  1 &  0 &  0\\
   1 &  0 &  1 &  1
   \end{smallmatrix}\end{bmatrix},~
\begin{bmatrix}\begin{smallmatrix}
   0  & 1 &  0 &  0\\
   1  & 1 &  1 &  0\\
   1  & 0 &  1 &  1\\
   1  & 0 &  0 &  1
   \end{smallmatrix}\end{bmatrix},~
 \begin{bmatrix}\begin{smallmatrix}
   1  & 1  & 0 &  0\\
   1  & 0  & 1 &  0\\
   1  & 0  & 0 &  1\\
   1  & 0  & 0 &  0
   \end{smallmatrix}\end{bmatrix},~
\begin{bmatrix}\begin{smallmatrix}
   1  & 0  & 1  & 1\\
   0  & 1 &  1  & 0\\
   1  & 0 &  0  & 0\\
   0  & 1 &  1 &  1
   \end{smallmatrix}\end{bmatrix}.
   \end{split}
\end{equation}
One may readily check that the $16$ codewords in \eqref{eqn:C1} coincide with the $16$ codewords in \eqref{eqn:M1}, and the $16$ codewords in \eqref{eqn:C2} coincide with the $16$ codewords in \eqref{eqn:M2}.
\hfill $\blacksquare$ %
\end{example}

Recall that in \eqref{eqn:V=L}, the Gabidulin (linear) code $\mathcal{V} = \{\mathbf{u}\mathbf{G}: \mathbf{u}\in \mathbb{F}_{q^N}^k\}$ over $\mathbb{F}_{q^N}$ can be equivalently expressed based on the set $\mathcal{L} = \{u_0x+u_1x^q+\ldots+u_{k-1}x^{q^{k-1}}:~u_0, u_1, \ldots,u_{k-1} \in \mathbb{F}_{q^N}\}$ of $q$-linearized polynomials over $\mathbb{F}_{q^N}$. %
Based on an extension of $\mathcal{L}$, a new class of MRD codes, called \emph{twisted Gabidulin} codes \cite{TG}, is defined as
\begin{equation}
\mathcal{S}=\{[h(\beta_0)~h(\beta_1)~\ldots~h(\beta_{n-1})]: h(x) \in \mathcal{H}\},
\end{equation}
where $\beta_0, \beta_1, \ldots, \beta_{n-1}$ denotes arbitrary $n$ $\mathbb{F}_q$-linearly independent elements in $\mathbb{F}_{q^N}$, and $\mathcal{H}$ denotes the following set of $q$-linearized polynomials over $\mathbb{F}_{q^N}$
\begin{equation}
\mathcal{H}=\{u_0x+u_1x^q+\ldots+u_{k-1}x^{q^{k-1}}+\eta u_0^{q^h}x^{q^k}:~u_0, u_1, \ldots,u_{k-1} \in \mathbb{F}_{q^N}\},
\end{equation}
with $h$ a positive integer and $\eta \in \mathbb{F}_{q^N}$ satisfying $\eta ^{\frac{q^N-1}{q-1}}\neq (-1)^{nk}$. %
With respect to a basis $\mathcal{B}$ of $\mathbb{F}_{q^N}$ over $\mathbb{F}_q$, the $(N \times n, q^{Nk}, d)$ twisted Gabidulin MRD code $\mathcal{M}_{\mathcal{B}}(\mathcal{S})$ can be expressed as 
\begin{equation}
\label{eqn:proposition MBT}
\mathcal{M}_{\mathcal{B}}(\mathcal{S})=\{[\mathbf{v}_{\mathcal{B}}(h(\beta_0))^\mathrm{T}~\mathbf{v}_{\mathcal{B}}(h(\beta_1))^\mathrm{T}~\ldots~\mathbf{v}_{\mathcal{B}}(h(\beta_{n-1}))^\mathrm{T}]: h(x) \in \mathcal{H}\}.
\end{equation}

It is clear that when $\eta = 0$, $\mathcal{H}$ degenerates to $\mathcal{L}$, and twisted Gabidulin codes are same as conventional Gabidulin codes. %
When $\eta \neq 0$, twisted Gabidulin codes have been proven to contain new MRD codes that are not equivalent to any Gabidulin code (See e.g. \cite{TG}). %
We now compare the difference between the circular-shift-based MRD codes $\mathcal{C}_1, \mathcal{C}_2$ and twisted Gabidulin codes. %

First, $\mathcal{C}_1$ and $\mathcal{C}_2$ can be constructed over $\mathbb{F}_{2}$. In comparison, due to the constraint of $\eta$, when $\eta \neq 0$, a twisted Gabidulin code cannot be constructed over $\mathbb{F}_2$ (See, e.g., \cite{Gabidulin}, \cite{Rank-metric codes and their applications}). %
Next, consider the case $q > 2$, $\eta \neq 0$ and $h = 0$. Similar to the equivalent characterization of the $(N \times n, q^{Nk}, d)$ Gabidulin MRD code $\mathcal{M}_{\mathcal{B}}(\mathcal{V})$ in \eqref{expression Gabidulin} from the perspective of $\mathcal{L}_N^\dag$, the $(N \times n, q^{Nk}, d)$ twisted Gabidulin MRD code $\mathcal{M}_{\mathcal{B}}(\mathcal{S})$ can be equivalently characterized as
\begin{equation}
\label{expression twisted Gabidulin}
\mathcal{M}_{\mathcal{B}}(\mathcal{S}) = \{[H^\dag(\mathbf{V}\mathbf{A}(\beta_0)\mathbf{V}^{-1})^\mathrm{T}~H^\dag(\mathbf{V}\mathbf{A}(\beta_1)\mathbf{V}^{-1})^\mathrm{T}~\ldots~H^\dag(\mathbf{V}\mathbf{A}(\beta_{n-1})\mathbf{V}^{-1})^\mathrm{T}]: H^\dag(x) \in \mathcal{H}_N^\dag\},
\end{equation}
where $\mathcal{H}_N^\dag$ denotes the following set of $q$-linearized polynomials over the row vector space $\mathbb{F}_{q}^N$
\begin{equation}
\mathcal{H}_N^\dag=\{\mathbf{e}_0(x+\mathbf{V}\mathbf{A}(\eta)\mathbf{V}^{-1}x^{q^k})+\mathbf{e}_1x^{q}+\ldots+\mathbf{e}_{k-1}x^{q^{k-1}}:~\mathbf{e}_0, \mathbf{e}_1, \ldots,\mathbf{e}_{k-1} \in \mathbb{F}_{q}^N\}.
\end{equation}
Herein, $\mathbf{A}(\eta)$ denotes the matrix representation of $\eta \in \mathbb{F}_{q^N}$, and $\mathbf{V}$ is the invertible matrix in $\mathbb{F}_q^{N\times N}$ associated with the basis $\mathcal{B}$ as defined in Proposition \ref{prop:new charac of Gabidulin}. %
Consequently, by an essentially same approach to prove Theorem \ref{theorem:connection with Gabidulin} (except that the first block row of \eqref{eqn:psi_ij} is replaced with  $[\mathbf{V}\mathbf{A}(\beta_0+\eta\beta_0^{q^k})\mathbf{V}^{-1}~\mathbf{V}\mathbf{A}(\beta_1+\eta\beta_1^{q^k})\mathbf{V}^{-1}~\ldots~\mathbf{V}\mathbf{A}(\beta_{n-1}+\eta\beta_{n-1}^{q^k})\mathbf{V}^{-1}]$), we can deduce the following proposition that demonstrates the difference between $\mathcal{C}_1$, $\mathcal{C}_2$ under certain parameter settings and any $(J\times n, q^{Jk}, d)$ twisted Gabidulin MRD code (under the setting $h = 0$). %

\begin{proposition}
\label{propo:connection with twisted_Gabidulin}
Consider the $(J\times n, q^{Jk}, d)$ circular-shift-based MRD codes $\mathcal{C}_1$ and $\mathcal{C}_2$ with $k<n$ and $J \neq m_L$. 
If the parameters $l_0, l_1, \ldots, l_{n-1}$ of $\mathcal{C}_1$, $\mathcal{C}_2$ satisfy \eqref{eqn:lj setting}, both $\mathcal{C}_1$ and $\mathcal{C}_2$ are different from any $(J\times n, q^{Jk}, d)$ twisted Gabidulin MRD code $\mathcal{M}_{\mathcal{B}}(\mathcal{S})$ with $h=0$, regardless of the choices of basis $\mathcal{B}$ and twisted Gabidulin (linear) code $\mathcal{S}$ over $\mathbb{F}_{q^J}$.  
\end{proposition}

The characterization of the relationship between circular-shift-based MRD codes and twisted Gabidulin codes with $q > 2$, $\eta \neq 0$, $h \geq 1$ is beyond the focus of this paper and left as an interesting future work. 
\subsection{Generalization of Gabidulin codes}
\label{subsec:Generalization of Gabidulin codes}
In this subsection, in addition to the difference between the proposed MRD codes and Gabidulin codes established in the previous subsection, we further show that $\mathcal{C}_1$ and $\mathcal{C}_2$ are equivalent to a generalization of Gabidulin codes, under some particular selection of $\mathbf{G}_L$ and $\mathbf{H}_L$.

We first reformulate $\mathcal{C}_1$ and $\mathcal{C}_2$ in a form of matrix-matrix multiplication over $\mathbb{F}_{q^{m_L}}$, which is analogous to the matrix representation of Gabidulin codes $\mathcal{M} = \{\mathbf{M}(\mathbf{u}): \mathbf{u}\in \mathbb{F}_{q^{m_L}}^k\}$ given in Definition \ref{def:matrix  representation}. %
Adopting the definitions from Section \ref{sec:Circular-shift-based MRD codes}, let $\mathbf{m}=[\mathbf{m}_0, \mathbf{m}_1, \ldots, \mathbf{m}_{k-1}]$ denote a row vector in $\mathbb{F}_q^{Jk}$, where $\mathbf{m}_s \in \mathbb{F}_q^J$ for $0 \leq s \leq k-1$. %
Recall that $\mathbf{u}_{j}$ denotes the $(j+1)^{st}$ column of $\mathbf{U}$ defined in \eqref{U_def} for $0 \leq j \leq L-1$. %
For $j \in \mathcal{J}$, let $\mathbf{t}_j=[t_{0,j}~t_{1,j}~\ldots~t_{k-1,j}] \in \mathbb{F}_{q^{m_L}}^k$ denote the row vector with
\begin{equation}
\label{eqn:tsj}
t_{s,j}=L_{\mathbb{F}_{q}}\mathbf{m}_s\mathbf{u}_{j}, \quad 0 \leq s \leq k-1,
\end{equation}
and let $L_{\mathbf{t}_j}(x)$ denote the following $q$-linearized polynomial over $\mathbb{F}_{q^{m_L}}$
\begin{equation}
\label{eqn:Ltj(x)}
L_{\mathbf{t}_j}(x)=
\begin{cases}
\sum\nolimits_{s=0}^{k-1} t_{s,j}x^{q^s}, \quad \quad \quad ~\mathrm{for}~\mathcal{C}_1,\\
\sum\nolimits_{s=0}^{k-1} t_{s,j}\tau(\beta^{j})x^{q^s},\quad \mathrm{for}~\mathcal{C}_2,
\end{cases}
\end{equation}
where $\tau(x)$ is defined in \eqref{def: tau(x)}. %
Define the following  matrix in $\mathbb{F}_{q^{m_L}}^{J\times n}$ 
\begin{equation}
\label{eqn:L_t}
\mathbf{L}_{\mathbf{m}}=\begin{bmatrix}
L_{\mathbf{t}_{d_1}}(\beta^{d_1l_0}) & L_{\mathbf{t}_{d_1}}(\beta^{d_1l_1}) &\ldots &L_{\mathbf{t}_{d_1}}(\beta^{d_1l_{n-1}})\\
L_{\mathbf{t}_{d_2}}(\beta^{d_2l_0}) & L_{\mathbf{t}_{d_2}}(\beta^{d_2l_1}) &\ldots &L_{\mathbf{t}_{d_2}}(\beta^{d_2l_{n-1}})\\
\vdots & \vdots &\ldots &\vdots\\
L_{\mathbf{t}_{d_J}}(\beta^{d_Jl_0}) & L_{\mathbf{t}_{d_J}}(\beta^{d_Jl_1}) &\ldots &L_{\mathbf{t}_{d_J}}(\beta^{d_Jl_{n-1}})
\end{bmatrix},
\end{equation}
where $d_1< d_2 <\ldots < d_J$ with  $d_1, d_2,\ldots, d_J \in \mathcal{J}$. %

The above definitions enable us to characterize $\mathcal{C}_1$ and $\mathcal{C}_2$ in a form of matrix-matrix multiplications over $\mathbb{F}_{q^{m_L}}$. %
Hereafter in this paper, let $\mathbf{U}$ and $\mathbf{U}'$ be two full rank matrices over $\mathbb{F}_{q^{m_L}}$ that satisfy \eqref{U_def} and ensure that $\mathbf{G}_L$ and $\mathbf{H}_L$ constructed according to \eqref{eqn:GH_def} are over $\mathbb{F}_{q}$. %
We present the following theorem and lemma that will be useful in the analysis of $\mathcal{C}_1$ and $\mathcal{C}_2$. %
\begin{theorem}
\label{theorem:matrix form of C}
The $(J\times n, q^{Jk}, d)$ circular-shift-based MRD codes $\mathcal{C}_1$ and $\mathcal{C}_2$ can be expressed as follows
\begin{align}
\label{matrix C1}
\mathcal{C}_1 =&\left\{[\mathbf{u}'_j]_{j\in \mathcal{ J}}\mathbf{L}_{\mathbf{m}}: \mathbf{m}\in \mathbb{F}_{q}^{Jk}\right \},\\
\label{matrix C2}
\mathcal{C}_2 =&\left\{[\mathbf{u}_{L-j}]_{j\in \mathcal{ J}}\mathbf{L}_{\mathbf{m}}: \mathbf{m}\in \mathbb{F}_{q}^{Jk}\right \},
\end{align}
where $[\mathbf{u}'_j]_{j\in \mathcal{ J}} \in \mathbb{F}_{q^{m_L}}^{J\times J}$ and $[\mathbf{u}_{L-j}]_{j\in \mathcal{ J}} \in \mathbb{F}_{q^{m_L}}^{J\times J}$ are two  full rank matrices prescribed in \eqref{U_def}.
\begin{IEEEproof}
Please refer to Appendix-\ref{appendix:Proof of theorem}.
\end{IEEEproof}
\end{theorem}

\begin{lemma}
\label{lemma: basis}
For any integer $h$ and $j \in \mathcal{J}$, the elements $\beta^{jh}, \beta^{j(h+1)}, \ldots ,\beta^{j(h+m_L-1)}$ form a basis of $\mathbb{F}_{q^{m_L}}$ over $\mathbb{F}_{q}$, where $\beta \in \mathbb{F}_{q^{m_L}}$ denotes a primitive $L^{th}$ root of unity. %
\begin{IEEEproof}
Please refer to Appendix-\ref{appendix:basis}.
\end{IEEEproof}
\end{lemma}

For $1 \leq i, s\leq J/{m_L}$, let $\mathcal{M}_{i,s}$ denote the following $(m_L\times n, q^{m_Lk}, d)$ Gabidulin code 
\begin{equation}
\label{difference}
\mathcal{M}_{i,s}=\{\mathbf{M}_\mathbf{o}(\mathcal{B}_{i,s})\mathbf{L}_{\bm{\lambda}_s}: \bm{\lambda}_s\in \mathbb{F}_{q^{m_L}}^k\},
\end{equation}
where $\mathcal{B}_{i,s}$ is a basis of $\mathbb{F}_{q^{m_L}}$ over $\mathbb{F}_{q}$, and 
$\mathbf{L}_{\bm{\lambda}_s}$ is a matrix in $\mathbb{F}_{q^{m_L}}^{m_L \times n}$ defined as follows. %
Given a row vector $\bm{\lambda}_s=[\lambda_0, \lambda_1, \ldots, \lambda_{k-1}] \in \mathbb{F}_{q^{m_L}}^k$, define the $q$-linearized polynomial over $\mathbb{F}_{q^{m_L}}$ as follows
\begin{equation}
\label{eqn:L lambda s}
L_{\bm{\lambda}_s}(x)=\sum\nolimits_{i=0}^{k-1} \lambda_ix^{q^i}.
\end{equation}
Let $\mathcal{F}_{s}=\{\beta_{s,0}, \ldots, \beta_{s,n-1}\}$ be a set of $n$ $\mathbb{F}_q$-linearly independent elements in $\mathbb{F}_{q^{m_L}}$. %
Then $\mathbf{L}_{\bm{\lambda}_s}$ is given by
\begin{equation}
\label{eqn:Lus}
\mathbf{L}_{\bm{\lambda}_s}=
\begin{bmatrix}
L_{\bm{\lambda}_s}(\beta_{s,0})  & L_{\bm{\lambda}_s}(\beta_{s,1})&\ldots & L_{\bm{\lambda}_s}(\beta_{s,n-1}) \\
L_{\bm{\lambda}_s}(\beta_{s,0})^q  &L_{\bm{\lambda}_s}(\beta_{s,1})^q& \ldots & L_{\bm{\lambda}_s}(\beta_{s,n-1})^q \\
 \vdots  & \vdots&\ldots& \vdots \\
L_{\bm{\lambda}_s}(\beta_{s,0})^{q^{m_L-1}}  &L_{\bm{\lambda}_s}(\beta_{s,1})^{q^{m_L-1}}& \ldots & L_{\bm{\lambda}_s}(\beta_{s,n-1})^{q^{m_L-1}} \\
\end{bmatrix}. %
\end{equation}
Based on the $(m_L\times n, q^{m_Lk}, d)$ Gabidulin codes ${\mathcal{M}_{i,s}}$ defined above for $1 \leq i, s\leq J/{m_L}$, we define the following $(J\times n, q^{Jk}, d)$ rank-metric code $\widetilde{\mathcal{M}}$
\begin{equation}
\label{eqn:generalization of Gabidulin codes}
\widetilde{\mathcal{M}}=
\begin{bmatrix}
\sum\nolimits_{s=1}^{J/{m_L}}\mathcal{M}_{1,s}\\
\sum\nolimits_{s=1}^{J/{m_L}}\mathcal{M}_{2,s}\\
\vdots\\
\sum\nolimits_{s=1}^{J/{m_L}}\mathcal{M}_{J/{m_L},s}\\
\end{bmatrix},
\end{equation}
which can be viewed as a generalization of Gabidulin codes. %
To the best of our knowledge, this generalization of Gabidulin codes has not been explored in the literature. %
It is worth noting that $\widetilde{\mathcal{M}}$ is not necessarily an MRD code.

The following theorem shows that under particular selection of $\mathbf{G}_L$ and $\mathbf{H}_L$, the $(J\times n, q^{Jk}, d)$ circular-shift-based MRD codes $\mathcal{C}_1$ and $\mathcal{C}_2$ can also be represented in the form of $\widetilde{\mathcal{M}}$. %
In this sense, the new constructed $(J\times n, q^{Jk}, d)$ MRD codes (with $n \leq m_L<J$) differ from conventional Gabidulin codes over $\mathbb{F}_{q^{m_L}}$. 
According to Theorem 2.47 in \cite{Lidl:Finite_Field}, $\mathcal{J}$ defined in \eqref{J} can be expressed as
\begin{equation}
\mathcal{J}=\mathcal{J}_1\cup\mathcal{J}_2\cup\ldots\cup\mathcal{J}_{J/{m_L}},
\end{equation}
where each subset $\mathcal{J}_s$, $1 \leq s \leq J/{m_L}$, contains $m_L$ elements and is closed under multiplication by $q$ modulo $L$. %
For every $\mathcal{J}_s$, let $\varepsilon_s$ denote a defined integer in it so that all $m_L$ elements in $\mathcal{J}_s$ can be written as $\varepsilon_sq^j~\mathrm{mod}~L$ with $0\leq j \leq m_L-1$. %

For $1 \leq i, s\leq J/{m_L}$, let $\mathcal{B}_{i,s}$ denote the following basis of $\mathbb{F}_{q^{m_L}}$ over $\mathbb{F}_{q}$
\begin{align}
\label{Bis}
\mathcal{B}_{i,s}=[\beta^{(i-1)m_L(L-\varepsilon_s)}~\ldots~\beta^{(im_L-1)(L-\varepsilon_s)}],
\end{align}
and let $\mathcal{F}_{s}$ denote the following set of $n$ $\mathbb{F}_q$-linearly independent elements in $\mathbb{F}_{q^{m_L}}$ 
\begin{equation}
\label{eqn:Ls}
\mathcal{F}_{s} = \{\beta^{\varepsilon_sl_0}, \beta^{\varepsilon_sl_1}, \ldots, \beta^{\varepsilon_sl_{n-1}}\}.
\end{equation}

\begin{theorem}
\label{generalization} 
Set $\mathbf{H}_L = [\mathbf{I}_J~\mathbf{0}]^\mathrm{T}$ in the definition of $\mathcal{C}_1$, and  $\mathbf{G}_L = [\mathbf{I}_J~\mathbf{0}]$ in the definition of $\mathcal{C}_2$. %
Let $\widetilde{\mathcal{M}}$ be the $(J\times n, q^{Jk}, d)$ rank-metric code defined in \eqref{eqn:generalization of Gabidulin codes} with respect to $\mathcal{B}_{i,s}$ and $\mathcal{F}_{s}$. %
Then,  
\begin{equation}
\mathcal{C}_1=\mathcal{C}_2=\widetilde{\mathcal{M}}.
\end{equation}
\begin{IEEEproof}
Please refer to Appendix-\ref{appendix:generalization}.
\end{IEEEproof}
\end{theorem}

Analogously, the $(J\times n, q^{Jk}, d)$ circular-shift-based MRD codes $\mathcal{C}_1$ under the setting $\mathbf{H}_L = [\mathbf{0}~\mathbf{I}_J]^\mathrm{T}$ and $\mathcal{C}_2$ under the setting $\mathbf{G}_L=[\mathbf{0}~\mathbf{I}_J]$ can also be represented in the form of $\widetilde{\mathcal{M}}$.  %

\begin{example}
\label{example 5}
Consider $q=2$ and $L=7$, so that $J=6$, $m_L=3$. Let $\varepsilon_1=1$ and $\varepsilon_2=3$. %
We have
\begin{equation}
\begin{split}
\mathcal{B}_{1,1}=[1~\beta^{6}~\beta^{12}],~\mathcal{B}_{1,2}=[1~\beta^{4}~\beta^{8}],~\mathcal{B}_{2,1}=[\beta^{18}~\beta^{24}~\beta^{30}],~\mathcal{B}_{2,2}=[\beta^{12}~\beta^{16}~\beta^{20}],
\end{split}
\end{equation}
where $\beta \in \mathbb{F}_{2^3}$ denotes a primitive $7^{th}$ root of unity. %
Set $\mathcal{F}_{1}=\{1, \beta, \beta^2\}, \mathcal{F}_{2}=\{1, \beta^{3},\beta^6\}$. %
It can be verified that both $\mathcal{F}_{1}$ and $\mathcal{F}_{2}$ consist of $3$ $\mathbb{F}_2$-linearly independent elements in $\mathbb{F}_{2^3}$. %
For $1 \leq i, s\leq 2$, let $\mathcal{M}_{i,s}=\{\mathbf{M}_\mathbf{o}(\mathcal{B}_{i,s})\mathbf{L}_{\bm{\lambda}_s}: \bm{\lambda}_s\in \mathbb{F}_{2^3}\}$ denote the $(3 \times 3, 2^3 ,3)$ Gabidulin code, where $\mathbf{L}_{\bm{\lambda}_s}$ is defined in  \eqref{eqn:Lus}. %
Thus, the codewords in $\begin{bmatrix}\begin{smallmatrix}\mathcal{M}_{1,1}\\\mathcal{M}_{2,1}\end{smallmatrix}\end{bmatrix}$ are 
\begin{align}
\label{eqn:m tilde example 1}
\begin{bmatrix}\begin{smallmatrix}
   0 &0&0\\
   0& 0& 0\\
   0& 0& 0\\
   0 &0 &0\\
   0& 0 &0\\
   0 &0& 0\\
 \end{smallmatrix}\end{bmatrix},
\begin{bmatrix}\begin{smallmatrix}
   1 &  0 &  0\\
   1 &  1 &  0\\
   1 &  1 &  1\\
   0 &  1 &  1\\
   1 &  0 &  1\\
   0  & 1 &  0\\
 \end{smallmatrix}\end{bmatrix},
\begin{bmatrix}\begin{smallmatrix}
   0 &  0 &  1\\
   1 &  0  & 0\\
   1 &  1 &  0\\
   1 &  1&   1\\
   0 &  1  & 1\\
   1 &  0  & 1\\
  \end{smallmatrix} \end{bmatrix},
\begin{bmatrix}\begin{smallmatrix}
   0  & 1 &  0\\
   0  & 0 &  1\\
   1 &  0 &  0\\
   1 &  1 &  0\\
   1 &  1 &  1\\
   0 &  1 &  1
\end{smallmatrix}\end{bmatrix},
\begin{bmatrix}\begin{smallmatrix}
1  & 0 &  1\\
   0  & 1  & 0\\
   0  & 0 &  1\\
   1  & 0 &  0\\
   1 &  1 &  0\\
   1  & 1 &  1
 \end{smallmatrix}\end{bmatrix},
 \begin{bmatrix}\begin{smallmatrix}  
   0  & 1 &  1\\
   1 &  0  & 1\\
   0  & 1 &  0\\
   0  & 0 &  1\\
   1  & 0&   0\\
   1  & 1 &  0
   \end{smallmatrix}\end{bmatrix},
\begin{bmatrix} \begin{smallmatrix}  
   1  & 1  & 1\\
   0 &  1  & 1\\
   1  & 0  & 1\\
   0 &  1  & 0\\
   0  & 0  & 1\\
   1  & 0 &  0
   \end{smallmatrix}\end{bmatrix},
\begin{bmatrix}  \begin{smallmatrix}
   1  & 1 &  0\\
   1 &  1  & 1\\
   0 &  1 &  1\\
   1 &  0 &  1\\
   0 &  1 &  0\\
   0 &  0 &  1
  \end{smallmatrix} \end{bmatrix},
   \end{align}
and the codewords in $\begin{bmatrix}\begin{smallmatrix}\mathcal{M}_{1,2}\\\mathcal{M}_{2,2}\end{smallmatrix}\end{bmatrix}$ are
\begin{align}
\label{eqn:m tilde example 2}
\begin{bmatrix}\begin{smallmatrix}
   0 &0 &0\\
   0 &0 &0\\
   0& 0 &0\\
   0 &0 &0\\
   0 &0 &0\\
   0 &0 &0\\
 \end{smallmatrix}\end{bmatrix},
\begin{bmatrix}\begin{smallmatrix}
   1  & 1 &  1\\
   0 &  1 &  1\\
   0 &  0 &  1\\
   1 &  0 &  0\\
   0 &  1  & 0\\
   1 &  0  & 1
 \end{smallmatrix}\end{bmatrix},
\begin{bmatrix}\begin{smallmatrix}
   0  & 0 &  1\\
   1 &  0 &  0\\
   0 &  1 &  0\\
   1 &  0 &  1\\
   1  & 1 &  0\\
   1  & 1 &  1
   \end{smallmatrix}\end{bmatrix},
\begin{bmatrix}\begin{smallmatrix}
   0 &  1  & 0\\
   1  & 0  & 1\\
   1  & 1  & 0\\
   1 &  1 &  1\\
   0 &  1 &  1\\
   0  & 0 &  1
\end{smallmatrix}\end{bmatrix},
\begin{bmatrix}\begin{smallmatrix}
1 &  1  & 0\\
   1 &  1  & 1\\
   0 &  1  & 1\\
   0 &  0 &  1\\
   1 &  0 &  0\\
   0  & 1 &  0
\end{smallmatrix}\end{bmatrix},
\begin{bmatrix}\begin{smallmatrix}
0  & 1  & 1\\
   0 &  0 &  1\\
   1 &  0 &  0\\
   0 &  1 &  0\\
   1 &  0 &  1\\
   1  & 1 &  0
  \end{smallmatrix} \end{bmatrix},
\begin{bmatrix}\begin{smallmatrix}
1  & 0  & 0\\
   0 &  1 &  0\\
   1  & 0 &  1\\
   1  & 1 &  0\\
   1  & 1  & 1\\
   0 &  1 &  1
  \end{smallmatrix} \end{bmatrix},
\begin{bmatrix}\begin{smallmatrix}
1  & 0  & 1\\
   1  & 1 &  0\\
   1  & 1 &  1\\
   0  & 1  & 1\\
   0  & 0  & 1\\
   1 &  0  & 0
   \end{smallmatrix}\end{bmatrix}.
   \end{align}  
Based on the construction of $\widetilde{\mathcal{M}}$ in \eqref{eqn:generalization of Gabidulin codes},  each of the $64$ codewords in $\widetilde{\mathcal{M}}=\begin{bmatrix}\begin{smallmatrix}
\mathcal{M}_{1,1}+\mathcal{M}_{1,2}\\
\mathcal{M}_{2,1}+\mathcal{M}_{2,2}\\
\end{smallmatrix}\end{bmatrix}$ is the sum of one codeword from \eqref{eqn:m tilde example 1} and one from \eqref{eqn:m tilde example 2}.  %
One may readily check that the $64$ codewords in $\widetilde{\mathcal{M}}$ coincide with the $64$ codewords in the $(6\times 3, 2^6 ,3)$ circular-shift-based MRD code in Example \ref{example 2}, as well as those in Example \ref{example 3} (Please refer to Appendix-\ref{appendix:Equivalence verification} for the details). %
\hfill $\blacksquare$
\end{example}

Theorem \ref{generalization} shows that when the $J\times L$ full rank matrix $\mathbf{G}_L$ or the $L\times J$ full rank matrix $\mathbf{H}_L$ is selected in a specific form, the $(J\times n, q^{Jk}, d)$ circular-shift-based MRD codes $\mathcal{C}_1$ and $\mathcal{C}_2$ are equivalent to a generalization of Gabidulin codes over $\mathbb{F}_{q^{m_L}}$. %
However, as the following example demonstrates, there exist selections of $\mathbf{G}_L$ and $\mathbf{H}_L$ under which the circular-shift-based MRD codes $\mathcal{C}_1$ and $\mathcal{C}_2$ cannot be represented in the form of $\widetilde{\mathcal{M}}$. %

\begin{example}
Consider $q=2$ and $L=7$. In this case, $m_L=3$, $\mathcal{J}=\{1,2,3,4,5,6\}$ and $\tau(\mathbf{C}_7)=\mathbf{I}_7+\mathbf{C}_7$. %
Set 
\begin{equation}
\mathbf{G}_7=[\mathbf{I}_6~\mathbf{A}],~\mathbf{A}=[1 ~1~0~0~0~0]^\mathrm{T},~\mathbf{H}_7=\begin{bmatrix}\begin{smallmatrix}
0  & 1 &  1  & 1  & 1 &  1\\
1  & 0  & 1  & 1  & 1 &  1\\
0  & 0 &  1 &  0  & 0  & 0\\
0  & 0  & 0  & 1 &  0  & 0\\
0  & 0  & 0 &  0  & 1  & 0\\
0  & 0  & 0 &  0  & 0  & 1\\
1 &  1  & 1 &  1  & 1  & 1
\end{smallmatrix}\end{bmatrix},~\mathbf{\Psi}_{1\times 3} = 
\begin{bmatrix}
\mathbf{I}_7&\mathbf{C}_7 & \mathbf{C}_7^2
\end{bmatrix}. 
\end{equation}
According to \eqref{U_def} and \eqref{eqn:GH_def}, we obtain the following matrices
\begin{equation}
\label{exist example}
[\mathbf{u}'_j]_{j\in \mathcal{ J}}=
\begin{bmatrix}
\alpha^5 &  \alpha^3 &  \alpha^6  & \alpha^6  & \alpha^3 &  \alpha^5\\
\alpha^3 & \alpha^6 &  \alpha  & \alpha^5  & \alpha^4  & \alpha^2\\
1  & 1  & \alpha^4  & 1 &  \alpha^2  & \alpha\\
0  & 0  & \alpha^3  & 0 &  \alpha^5 &  \alpha^6\\
\alpha^6  & \alpha^5  & 0 &  \alpha^3  & 0 &  0\\
\alpha  & \alpha^2  & 1  & \alpha^4 &  1 &  1
\end{bmatrix}, \quad
[\mathbf{u}_{L-j}]_{j\in \mathcal{ J}}=
\begin{bmatrix}
\alpha^3  &  \alpha^6  &  \alpha   & \alpha^5  &  \alpha^4   & \alpha^2\\
\alpha^5   & \alpha^3  &  \alpha^6  &  \alpha^6  &  \alpha^3   & \alpha^5\\
\alpha^5   & \alpha^3   & \alpha   & \alpha^6   & \alpha^4   & \alpha^2\\
\alpha^4   & \alpha   & \alpha^5   & \alpha^2   & \alpha^6   & \alpha^3\\
\alpha^3  &  \alpha^6   & \alpha^2   & \alpha^5   & \alpha   & \alpha^4\\
\alpha^2  &  \alpha^4   & \alpha^6   & \alpha   & \alpha^3   & \alpha^5
\end{bmatrix},
   \end{equation}
where $\alpha$ denotes a root of the primitive polynomial $p(x)=x^3+x+1$ over $\mathbb{F}_2$. %
Based on the settings above, let $\mathcal{C}_1$ and $\mathcal{C}_2$ respectively denote the $(6\times 3, 2^6 ,3)$ circular-shift-based MRD codes $\mathcal{C}_1 = \{\Delta(\mathbf{m}(\mathbf{I}_1\otimes\mathbf{G}_7)\mathbf{\Psi}_{1\times 3}(\mathbf{I}_3\otimes\mathbf{H}_7)): \mathbf{m}\in \mathbb{F}_{2}^{6}\}$ and 
$\mathcal{C}_2 = \{\Delta(\mathbf{m}(\mathbf{I}_1\otimes\mathbf{G}_7)\mathbf{\Psi}_{1\times 3}(\mathbf{I}_3\otimes(\tau(\mathbf{C}_7)\mathbf{G}_7^\mathrm{T}))): \mathbf{m}\in \mathbb{F}_{2}^{6}\}$. %
As discussed in the proof of Theorem \ref{generalization}, a necessary condition for representing the circular-shift-based MRD codes in the form of $\widetilde{\mathcal{M}}$ is that the elements in every column in $[\mathbf{u}'_j]_{j\in \mathcal{ J}}$ and $[\mathbf{u}_{L-j}]_{j\in \mathcal{J}}$ can be decomposed into $J/m_L$ bases of $\mathbb{F}_{q^{m_L}}$ over $\mathbb{F}_{q}$, %
which does not hold for $\mathbf{u}'_j$ or $\mathbf{u}_{L-j}$ given in \eqref{exist example}. %
Consequently, $\mathcal{C}_1$  and $\mathcal{C}_2$ cannot be represented in the form of $\widetilde{\mathcal{M}}$.  \hfill $\blacksquare$%
\end{example}

In view of the previous example, we conclude with the following proposition.

\begin{proposition}
There exist constructions of $\mathbf{G}_L$ and $\mathbf{H}_L$ over $\mathbb{F}_{q}$ such that the corresponding circular-shift-based MRD codes $\mathcal{C}_1$ and $\mathcal{C}_2$ cannot be represented in the form of $\widetilde{\mathcal{M}}$. %
\end{proposition}

\section{Computational complexity analysis}
\label{sec:Complexity Analysis}
A key advantage of circular-shift-based MRD codes is that the construction avoids the arithmetic over $\mathbb{F}_{q^{m_L}}$. %
For an $L$-dimensional row vector $\mathbf{m}$, we follow the conventional assumption (similar to \cite{Tang_LNC_TIT}\cite{Hou_Basic}) that there is no computational complexity cost to perform the column-wise circular-shift operation $\mathbf{m}\mathbf{C}_L^j$ on $\mathbf{m}$. %
From a practical point of view, we only consider the case $q = 2$. In addition, for the ease of presentation, we only deal with the case that $L$ is prime, so that $J=L-1$, $\tau(\mathbf{C}_L)=\mathbf{I}_L+\mathbf{C}_L$ and we can set that $\mathbf{G}_L = [\mathbf{I}_{L-1} ~ \mathbf{0}]$ and $\mathbf{H}_L= [\mathbf{I}_{L-1} ~ \mathbf{1}]^\mathrm{T}$. %
For $n \leq m_L$, we compare the computational complexity of generating a codeword of the following classes of rank-metric codes: i) the $((L-1) \times n, 2^{(L-1)k}, d)$ circular-shift-based MRD codes $\mathcal{C}_1$ and $\mathcal{C}_2$ introduced in Section \ref{sec:Circular-shift-based MRD codes}; ii) the $((L-1) \times n, 2^{(L-1)k}, d)$ Gabidulin MRD code, which is induced from an $(n, k)$ Gabidulin (linear) code over $\mathbb{F}_{2^{L-1}}$; iii) the $((L-1) \times n, 2^{(L-1)k}, d)$ rank-metric code $\widetilde{\mathcal{M}}$ which is constructed based on $(n, k)$ Gabidulin (linear) codes over $\mathbb{F}_{2^{m_L}}$ as introduced in \eqref{eqn:generalization of Gabidulin codes}.

We first analyze the computational complexity of generating a codeword ({\emph{i.e.}, a matrix in $\mathbb{F}_{2}^{(L-1)\times n}$) of an $((L-1) \times n, 2^{(L-1)k}, d)$ circular-shift-based MRD code $\mathcal{C}_1 = \{\Delta(\mathbf{m}(\mathbf{I}_k\otimes\mathbf{G}_L)\mathbf{\Psi}_{k\times n}(\mathbf{I}_n\otimes\mathbf{H}_L)): \mathbf{m}\in \mathbb{F}_{2}^{(L-1)k}\}$. %
Given a nonzero row vector $\mathbf{m}=[\mathbf{m}_0, \mathbf{m}_1, \ldots, \mathbf{m}_{k-1}]$ over $\mathbb{F}_{2}$, where $\mathbf{m}_s \in \mathbb{F}_2^{L-1}$ for $ 0 \leq s \leq k-1$, 
let $\Delta(\mathbf{c})=[\mathbf{c}_0^\mathrm{T}~ \mathbf{c}_1^\mathrm{T} ~ \ldots ~ \mathbf{c}_{n-1}^\mathrm{T}]$ denote the codeword of $\mathcal{C}_1$. %
For every $0\leq i \leq n-1$, the column $\mathbf{c}_i^\mathrm{T}$ in $\Delta(\mathbf{c})$ can be computed as follows
\begin{align}
\label{eqn:step 1}
&\mathbf{c}_i'=\sum\nolimits_{s=0}^{k-1}\mathbf{m}_s\mathbf{G}_L\mathbf{C}_L^{2^sl_{i}},\\
\label{eqn:step 2}
&\mathbf{c}_i^\mathrm{T} = (\mathbf{c}_i'\mathbf{H}_L)^\mathrm{T}.
\end{align}
It can be easily checked that it takes $(k-1)L$ and $L-1$ XOR operations to compute \eqref{eqn:step 1} and \eqref{eqn:step 2}, respectively. %
In all, it takes $n(kL-1)=nkL-n$ XOR operations to generate a codeword of $\mathcal{C}_1$. %

We then analyze the computational complexity of generating a codeword ({\emph{i.e.}, a matrix in $\mathbb{F}_{2}^{(L-1)\times n}$) of an $((L-1) \times n, 2^{(L-1)k}, d)$ circular-shift-based MRD code $\mathcal{C}_2 = \{\Delta(\mathbf{m}(\mathbf{I}_k\otimes\mathbf{G}_L)\mathbf{\Psi}_{k\times n}(\mathbf{I}_n\otimes(\tau(\mathbf{C}_L)\mathbf{G}_L^\mathrm{T}))): \mathbf{m}\in \mathbb{F}_{2}^{(L-1)k}\}$. %
Given a nonzero row vector $\mathbf{m}=[\mathbf{m}_0, \mathbf{m}_1, \ldots, \mathbf{m}_{k-1}]$ over $\mathbb{F}_{2}$, where $\mathbf{m}_s \in \mathbb{F}_2^{L-1}$ for $ 0 \leq s \leq k-1$, 
let $\Delta(\bar{\mathbf{c}})=[\bar{\mathbf{c}}_0^\mathrm{T}~ \bar{\mathbf{c}}_1^\mathrm{T} ~ \ldots ~ \bar{\mathbf{c}}_{n-1}^\mathrm{T}]$ denote the codeword of $\mathcal{C}_2$. %
First compute the following row vector $\bar{\mathbf{m}}_s \in \mathbb{F}_2^L$ for every $0\leq s \leq k-1$,
\begin{equation}
\label{eqn:step 3}
\bar{\mathbf{m}}_s = \mathbf{m}_s\mathbf{G}_L(\mathbf{I}_L+\mathbf{C}_L).
\end{equation}
Thus, for every $0\leq i \leq n-1$, the column $\bar{\mathbf{c}}_i^\mathrm{T}$ in $\Delta(\bar{\mathbf{c}})$ can be computed as follows
\begin{equation}
\label{eqn:step 4}
\bar{\mathbf{c}}_i^\mathrm{T} = \left(\sum\nolimits_{s=0}^{k-1}\bar{\mathbf{m}}_s\mathbf{C}_L^{2^sl_{i}}\mathbf{G}_L^\mathrm{T}\right)^\mathrm{T}.
\end{equation}
It can be easily checked that it takes $L$ and $(k-1)(L-1)$ XOR operations to compute \eqref{eqn:step 3} and \eqref{eqn:step 4}, respectively. %
In all, it takes $kL+n(k-1)(L-1)=nkL-(k-1)n-(n-k)L$ XOR operations to generate a codeword of $\mathcal{C}_2$. %
It turns out that when $k > 1$, the computational complexity to generate a codeword in $\mathcal{C}_2$ is slightly lower than that to generate a codeword in $\mathcal{C}_1$. %

We next analyze the computational complexity of generating a codeword ({\emph{i.e.}, a matrix in $\mathbb{F}_{2}^{(L-1)\times n}$) of an $((L-1) \times n, 2^{(L-1)k}, d)$ Gabidulin MRD code over $\mathbb{F}_2$, which is equivalent to generate a codeword ({\emph{i.e.}, a row vector in $\mathbb{F}_{2^{L-1}}^n$) of an $(n, k)$ Gabidulin code $\mathcal{V}$ over $\mathbb{F}_{2^{L-1}}$. %
Given a nonzero row vector $\mathbf{u}=[u_0, u_1, \ldots, u_{k-1}] \in \mathbb{F}_{2^{L-1}}^k$, let $\mathbf{v}=[v_0,v_1,\ldots,v_{n-1}]$ denote the codeword of $\mathcal{V}$. %
For every $0\leq i \leq n-1$, the element $v_i$ in $\mathbf{v}$ can be computed as follows
\begin{equation}
v_i=\sum\nolimits_{s=0}^{k-1} u_s{\beta_i}^{2^s}~\forall 0\leq i \leq n-1,
\end{equation}
where $\beta_0, \beta_1, \ldots, \beta_{n-1}$ are $n$ $\mathbb{F}_2$-linearly independent elements in $\mathbb{F}_{2^{L-1}}$. %
Addition of two elements in $\mathbb{F}_{2^{L-1}}$ takes $L-1$ XOR operations. It turns out that it requires $n(k-1)(L-1)$ XOR operations and $nk$ multiplications in $\mathbb{F}_{2^{L-1}}$ to obtain $\mathbf{v}$. %
Herein, similar to \cite{Hou_Basic} and \cite{Tang_LNC_TIT}, as a fundamental comparison of the encoding complexity between circular-shift-based linear codes and linear codes over $\mathbb{F}_{2^{L-1}}$, we only consider the generic multiplication in $\mathbb{F}_{2^{L-1}}$ by multiplying two $\mathbb{F}_2$-polynomials modulo an irreducible $\mathbb{F}_2$-polynomial, which takes $O(L^2)$ XOR operations. %
For a more detailed justification on such consideration, please refer to Section VI-A of \cite{Hou_Basic}. %
In all, it takes $n(k-1)(L-1) + nkO(L^2) = nkO(L^2)$ XOR operations to generate a codeword of the $(n, k)$ Gabidulin code over $\mathbb{F}_{2^{L-1}}$. %
Notice that in the above analysis, we skipped the step to transform a codeword in $\mathbb{F}_{2^{L-1}}^n$ to the $(L-1)\times n$ binary matrix codeword with respect to a defined basis, and the computational complexity of this step depends on the particular selection of the basis.

Finally, for $m_L < L-1$, we analyze the computational complexity of generating a codeword ({\emph{i.e.}, a matrix in $\mathbb{F}_{2}^{(L-1)\times n}$) of the $((L-1) \times n, 2^{(L-1)k}, d)$ rank-metric code $\widetilde{\mathcal{M}}$ formulated in \eqref{eqn:generalization of Gabidulin codes}. %
Let $\widetilde{\mathbf{M}}$ denote a codeword of $\widetilde{\mathcal{M}}$, which is given by
\begin{equation}
\widetilde{\mathbf{M}}=
\begin{bmatrix}
\sum\nolimits_{s=1}^{(L-1)/{m_L}}\mathbf{M}_{1,s}\\
\sum\nolimits_{s=1}^{(L-1)/{m_L}}\mathbf{M}_{2,s}\\
\vdots\\
\sum\nolimits_{s=1}^{(L-1)/{m_L}}\mathbf{M}_{(L-1)/{m_L},s}\\
\end{bmatrix},
\end{equation}
where each $\mathbf{M}_{i,s}$, $1 \leq i, s\leq (L-1)/{m_L}$, denotes a codeword of an $(m_L\times n, q^{m_Lk}, d)$ Gabidulin code $\mathcal{M}_{i,s}$. %
The process of generating $\widetilde{\mathbf{M}}$ consists of two steps: 1) for $1 \leq i, s\leq (L-1)/{m_L}$, generate the codeword $\mathbf{M}_{i,s}$ ({\emph{i.e.}, a matrix in $\mathbb{F}_{2}^{m_L\times n}$) of $\mathcal{M}_{i,s}$; 2) sum every group of $\frac{L-1}{m_L}$ codewords  $\mathbf{M}_{i,s}$, $1 \leq s \leq \frac{L-1}{m_L}$. %
By a discussion similar to the previous paragraph, generating $\mathbf{M}_{i,s}$ requires $n(k-1)m_L$ XOR operations and $nk$ multiplications in $\mathbb{F}_{2^{m_L}}$, which implies that step $1)$ requires $n(k-1)\frac{(L-1)^2}{m_L}$ XOR operations and $nk(\frac{L-1}{m_L})^2$ multiplications in $\mathbb{F}_{2^{m_L}}$. %
Since the addition of two row vectors in $\mathbb{F}_{2^{m_L}}^n$ takes $nm_L$ XOR operations, step $2)$ requires $nm_L\frac{L-1}{m_L}(\frac{L-1}{m_L}-1)=n\frac{(L-1)^2}{m_L}-n(L-1)$ XOR operations. %
In all, generating a codeword of $\widetilde{\mathcal{M}}$ requires $nk\frac{(L-1)^2}{m_L}-n(L-1)$ XOR operations and $nk(\frac{L-1}{m_L})^2$ multiplications in $\mathbb{F}_{2^{m_L}}$. %
Same as the discussion in the previous paragraph, we assume a multiplication in $\mathbb{F}_{2^{m_L}}$ takes $O(m_L^2)$ XOR operations. Then, generating a codeword of $\widetilde{\mathcal{M}}$ requires $nk\frac{(L-1)^2}{m_L}-n(L-1)+nk(\frac{L-1}{m_L})^2O(m_L^2) = O(nkL^2)$ XOR operations. %
 
In summary, Table \ref{table:XOR_number} compares the computational complexity to generate a codeword of different $((L-1) \times n, 2^{(L-1)k}, d)$ rank-metric codes. It can be observed from Table \ref{table:XOR_number} that no multiplication in $\mathbb{F}_{2^{L-1}}$ or $\mathbb{F}_{2^{m_L}}$ is required in the construction of circular-shift-based MRD codes $\mathcal{C}_1$ and $\mathcal{C}_2$. %
Generating a codeword of an $((L-1) \times n, 2^{(L-1)k}, d)$ circular-shift-based MRD code requires $O(nkL)$ XOR operations, while generating a codeword of an $((L-1) \times n, 2^{(L-1)k}, d)$ Gabidulin code, based on the customary construction of an $(n, k)$ Gabidulin (linear) code over $\mathbb{F}_{2^{L-1}}$, requires $O(nkL^2)$ XOR operations. %
Notice that when $L-1 = m_L$, the $((L-1) \times n, 2^{(L-1)k}, d)$ rank-metric code $\widetilde{\mathcal{M}}$ degenerates to an $((L-1) \times n, 2^{(L-1)k}, d)$ Gabidulin code. In this case, the computational complexity to generate a codeword of $\widetilde{\mathcal{M}}$ and of an $((L-1) \times n, 2^{(L-1)k}, d)$ Gabidulin code is same, consistent with the analysis presented in Table \ref{table:XOR_number}.

\begin{table}[h]
\center
\caption{Computational complexity analysis of constructing a codeword (a matrix in $\mathbb{F}_{2}^{(L-1)\times n}$) of different $((L-1) \times n, 2^{(L-1)k}, d)$ rank-metric codes, where $L$ is prime and $n \leq m_L$. The total complexity is obtained based on the assumption that a multiplication in $\mathbb{F}_{2^N}$ takes $O(N^2)$ XOR operations, which has also been adopted in \cite{Hou_Basic} and \cite{Tang_LNC_TIT}.}
\label{table:XOR_number}
\renewcommand{\arraystretch}{2}  
\begin{tabular}{|c|c|c|c|}
 \hline
  Type of rank-metric codes& \ \# XOR operations &  \# Multiplications &  \makecell*[c] {Total complexity \\ (in \# XOR operations)}  \\
  \hline
  Gabidulin code reviewed in Definition \ref{def:Gabidulin codes}& $n(k-1)(L-1)$ & $nk$ multiplications in $\mathbb{F}_{2^{L-1}}$ & $O(nkL^2)$\\
  \hline
  $\widetilde{\mathbf{M}}$ introduced in \eqref{eqn:generalization of Gabidulin codes} & $nk\frac{(L-1)^2}{m_L}-n(L-1)$ &$nk(\frac{L-1}{m_L})^2$ multiplications in $\mathbb{F}_{2^{m_L}}$& $O(nkL^2)$\\
  \hline
  $\mathcal{C}_1$ introduced in Section \ref{sec:Circular-shift-based MRD codes} &$nkL-n$& 0 &  $O(nkL)$\\
  \hline
   $\mathcal{C}_2$ introduced in Section \ref{sec:Circular-shift-based MRD codes} &$nkL-(k-1)n-(n-k)L$& 0& $O(nkL)$\\
  \hline
\end{tabular}
\end{table}

\section{Concluding Remarks}
\label{section:Concluding Remarks}
In this paper, we introduced a construction of $(J \times n, q^{Jk}, d)$ MRD codes based on circular-shift operations with efficient encoding, where $J$ equals to the Euler's totient function of a defined $L$ subject to $\gcd(q, L) = 1$. %
Different from conventional construction of $(J \times n, q^{Jk}, d)$ Gabidulin codes, which relies on the arithmetic of $\mathbb{F}_{q^J}$, the codewords of the proposed MRD codes are directly generated based on the arithmetic of $\mathbb{F}_{q}$. %
To clarify the inherent difference and connection between the proposed MRD codes and Gabidulin codes, we equivalently characterized them based on a set of $q$-linearized polynomials over the row vector space $\mathbb{F}_{q}^N$ and showed that their polynomial evaluations are different. %
For the case $J \neq m_L$, where $m_L$ denotes the multiplicative order of $q$ modulo $L$, we proved that the proposed MRD codes, in a family of settings, are different from any Gabidulin code and any twisted Gabidulin code with $h = 0$. %
We also provided explicit examples demonstrating that outside these settings, the proposed MRD codes may either coincide with or differ from $(J \times n, q^{Jk}, d)$ Gabidulin codes. %
In the special case $J = m_L$, even though a $(J \times n, q^{Jk}, d)$ circular-shift-based MRD code coincides with a $(J \times n, q^{Jk}, d)$ Gabidulin code, the proposed construction has its own merit because it provides a new approach to construct Gabidulin codes purely based on the arithmetic of $\mathbb{F}_{q}$ and avoids the arithmetic of $\mathbb{F}_{q^{J}}$. %
Moreover, we proved that the proposed MRD codes are equivalent to a generalization of Gabidulin codes in the following form
\begin{equation}
\widetilde{\mathcal{M}}=
\begin{bmatrix}
\sum\nolimits_{s=1}^{J/{m_L}}\mathcal{M}_{1,s}\\
\sum\nolimits_{s=1}^{J/{m_L}}\mathcal{M}_{2,s}\\
\vdots\\
\sum\nolimits_{s=1}^{J/{m_L}}\mathcal{M}_{J/{m_L},s}\\
\end{bmatrix}, 
\end{equation}
where each $\mathcal{M}_{i,s}$, $1 \leq i, s\leq J/{m_L}$, denotes an $(m_L\times n, q^{m_Lk}, d)$ Gabidulin code. %
When $q=2$, $L$ is prime and $n\leq m_L$, it is analyzed that generating a codeword of the proposed $((L-1) \times n, 2^{(L-1)k}, d)$ MRD codes requires $O(nkL)$ XOR operations, while generating a codeword of $((L-1) \times n, 2^{(L-1)k}, d)$ Gabidulin codes, based on customary construction, requires $O(nkL^2)$ XOR operations. %

As potential future work, it is interesting to characterize the relationship between circular-shift-based MRD codes with parameters not satisfying \eqref{eqn:lj setting} and Gabidulin codes, and to establish relationship between the proposed MRD codes and twisted Gabidulin codes with $q > 2$, $\eta \neq 0$, $h \geq 1$. %
In addition, the decoding for the proposed MRD codes and the design of efficient decoding algorithms lie beyond the scope of the present study and are left for future research. %
Exploring potential applications of the proposed MRD codes in subspace coding is also a interesting direction for future research.

\appendix
\subsection{Proof of Proposition \ref{prop:Connections between M1 and M2}}
\label{appendix:proof proposition M1M2}
Given an arbitrary nonzero row vector $\mathbf{u} \in\mathbb{F}_{q^N}^k$, let  $\mathbf{M}(\mathbf{u})$ denote the codeword of $\mathcal{M} = \{\mathbf{M}(\mathbf{u}): \mathbf{u}\in \mathbb{F}_{q^N}^k\}$ in Definition \ref{def:matrix  representation}, and let $\mathbf{v}$ denote the codeword of $\mathcal{V} = \{\mathbf{u}\mathbf{G}: \mathbf{u}\in \mathbb{F}_{q^N}^k\}$ in Definition \ref{def:Gabidulin codes}. %
Based on \eqref{eqn:counterpart matrix eqn}, in order to prove \eqref{eqn:M1 and M2}, it is equivalent to prove that for any $\mathbf{u} \in \mathbb{F}_{q^N}^k$, there is a unique $\mathbf{v} \in \mathcal{V}$ such that
\begin{equation}
\label{eqn:BV}
\mathcal{B}'\mathbf{M}(\mathbf{u})=\mathbf{v}.
\end{equation} 
In terms of the $q$-linearized polynomial $L_\mathbf{u}(x)$ defined in \eqref{eqn:L_u(x)}, $\mathbf{v}$ can be represented as
\begin{equation}
\mathbf{v}=\mathbf{u}\mathbf{G}=[L_\mathbf{u}(\beta_0)~L_\mathbf{u}(\beta_1)~\ldots~L_\mathbf{u}(\beta_{n-1})].
\end{equation} %
Recall that $\mathbf{M}(\mathbf{u})=\mathbf{M}_\mathbf{o}(\mathcal{B})\mathbf{L}_\mathbf{u}$ as given in \eqref{eqn:Mu}. %
Consequently, 
\begin{equation}
\label{proof proposition M1M2 1}
\mathcal{B}'\mathbf{M}(\mathbf{u})=\mathcal{B}'\mathbf{M}_\mathbf{o}(\mathcal{B})\mathbf{L}_\mathbf{u}.
\end{equation}
By making use of \eqref{def:dual basis}, we have
\begin{equation}
\mathcal{B}'\mathbf{M}_\mathbf{o}(\mathcal{B})=[1~0~\ldots~0].
\end{equation}
Subsequently, \eqref{proof proposition M1M2 1} can be expressed as
\begin{equation}
\mathcal{B}'\mathbf{M}(\mathbf{u})=[1~0~\ldots~0]\mathbf{L}_\mathbf{u}=[L_\mathbf{u}(\beta_0)~L_\mathbf{u}(\beta_1)~\ldots~L_\mathbf{u}(\beta_{n-1})].
\end{equation}
Thus, we can obtain $\mathcal{B}'\mathbf{M}(\mathbf{u})=\mathbf{v}$, which completes the proof of \eqref{eqn:M1 and M2}.

\subsection{Proof of Lemma \ref{lemma:C2=TC1}}
\label{appendix:C2=TC1}
Observe that the set $\{L-j: j \in \mathcal{J}\}$ is equivalent to $\mathcal{J}$. %
Since $\tau(\beta^j) \neq 0$ for all $j \in \mathcal{J}$, and 
both $[\mathbf{u}_{L-j}]_{j \in \mathcal{J}}$ and $[\mathbf{u}_{j}]_{j \in \mathcal{J}}$ are full rank, the $J\times J$ matrix $\mathbf{T}$ is full rank. %

Following a similar approach to that used in proving \eqref{eqn:GH_mu_C}, we next show that
\begin{equation}
\label{eqn:GG_mu_C}
\mathbf{u}_{L-j} \beta^{(L-j)l}\tau(\beta^{L-j})=\mathbf{G}_L\mathbf{C}_L^l\tau(\mathbf{C}_L)\mathbf{G}_L^\mathrm{T}\mathbf{u}_j', \quad \forall 0\leq l \leq L-1,~j \in \mathcal{J},
\end{equation}
where $\mathbf{u}_{L-j}$ and $\mathbf{u}_j'$ denote the $(L-j+1)^{st}$ column of $\mathbf{U}$ and the $(j+1)^{st}$ column of $\mathbf{U}'$, respectively. %
By making use of \eqref{eqn:GH_def} and \eqref{eqn:VL VL'}, Eq. \eqref{eqn:GG_mu_C} can be expressed as
\begin{align}
\label{eqn:GG step 2}
\mathbf{u}_{L-j} \beta^{(L-j)l}\tau(\beta^{L-j})=\mathbf{U}\mathrm{Diag}([\beta^{lj}]_{0 \leq j \leq L-1})\mathrm{Diag}([\tau(\beta^{j})]_{0 \leq j \leq L-1})\widetilde{\mathbf{V}}_L^2\mathbf{U}^\mathrm{T}\mathbf{u}_j'.
\end{align}
According to \cite{Jin-arXiv} and \cite{tang2020circular}, we have
$\widetilde{\mathbf{V}}_L^2=L_{\mathbb{F}_{q}}\begin{bmatrix} 1&\mathbf{0}\\\mathbf{0}&\mathbf{A}_{L-1}\end{bmatrix}$,
where $\mathbf{A}_{L-1}$ denotes the $(L-1)\times(L-1)$ anti-diagonal identity matrix. %
Consequently, the right-hand side of \eqref{eqn:GG step 2} can be written as
\begin{align*}
\mathbf{U}\mathrm{Diag}([\beta^{lj}]_{0 \leq j \leq L-1})\mathrm{Diag}([\tau(\beta^{j})]_{0 \leq j \leq L-1})\widetilde{\mathbf{V}}_L^2\mathbf{U}^\mathrm{T}\mathbf{u}_j'=L_{\mathbb{F}_{q}}\sum\nolimits_{i\in \mathcal{J}}\beta^{il}\tau(\beta^i)\mathbf{u}_i\mathbf{u}_{L-i}^\mathrm{T}\mathbf{u}_j'=\mathbf{u}_{L-j} \beta^{(L-j)l}\tau(\beta^{L-j}),
\end{align*}
where the second equality holds because for $j_1, j_2 \in \mathcal{J}$, ${\mathbf{u}}_{j_1}^\mathrm{T}\mathbf{u}'_{j_2}=L_{\mathbb{F}_{q}}^{-1}$ when $j_1=j_2$, and ${\mathbf{u}}_{j_1}^\mathrm{T}\mathbf{u}'_{j_2}=0$ otherwise. %
Therefore, \eqref{eqn:GG_mu_C} is proved.

According to \eqref{eqn:GH_mu_C} and \eqref{eqn:GG_mu_C}, we have
\begin{align}
\mathbf{G}_L\mathbf{C}_L^l\tau(\mathbf{C}_L)\mathbf{G}_L^\mathrm{T}[\mathbf{u}_{L-j}']_{j \in \mathcal{J}}=&[\mathbf{u}_{j}]_{j \in \mathcal{J}}
\mathrm{Diag}([\beta^{jl}\tau(\beta^{j})]_{j\in \mathcal{J}}),\\
\label{eqn:GCHU}
\mathbf{G}_L\mathbf{C}_L^l\mathbf{H}_L[\mathbf{u}_j]_{j \in \mathcal{J}}=&[\mathbf{u}_{j}]_{j \in \mathcal{J}}\mathrm{Diag}([\beta^{jl}]_{j\in \mathcal{J}}).
\end{align}
Since $\{L-j: j \in \mathcal{J}\}=\mathcal{J}$ and $[\mathbf{u}_{j}']_{j \in \mathcal{J}}[\mathbf{u}_{j}]_{j \in \mathcal{J}}^\mathrm{T}=L_{\mathbb{F}_{q}}^{-1}\mathbf{I}_J$, we have $[\mathbf{u}_{L-j}']_{j \in \mathcal{J}}[\mathbf{u}_{L-j}]_{j \in \mathcal{J}}^\mathrm{T}=L_{\mathbb{F}_{q}}^{-1}\mathbf{I}_J$. It follows that
\begin{equation}
\label{eqn:GCG}
\mathbf{G}_L\mathbf{C}_L^l\tau(\mathbf{C}_L)\mathbf{G}_L^\mathrm{T}=L_{\mathbb{F}_{q}}[\mathbf{u}_{j}]_{j \in \mathcal{J}}
\mathrm{Diag}([\beta^{jl}\tau(\beta^{j})]_{j\in \mathcal{J}})[\mathbf{u}_{L-j}]_{j \in \mathcal{J}}^\mathrm{T}.
\end{equation}
By using  \eqref{eqn:GCHU} and \eqref{eqn:GCG}, we obtain
\begin{equation}
\mathbf{G}_L\mathbf{C}_L^l\tau(\mathbf{C}_L)\mathbf{G}_L^\mathrm{T}=
\mathbf{G}_L\mathbf{C}_L^l\mathbf{H}_L\mathbf{T}^\mathrm{T},\quad \forall 0\leq l \leq L-1.
\end{equation}
This completes the proof.

\subsection{Proof of Proposition \ref{prop:new charac of Gabidulin}}
\label{appendix:prop:new charac of Gabidulin}
With respect to $\mathcal{B}$, the $(N\times n, q^{Nk}, d)$ Gabidulin MRD code $\mathcal{M}_{\mathcal{B}}(\mathcal{V})$ can be expressed as 
\begin{equation}
\label{eqn:proposition MBV}
\mathcal{M}_{\mathcal{B}}(\mathcal{V})=\{[\mathbf{v}_{\mathcal{B}}(L(\beta_0))^\mathrm{T}~\mathbf{v}_{\mathcal{B}}(L(\beta_1))^\mathrm{T}~\ldots~\mathbf{v}_{\mathcal{B}}(L(\beta_{n-1}))^\mathrm{T}]: L(x) \in \mathcal{L}\},
\end{equation}
where $\mathbf{v}_{\mathcal{B}}(L(\beta_i))$ represents the $q$-ary representation of $L(\beta_i) \in \mathbb{F}_{q^N}$ with respect to the basis $\mathcal{B}$. %
Recall that the set $\{\mathbf{A}(\sum\nolimits_{i=0}^{N-1} a_i\gamma^i): a_0, \ldots, a_{N-1} \in \mathbb{F}_q \}$ forms a matrix representation of $\mathbb{F}_{q^N}$ over $\mathbb{F}_q$ and $\mathcal{B}_{\gamma}=[\gamma^{N-1}~\gamma^{N-2}~\ldots~\gamma~1]$ is the basis of $\mathbb{F}_{q^N}$ over $\mathbb{F}_q$. For every $0 \leq t \leq n-1$, we have
\begin{equation}
\label{eqn:propsition vB}
\begin{split}
\mathbf{v}_{\mathcal{B}}(L(\beta_t))
=&\mathbf{v}_{\mathcal{B}_{\gamma}}(L(\beta_t))\mathbf{V}^{-1}
=\mathbf{v}_{\mathcal{B}_{\gamma}}(\sum\nolimits_{i=0}^{k-1}u_i\beta_t^{q^i})\mathbf{V}^{-1}
=\sum\nolimits_{i=0}^{k-1}\mathbf{v}_{\mathcal{B}_{\gamma}}(u_i)\mathbf{A}(\beta_t)^{q^i}\mathbf{V}^{-1}\\
=&\sum\nolimits_{i=0}^{k-1}\mathbf{v}_{\mathcal{B}}(u_i)\mathbf{V}\mathbf{A}(\beta_t)^{q^i}\mathbf{V}^{-1}
=L^\dag (\mathbf{V}\mathbf{A}(\beta_t)\mathbf{V}^{-1}),
\end{split}
\end{equation}
where $L(x) \in \mathcal{L}$ and $L^\dag (x) \in \mathcal{L}_N^\dag$. %
Combining \eqref{eqn:proposition MBV} and \eqref{eqn:propsition vB} yields \eqref{expression Gabidulin}.

\subsection{Proof of Lemma \ref{GCH neq VAV}}
\label{appendix:proof lemma:GCH neq VAV}
Assume $J \neq m_L$. Define $\mathcal{A}_1=\{\mathbf{G}\mathbf{C}_L^{l}\mathbf{H}:l  \in \mathcal{J}\}$, where each $J\times J$ matrix $\mathbf{G}\mathbf{C}_L^{l}\mathbf{H}$ satisfies $(\mathbf{G}\mathbf{C}_L^{l}\mathbf{H})^L=\mathbf{I}_J$. %
Define $\mathcal{A}_2=\{\mathbf{A}(\alpha^{l(q^J-1)/L}): l \in \mathcal{J} \}$, where $\alpha$ denotes a primitive element of $\mathbb{F}_{q^J}$. Thus, for $s \in \mathbb{F}_{q^J}$, $\mathbf{A}(s)^L = \mathbf{I}_J$ if and only if $\mathbf{A}(s) \in \mathcal{A}_2$. %
In addition, for every $\mathbf{A}(s) \in \mathcal{A}_2$, there is an irreducible polynomial $f(x)$ of degree $m_L$ over $\mathbb{F}_q$ such that $f(\mathbf{A}(s)) = \mathbf{0}$. %
It suffices to prove that $\mathbf{G}\mathbf{C}_L^l\mathbf{H}$ cannot be represented in the form $\mathbf{V}\mathbf{A}(s)\mathbf{V}^{-1}$, regardless of the choices of invertible matrix $\mathbf{V} \in \mathbb{F}_q^{J\times J}$ and $\mathbf{A}(s)\in\mathcal{A}_2$. %

Assume $\mathbf{G}\mathbf{C}_L^l\mathbf{H} = \mathbf{V}\mathbf{A}(s)\mathbf{V}^{-1}$ for some $\mathbf{A}(s)\in\mathcal{A}_2$. Denote by $f(x)$ the irreducible polynomial over $\mathbb{F}_q$ of degree $m_L$ with $f(\mathbf{A}(s)) = \mathbf{0}$. Since $f(\mathbf{V}\mathbf{A}(s)\mathbf{V}^{-1}) = \mathbf{V}f(\mathbf{A}(s))\mathbf{V}^{-1} = \mathbf{0}$, $f(\mathbf{G}\mathbf{C}_L^l\mathbf{H})= \mathbf{0}$. 
Since 
\begin{equation}
f(\mathbf{G}\mathbf{C}_L^l\mathbf{H})=L_{\mathbb{F}_q}
[\mathbf{u}_j]_{j \in \mathcal{J}}\mathrm{Diag}([f(\beta^{lj})]_{j\in \mathcal{J}})[\mathbf{u}'_j]_{j \in \mathcal{J}}^\mathrm{T}, 
\end{equation}
$f(\mathbf{G}\mathbf{C}_L^l\mathbf{H})= \mathbf{0}$ implies that $f(\beta^{lj})=0$ for all $j \in \mathcal{J}$. %
However, as $\{lj: j \in \mathcal{J}\}=\mathcal{J}$ for every $l \in \mathcal{J}$, $\beta^{lj_1}, \beta^{lj_2}, \ldots, \beta^{lj_J}$ are $J>m_L$ distinct roots of $f(x)$, which contradicts the fact that the degree of $f(x)$ is $m_L$. %
This contradiction completes the proof. 

\subsection{Proof of Theorem \ref{theorem:matrix form of C}}
\label{appendix:Proof of theorem}
Given a nonzero row vector $\mathbf{m}=[\mathbf{m}_0, \mathbf{m}_1, \ldots, \mathbf{m}_{k-1}]$ over $\mathbb{F}_{q}$, where $\mathbf{m}_s \in \mathbb{F}_{q}^J$ for $0 \leq s \leq k-1$, let $\Delta(\mathbf{c})$ and $\Delta(\bar{\mathbf{c}})$ respectively denote the codewords of $\mathcal{C}_1$ and $\mathcal{C}_2$, which are given by
\begin{equation}
\label{Justification of c_step_1}
\Delta(\mathbf{c})=\begin{bmatrix}
\sum_{s=0}^{k-1}\mathbf{m}_s\mathbf{G}_L\mathbf{C}_L^{q^sl_{0}}\mathbf{H}_L\\
\sum_{s=0}^{k-1}\mathbf{m}_s\mathbf{G}_L\mathbf{C}_L^{q^sl_{1}}\mathbf{H}_L\\
\vdots\\
\sum_{s=0}^{k-1}\mathbf{m}_s\mathbf{G}_L\mathbf{C}_L^{q^sl_{n-1}}\mathbf{H}_L\\
\end{bmatrix}^\mathrm{T}, \quad
\Delta(\bar{\mathbf{c}})=\begin{bmatrix}
\sum_{s=0}^{k-1}\mathbf{m}_s\mathbf{G}_L\mathbf{C}_L^{q^sl_{0}}\tau(\mathbf{C}_L)\mathbf{G}_L^\mathrm{T}\\
\sum_{s=0}^{k-1}\mathbf{m}_s\mathbf{G}_L\mathbf{C}_L^{q^sl_{1}}\tau(\mathbf{C}_L)\mathbf{G}_L^\mathrm{T}\\
\vdots\\
\sum_{s=0}^{k-1}\mathbf{m}_s\mathbf{G}_L\mathbf{C}_L^{q^sl_{n-1}}\tau(\mathbf{C}_L)\mathbf{G}_L^\mathrm{T}\\
\end{bmatrix}^\mathrm{T}.
\end{equation}

By making use of \eqref{eqn:GH_mu_C}, we have
\begin{equation}
\label{c_GH_step_1}
\Delta(\mathbf{c})^\mathrm{T}[\mathbf{u}_j]_{j\in {\cal J}}
=\begin{bmatrix}
\sum_{s=0}^{k-1}\mathbf{m}_s\mathbf{u}_{d_1}\beta^{d_1q^sl_{0}} & \sum_{s=0}^{k-1}\mathbf{m}_s\mathbf{u}_{d_2}\beta^{d_2q^sl_{0}} &\ldots&\sum_{s=0}^{k-1}\mathbf{m}_s\mathbf{u}_{d_J}\beta^{d_Jq^sl_{0}}\\
\sum_{s=0}^{k-1}\mathbf{m}_s\mathbf{u}_{d_1}\beta^{d_1q^sl_{1}} & \sum_{s=0}^{k-1}\mathbf{m}_s\mathbf{u}_{d_2}\beta^{d_2q^sl_{1}} &\ldots&\sum_{s=0}^{k-1}\mathbf{m}_s\mathbf{u}_{d_J}\beta^{d_Jq^sl_{1}}\\
\vdots&\vdots&\ldots&\vdots\\
\sum_{s=0}^{k-1}\mathbf{m}_s\mathbf{u}_{d_1}\beta^{d_1q^sl_{n-1}} & \sum_{s=0}^{k-1}\mathbf{m}_s\mathbf{u}_{d_2}\beta^{d_2q^sl_{n-1}} &\ldots&\sum_{s=0}^{k-1}\mathbf{m}_s\mathbf{u}_{d_J}\beta^{d_Jq^sl_{n-1}}
\end{bmatrix},
\end{equation}
where $d_1< d_2 <\ldots < d_J$ with  $d_1, d_2,\ldots, d_J \in \mathcal{J}$. %
Since $[\mathbf{u}_j]_{j\in {\cal J}}^{-1}=L_{\mathbb{F}_{q}}[\mathbf{u}'_j]_{j\in \mathcal{ J}}^\mathrm{T}$, \eqref{c_GH_step_1} can be expressed as
\begin{equation*}
\label{c_GH_step_2}
\begin{split}
\Delta(\mathbf{c})^\mathrm{T}=&
L_{\mathbb{F}_{q}}
\begin{bmatrix}
\sum_{s=0}^{k-1}\mathbf{m}_s\mathbf{u}_{d_1}\beta^{d_1q^sl_{0}} & \sum_{s=0}^{k-1}\mathbf{m}_s\mathbf{u}_{d_2}\beta^{d_2q^sl_{0}} &\ldots&\sum_{s=0}^{k-1}\mathbf{m}_s\mathbf{u}_{d_J}\beta^{d_Jq^sl_{0}}\\
\sum_{s=0}^{k-1}\mathbf{m}_s\mathbf{u}_{d_1}\beta^{d_1q^sl_{1}} & \sum_{s=0}^{k-1}\mathbf{m}_s\mathbf{u}_{d_2}\beta^{d_2q^sl_{1}} &\ldots&\sum_{s=0}^{k-1}\mathbf{m}_s\mathbf{u}_{d_J}\beta^{d_Jq^sl_{1}}\\
\vdots&\vdots&\ldots&\vdots\\
\sum_{s=0}^{k-1}\mathbf{m}_s\mathbf{u}_{d_1}\beta^{d_1q^sl_{n-1}} & \sum_{s=0}^{k-1}\mathbf{m}_s\mathbf{u}_{d_2}\beta^{d_2q^sl_{n-1}} &\ldots&\sum_{s=0}^{k-1}\mathbf{m}_s\mathbf{u}_{d_J}\beta^{d_Jq^sl_{n-1}}
\end{bmatrix}[\mathbf{u}'_j]_{j\in \mathcal{ J}}^\mathrm{T}\\
&=
\begin{bmatrix}
L_{\mathbf{t}_{d_1}}(\beta^{d_1l_0})& L_{\mathbf{t}_{d_2}}(\beta^{d_2l_0}) &\ldots&L_{\mathbf{t}_{d_J}}(\beta^{d_Jl_0})\\
L_{\mathbf{t}_{d_1}}(\beta^{d_1l_1})& L_{\mathbf{t}_{d_2}}(\beta^{d_2l_1}) &\ldots&L_{\mathbf{t}_{d_J}}(\beta^{d_Jl_1})\\
\vdots&\vdots&\ldots&\vdots\\
L_{\mathbf{t}_{d_1}}(\beta^{d_1l_{n-1}})& L_{\mathbf{t}_{d_2}}(\beta^{d_2l_{n-1}}) &\ldots&L_{\mathbf{t}_{d_J}}(\beta^{d_Jl_{n-1}})\\
\end{bmatrix}
[\mathbf{u}'_j]_{j\in \mathcal{ J}}^\mathrm{T},
\end{split}
\end{equation*}
where $L_{\mathbf{t}_{d_j}}(x)$ is defined in \eqref{eqn:Ltj(x)}. %
Consequently, we have
\begin{equation*}
\Delta(\mathbf{c})
=[\mathbf{u}'_j]_{j\in \mathcal{ J}}
\begin{bmatrix}
L_{\mathbf{t}_{d_1}}(\beta^{d_1l_0})& L_{\mathbf{t}_{d_1}}(\beta^{d_1l_1}) &\ldots&L_{\mathbf{t}_{d_1}}(\beta^{d_1l_{n-1}})\\
L_{\mathbf{t}_{d_2}}(\beta^{d_2l_0})& L_{\mathbf{t}_{d_2}}(\beta^{d_2l_1}) &\ldots&L_{\mathbf{t}_{d_2}}(\beta^{d_2l_{n-1}})\\
\vdots&\vdots&\ldots&\vdots\\
L_{\mathbf{t}_{d_J}}(\beta^{d_Jl_0})& L_{\mathbf{t}_{d_J}}(\beta^{d_Jl_1})&\ldots&L_{\mathbf{t}_{d_J}}(\beta^{d_Jl_{n-1}})\\
\end{bmatrix},
\end{equation*}
which implies that \eqref{matrix C1} is established. %

Using a similar technique as in the preceding discussion and \eqref{eqn:GG_mu_C} in Appendix-\ref{appendix:C2=TC1}, we obtain
\begin{equation*}
\Delta(\bar{\mathbf{c}})
=[\mathbf{u}_{L-j}]_{j\in {\cal J}}
\begin{bmatrix}
L_{\mathbf{t}_{d_1}}(\beta^{d_1l_0})& L_{\mathbf{t}_{d_1}}(\beta^{d_1l_1}) &\ldots&L_{\mathbf{t}_{d_1}}(\beta^{d_1l_{n-1}})\\
L_{\mathbf{t}_{d_2}}(\beta^{d_2l_0})& L_{\mathbf{t}_{d_2}}(\beta^{d_2l_1}) &\ldots&L_{\mathbf{t}_{d_2}}(\beta^{d_2l_{n-1}})\\
\vdots&\vdots&\ldots&\vdots\\
L_{\mathbf{t}_{d_J}}(\beta^{d_Jl_0})& L_{\mathbf{t}_{d_J}}(\beta^{d_Jl_1})&\ldots&L_{\mathbf{t}_{d_J}}(\beta^{d_Jl_{n-1}})\\
\end{bmatrix},
\end{equation*}
which implies that \eqref{matrix C2} is established. %

\subsection{Proof of Lemma \ref{lemma: basis}}
\label{appendix:basis}
Let $\mathcal{B}$ denote the row vector $[\beta^{jh}~ \beta^{j(h+1)}~ \ldots ~\beta^{j(h+m_L-1)}] \in \mathbb{F}_{q^{m_L}}^{m_L}$. %
It suffices to show that
\begin{equation}
\label{full rank 1}
\mathrm{rank}(\mathbf{M}_\mathbf{o}(\mathcal{B}))=m_L,
\end{equation}
where $\mathbf{M}_\mathbf{o}(\mathcal{B})$ is a  matrix in $\mathbb{F}_{q^{m_L}}^{m_L \times m_L}$. %
Observe that
\begin{equation}
\mathrm{rank}(\mathbf{M}_\mathbf{o}(\mathcal{B}))=\mathrm{rank}(\mathbf{M}_\mathbf{o}([1~\beta^j~\ldots~\beta^{j(m_L-1)}])).
\end{equation}
Since $\det(\mathbf{M}_\mathbf{o}([1~\beta^j~\ldots~\beta^{j(m_L-1)}]))=\prod\nolimits_{0 \leq j_1 <j_2 \leq m_L-1}(\beta^{jq^{j_2}}-\beta^{jq^{j_1}})$ and $\beta^j, \beta^{jq},\ldots,\beta^{jq^{m_L-1}}$ are $m_L$ distinct elements in $\mathbb{F}_{q^{m_L}}$, it follows that $\mathbf{M}_\mathbf{o}([1~\beta^j~\ldots~\beta^{j(m_L-1)}])$ has full rank $m_L$, implying that \eqref{full rank 1} holds as desired.

\subsection{Proof of Theorem \ref{generalization}}
\label{appendix:generalization}
We first reformulate the structure of $\mathbf{L}_{\mathbf{m}}$ defined in \eqref{eqn:L_t} for the construction of $\mathcal{C}_1$ and $\mathcal{C}_2$. %
Recall that $\varepsilon_s$ denotes a defined integer in $\mathcal{J}_s$ so that all $m_L$ elements in $\mathcal{J}_s$
can be written as $\varepsilon_sq^j~\mathrm{mod}~L$ with $0\leq j \leq m_L-1$. %
Let $\widetilde{\mathbf{v}}_j$ denote the $(j+1)^{st}$ column of $\widetilde{\mathbf{V}}_L$ defined in \eqref{eqn:VL tilde}. %
Since $\mathbf{H}_L$ is over $\mathbb{F}_{q}$, we have
\begin{equation}
\label{connection u'}
 \mathbf{u}'_{\varepsilon_sq^j\mathrm{mod}~ L}=L_{\mathbb{F}_q}^{-1}\mathbf{H}_L^\mathrm{T}\widetilde{\mathbf{v}}_{\varepsilon_sq^j\mathrm{mod}~ L}=L_{\mathbb{F}_q}^{-1}\mathbf{H}_L^\mathrm{T}(\widetilde{\mathbf{v}}_{\varepsilon_s})^{\circ q^j}=(L_{\mathbb{F}_q}^{-1}\mathbf{H}_L^\mathrm{T}\widetilde{\mathbf{v}}_{\varepsilon_s})^{\circ q^j}
 =(\mathbf{u}'_{\varepsilon_s})^{\circ q^j}, \forall 0\leq j \leq m_L-1.
\end{equation}
Analogously, let $\mathbf{v}_{j}$ denote the $(j+1)^{st}$ column of $\mathbf{V}_L$ defined in \eqref{eqn:VL}. %
Since $\mathbf{G}_L$ is over $\mathbb{F}_{q}$, we have
\begin{equation}
\label{connection u}
\mathbf{u}_{\varepsilon_sq^j\mathrm{mod}~ L}=L_{\mathbb{F}_q}^{-1}\mathbf{G}_L\mathbf{v}_{\varepsilon_sq^j\mathrm{mod}~ L}=L_{\mathbb{F}_q}^{-1}\mathbf{G}_L(\mathbf{v}_{\varepsilon_s})^{\circ q^j}=(L_{\mathbb{F}_q}^{-1}\mathbf{G}_L\mathbf{v}_{\varepsilon_s})^{\circ q^j}=(\mathbf{u}_{{\varepsilon_s}})^{\circ q^j},
  ~\forall 0\leq j \leq m_L-1.
\end{equation}
For $1 \leq s \leq J/{m_L}$, let $\bm{\lambda}_s=[\lambda_{0,s}~\lambda_{1,s}~\ldots~\lambda_{k-1,s}] \in \mathbb{F}_{q^{m_L}}^k$ denote the row vector with
\begin{equation*}
\lambda_{i,s}=\mathbf{m}_i\mathbf{u}_{\varepsilon_s}, \quad 0 \leq i \leq k-1,
\end{equation*}
and let $L_{\bm{\lambda}_s}(x)$ denote the following $q$-linearized polynomial over $\mathbb{F}_{q^{m_L}}$ 
\begin{equation*}
L_{\bm{\lambda}_s}(x)=
\begin{cases}
\sum\nolimits_{i=0}^{k-1} \lambda_{i,s}x^{q^i}, \quad \mathrm{for}~\mathcal{C}_1,\\
\sum\nolimits_{i=0}^{k-1} \lambda_{i,s}\tau(\beta^{\varepsilon_s})x^{q^i},\quad \mathrm{for}~\mathcal{C}_2.
\end{cases}
\end{equation*}
Recall that $L_{\mathbf{t}_j}(x)$ used in $\mathbf{L}_{\mathbf{m}}$ is defined in \eqref{eqn:Ltj(x)}. %
For every $d_i \in \mathcal{J}_s$, assume that $d_i=\varepsilon_sq^j$ for $0\leq j \leq m_L-1$. 
For the code $\mathcal{C}_1$, based on \eqref{eqn:Ltj(x)} and \eqref{connection u}, we have 
\begin{equation*}
L_{\mathbf{t}_{d_i}}(\beta ^{d_il_t}) 
= L_{\mathbb{F}_q}\sum\nolimits_{s = 0}^{k - 1} \mathbf{m}_s(\mathbf{u}_{\varepsilon _s})^{^\circ q^j}\beta ^{\varepsilon _sq^jl_tq^s}  
= L_{\mathbb{F}_q}{\left( {\sum\nolimits_{s = 0}^{k - 1} \mathbf{m}_s\mathbf{u}_{\varepsilon _s}\beta ^{\varepsilon _sl_tq^s} } \right)^{q^j}} 
= L_{\mathbb{F}_q}L_{\bm{\lambda}_s}(\beta ^{\varepsilon _sl_t})^{q^j}
\end{equation*}
for $0\leq j \leq m_L-1$ and $0 \leq t \leq n-1$. Similarly, for the  code $\mathcal{C}_2$, we have 
\begin{equation*}
L_{\mathbf{t}_{d_i}}(\beta ^{d_il_t}) 
= L_{\mathbb{F}_q}\sum\nolimits_{s = 0}^{k - 1} \mathbf{m}_s(\mathbf{u}_{\varepsilon _s})^{^\circ q^j}\tau(\beta^j)\beta ^{\varepsilon _sq^jl_tq^s}  
= L_{\mathbb{F}_q}{\left( {\sum\nolimits_{s = 0}^{k - 1} \mathbf{m}_s\mathbf{u}_{\varepsilon _s}\tau(\beta^{\varepsilon _s})\beta ^{\varepsilon _sl_tq^s} } \right)^{q^j}} 
= L_{\mathbb{F}_q}L_{\bm{\lambda}_s}(\beta ^{\varepsilon _sl_t})^{q^j}
\end{equation*}
for $0\leq j \leq m_L-1$ and $0 \leq t \leq n-1$. Consequently, we can obtain that
\begin{equation}
L_{\mathbf{t}_{d_i}}(\beta^{d_il_t})=L_{\mathbb{F}_{q}}L_{\bm{\lambda}_{s}}(\beta^{\varepsilon_sl_t})^{q^j}, \quad 0\leq j \leq m_L-1,~0 \leq t \leq n-1.
\end{equation}
Recall that $d_1<d_2<\ldots<d_J$. Let $\mathbf{d}$ denote the column vector $\mathbf{d}=[d_1~d_2~\ldots~d_J]^\mathrm{T}$, and let $\mathbf{d}'$ denote the column vector $\mathbf{d}'=[\varepsilon_1~\varepsilon_1q~\varepsilon_1q^{m_L-1}~\ldots~\varepsilon_s~\varepsilon_sq~\varepsilon_sq^{m_L-1}]^\mathrm{T} \mod~L$ obtained by rearranging the entries of $\mathbf{d}$. %
Define $\bm{\Gamma}$ as the permutation matrix in $\mathbb{F}_q^{J \times J}$ satisfying $\mathbf{d}=\bm{\Gamma}\mathbf{d}'$. Note that $\bm{\Gamma}^{-1}=\bm{\Gamma}^\mathrm{T}$. %
Consequently, $\mathbf{L}_{\mathbf{m}}$ can be written as
\begin{equation}
\label{eqn:k=1L_m}
\mathbf{L}_{\mathbf{m}}=L_{\mathbb{F}_{q}}\bm{\Gamma}
\begin{bmatrix}
\mathbf{L}_{\bm{\lambda}_{1}}\\
\mathbf{L}_{\bm{\lambda}_{2}}\\
\vdots\\
\mathbf{L}_{\bm{\lambda}_{J/{m_L}}}
\end{bmatrix},
\end{equation}
where the matrix $\mathbf{L}_{\bm{\lambda}_{s}} \in \mathbb{F}_{q^{m_L}}^{m_L \times n}$ is given by
\begin{equation*}
\mathbf{L}_{\bm{\lambda}_{s}}=
\begin{bmatrix}
L_{\bm{\lambda}_{s}}(\beta^{\varepsilon_sl_0}) & L_{\bm{\lambda}_{s}}(\beta^{\varepsilon_sl_1}) &\ldots &L_{\bm{\lambda}_{s}}(\beta^{\varepsilon_sl_{n-1}})\\
L_{\bm{\lambda}_{s}}(\beta^{\varepsilon_sl_0})^q & L_{\bm{\lambda}_{s}}(\beta^{\varepsilon_sl_1})^q &\ldots &L_{\bm{\lambda}_{s}}(\beta^{\varepsilon_sl_{n-1}})^q\\
\vdots & \vdots &\ldots &\vdots\\
L_{\bm{\lambda}_{s}}(\beta^{\varepsilon_sl_0})^{q^{m_L-1}} & L_{\bm{\lambda}_{s}}(\beta^{\varepsilon_sl_1})^{q^{m_L-1}} &\ldots &L_{\bm{\lambda}_{s}}(\beta^{\varepsilon_sl_{n-1}})^{q^{m_L-1}}
\end{bmatrix}.
\end{equation*}
Herein, $\mathcal{F}_{s} = \{\beta^{\varepsilon_sl_0}, \beta^{\varepsilon_sl_1}, \ldots, \beta^{\varepsilon_sl_{n-1}}\}$ denote the set of $n$ $\mathbb{F}_q$-linearly independent elements in $\mathbb{F}_{q^{m_L}}$ as defined in \eqref{eqn:Ls}. %

Consider the code $\mathcal{C}_1$ with $\mathbf{H}_L=[\mathbf{I}_J~\mathbf{0}]^\mathrm{T}$. %
Based on \eqref{connection u'}, $[\mathbf{u}'_j]_{j\in \mathcal{ J}}$ can be written as
\begin{equation}
\label{eqn:u'j}
[\mathbf{u}'_j]_{j\in \mathcal{ J}}=
[\mathbf{M}_\mathbf{o}({\mathbf{u}'}_{\varepsilon_1}^\mathrm{T})~\mathbf{M}_\mathbf{o}({\mathbf{u}'}_{\varepsilon_2}^\mathrm{T})~\ldots~\mathbf{M}_\mathbf{o}({\mathbf{u}'}_{\varepsilon_{J/{m_L}}}^\mathrm{T})]\bm{\Gamma}^\mathrm{T},
\end{equation}
where each $\mathbf{M}_\mathbf{o}({\mathbf{u}'}_{\varepsilon_s}^\mathrm{T})$, $1 \leq s \leq J/{m_L}$, denotes a matrix in $\mathbb{F}_{q^{m_L}}^{J \times m_L}$. %
Since $\mathbf{u}'_{\varepsilon_s}$ can be written as
\begin{equation*}
\mathbf{u}'_{\varepsilon_s}=L_{\mathbb{F}_{q}}^{-1}[\mathcal{B}_{1,s}~\mathcal{B}_{2,s}~\ldots~\mathcal{B}_{J/{m_L},s}]^\mathrm{T},
\end{equation*}
where each $\mathcal{B}_{i,s}$ is a basis of $\mathbb{F}_{q^{m_L}}$ over $\mathbb{F}_{q}$ defined in \eqref{Bis},  Eq. \eqref{eqn:u'j} can be expressed as
\begin{equation}
\label{psi matrix}
[\mathbf{u}'_j]_{j\in \mathcal{ J}}=L_{\mathbb{F}_{q}}^{-1}
\begin{bmatrix}
\mathbf{M}_\mathbf{o}(\mathcal{B}_{1,1}) & \mathbf{M}_\mathbf{o}(\mathcal{B}_{1,2}) &\ldots&\mathbf{M}_\mathbf{o}(\mathcal{B}_{1,{J/{m_L}}})\\
\mathbf{M}_\mathbf{o}(\mathcal{B}_{2,1})& \mathbf{M}_\mathbf{o}(\mathcal{B}_{2,2}) &\ldots&\mathbf{M}_\mathbf{o}(\mathcal{B}_{2,{J/{m_L}}})\\
\vdots\\
\mathbf{M}_\mathbf{o}(\mathcal{B}_{J/{m_L},1})& \mathbf{M}_\mathbf{o}(\mathcal{B}_{J/{m_L},2}) &\ldots&\mathbf{M}_\mathbf{o}(\mathcal{B}_{J/{m_L},{J/{m_L}}})\\
\end{bmatrix}\bm{\Gamma}^\mathrm{T}.
\end{equation}
Consequently, $[\mathbf{u}'_j]_{j\in \mathcal{ J}}\mathbf{L}_{\mathbf{m}}$ can be expressed as 
\begin{equation*}
\begin{split}
[\mathbf{u}'_j]_{j\in \mathcal{ J}}\mathbf{L}_{\mathbf{m}}=&
\begin{bmatrix}
\mathbf{M}_\mathbf{o}(\mathcal{B}_{1,1}) & \mathbf{M}_\mathbf{o}(\mathcal{B}_{1,2}) &\ldots&\mathbf{M}_\mathbf{o}(\mathcal{B}_{1,{J/{m_L}}})\\
\mathbf{M}_\mathbf{o}(\mathcal{B}_{2,1})& \mathbf{M}_\mathbf{o}(\mathcal{B}_{2,2}) &\ldots&\mathbf{M}_\mathbf{o}(\mathcal{B}_{2,{J/{m_L}}})\\
\vdots\\
\mathbf{M}_\mathbf{o}(\mathcal{B}_{J/{m_L},1})& \mathbf{M}_\mathbf{o}(\mathcal{B}_{J/{m_L},2}) &\ldots&\mathbf{M}_\mathbf{o}(\mathcal{B}_{J/{m_L},{J/{m_L}}})\\
\end{bmatrix}
\begin{bmatrix}
\mathbf{L}_{\bm{\lambda}_{1}}\\
\mathbf{L}_{\bm{\lambda}_{2}}\\
\vdots\\
\mathbf{L}_{\bm{\lambda}_{J/{m_L}}}
\end{bmatrix}\\
=&
\begin{bmatrix}
\sum\nolimits_{s=1}^{J/{m_L}}\mathbf{M}_\mathbf{o}(\mathcal{B}_{1,s})\mathbf{L}_{\bm{\lambda}_{1}}\\
\sum\nolimits_{s=1}^{J/{m_L}}\mathbf{M}_\mathbf{o}(\mathcal{B}_{2,s})\mathbf{L}_{\bm{\lambda}_{2}}\\
\vdots\\
\sum\nolimits_{s=1}^{J/{m_L}}\mathbf{M}_\mathbf{o}(\mathcal{B}_{J/{m_L},s})\mathbf{L}_{\bm{\lambda}_{J/{m_L}}}\\
\end{bmatrix}.
\end{split}
\end{equation*}
This implies that $\mathcal{C}_1$ can be expressed as follows
\begin{equation*}
\mathcal{C}_1=\begin{bmatrix}
\sum\nolimits_{s=1}^{J/{m_L}}\mathcal{M}_{1,s}\\
\sum\nolimits_{s=1}^{J/{m_L}}\mathcal{M}_{2,s}\\
\vdots\\
\sum\nolimits_{s=1}^{J/{m_L}}\mathcal{M}_{J/{m_L},s}\\
\end{bmatrix},
\end{equation*}
where for $1 \leq i,s \leq J/{m_L}$, each $\mathcal{M}_{i,s}=\{\mathbf{M}_\mathbf{o}(\mathcal{B}_{i,s})\mathbf{L}_{\bm{\lambda}_{s}}: \bm{\lambda}_{s}\in \mathbb{F}_{q^{m_L}}^k\}$ denotes an $(m_L\times n, q^{m_Lk}, d)$ Gabidulin code. %
In view of the construction of $\widetilde{\mathcal{M}}$ defined in \eqref{eqn:generalization of Gabidulin codes} with respect to $\mathcal{B}_{i,s}$ and $\mathcal{F}_{s}$, we can obtain that $\mathcal{C}_1=\widetilde{\mathcal{M}}$.

Consider the code $\mathcal{C}_2$ with $\mathbf{G}_L=[\mathbf{I}_J~\mathbf{0}]$. %
Based on \eqref{connection u}, $[\mathbf{u}_{L-j}]_{j\in \mathcal{ J}}$ can be written as
\begin{equation}
\label{eqn:uLj}
[\mathbf{u}_{L-j}]_{j\in \mathcal{ J}}=
[\mathbf{M}_\mathbf{o}(\mathbf{u}_{L-\varepsilon_1}^\mathrm{T})~\mathbf{M}_\mathbf{o}(\mathbf{u}_{L-\varepsilon_2}^\mathrm{T})~\ldots ~\mathbf{M}_\mathbf{o}(\mathbf{u}_{L-\varepsilon_{J/{m_L}}}^\mathrm{T})]\bm{\Gamma}^\mathrm{T},
\end{equation}
where each $\mathbf{M}_\mathbf{o}(\mathbf{u}_{L-\varepsilon_s}^\mathrm{T})$, $1 \leq s \leq J/{m_L}$, denotes a matrix in $\mathbb{F}_{q^{m_L}}^{J \times m_L}$. %
Since $\mathbf{u}_{L-\varepsilon_s}$ can also be written as
\begin{equation}
\mathbf{u}_{L-\varepsilon_s}=L_{\mathbb{F}_{q}}^{-1}[\mathcal{B}_{1,s}~\mathcal{B}_{2,s}~\ldots~\mathcal{B}_{J/{m_L},s}]^\mathrm{T},
\end{equation}
where each $\mathcal{B}_{i,s}$ is a basis of $\mathbb{F}_{q^{m_L}}$ over $\mathbb{F}_{q}$ defined in \eqref{Bis}, Eq. \eqref{eqn:uLj} can then be written in the form of the right-hand side of \eqref{psi matrix}. %
Consequently, the same argument as in the previous paragraph can be applied to conclude that $\mathcal{C}_2=\widetilde{\mathcal{M}}$.

\subsection{Equivalence verification for the MRD codes in Examples \ref{example 2}, \ref{example 3}, and \ref{example 5}}
\label{appendix:Equivalence verification}
Let $\mathbf{A}_i$ for $0 \leq i \leq 5$ denote the 6 codewords in \eqref{eqn:example GH 3} of Example \ref{example 2}, and let $\mathbf{B}_i$ for $0 \leq i \leq 5$ denote the 6 codewords in \eqref{eqn:example GG 5} of Example \ref{example 3}. %
We have
\begin{align*}
[\mathbf{B}_i]_{0\leq i \leq 5} = [\mathbf{A}_i]_{0 \leq i \leq 5} 
\begin{bmatrix}
  \mathbf{I}_3 & \mathbf{I}_3 & \mathbf{0} & \mathbf{0} & \mathbf{0} & \mathbf{0} \\
  \mathbf{0} & \mathbf{I}_3 & \mathbf{I}_3 & \mathbf{0} & \mathbf{0} & \mathbf{0} \\
  \mathbf{0} & \mathbf{0} & \mathbf{I}_3 & \mathbf{I}_3 & \mathbf{0} & \mathbf{0} \\
  \mathbf{0} & \mathbf{0} & \mathbf{0} & \mathbf{I}_3 & \mathbf{I}_3 & \mathbf{0} \\
  \mathbf{0} & \mathbf{0} & \mathbf{0} & \mathbf{0} & \mathbf{I}_3 & \mathbf{I}_3 \\
  \mathbf{0} & \mathbf{0} & \mathbf{0} & \mathbf{0} & \mathbf{0} & \mathbf{I}_3 
\end{bmatrix}.
\end{align*}
This implies that the $64$ codewords in the $(6\times 3, 2^6 ,3)$ circular-shift-based MRD code in Example \ref{example 2} coincide with the $64$ codewords in the $(6\times 3, 2^6 ,3)$ circular-shift-based MRD code in Example \ref{example 3}. %
Then, let $\mathbf{M}_{1,i}$, $0 \leq i \leq 6$ and $\mathbf{M}_{2,i}$, $0 \leq i \leq 6$ respectively denote the $7$ nonzero codewords given in \eqref{eqn:m tilde example 1} and in \eqref{eqn:m tilde example 2} of Example \ref{example 5}. %
We have
\begin{align*}
&[\mathbf{M}_{1,i}]_{0\leq i \leq 6} = [\mathbf{A}_i]_{0 \leq i \leq 5}
\begin{bmatrix}
  \mathbf{0}  & \mathbf{I}_3 & \mathbf{0} & \mathbf{I}_3 & \mathbf{I}_3 & \mathbf{I}_3& \mathbf{0} \\
  \mathbf{I}_3 & \mathbf{0} & \mathbf{I}_3 & \mathbf{I}_3 & \mathbf{I}_3 & \mathbf{0}& \mathbf{0} \\
  \mathbf{0} & \mathbf{I}_3 & \mathbf{I}_3 & \mathbf{I}_3 & \mathbf{0} & \mathbf{0} & \mathbf{I}_3\\
  \mathbf{I}_3 & \mathbf{I}_3 & \mathbf{I}_3 & \mathbf{0} & \mathbf{0} & \mathbf{I}_3 & \mathbf{0}\\
  \mathbf{I}_3 & \mathbf{I}_3 & \mathbf{0} & \mathbf{0} & \mathbf{I}_3 & \mathbf{0}& \mathbf{I}_3 \\
  \mathbf{I}_3 & \mathbf{0} & \mathbf{0} & \mathbf{I}_3 & \mathbf{0} & \mathbf{I}_3 & \mathbf{I}_3\\
\end{bmatrix},\\
&[\mathbf{M}_{2,i}]_{0\leq i \leq 6} = [\mathbf{A}_i]_{0 \leq i \leq 5}\begin{bmatrix}
  \mathbf{I}_3 & \mathbf{I}_3 & \mathbf{0} & \mathbf{0} & \mathbf{I}_3 & \mathbf{0} & \mathbf{I}_3\\
  \mathbf{0} & \mathbf{I}_3 & \mathbf{0} & \mathbf{I}_3 & \mathbf{I}_3 & \mathbf{I}_3& \mathbf{0} \\
  \mathbf{I}_3 & \mathbf{I}_3 & \mathbf{I}_3 & \mathbf{0} & \mathbf{0} & \mathbf{I}_3& \mathbf{0} \\
  \mathbf{0} & \mathbf{0} & \mathbf{I}_3 & \mathbf{0} & \mathbf{I}_3 & \mathbf{I}_3& \mathbf{I}_3\\
  \mathbf{0} & \mathbf{I}_3 & \mathbf{I}_3 & \mathbf{I}_3 & \mathbf{0} & \mathbf{0} & \mathbf{I}_3\\
  \mathbf{I}_3 & \mathbf{0} & \mathbf{0} & \mathbf{I}_3 & \mathbf{0} & \mathbf{I}_3& \mathbf{I}_3\\
\end{bmatrix}.
\end{align*}
This implies that the $64$ codewords in the $(6\times 3, 2^6 ,3)$ circular-shift-based MRD code in Example \ref{example 2} coincide with the $64$ codewords
in $\widetilde{\mathcal{M}}$ in Example \ref{example 5}. %

\subsection{List of Notation}
\label{appendix:list of notation}
\begin{flushleft}
\begin{tabular}{p{0.17\linewidth} p{0.73\linewidth}}
$\otimes$: & the Kronecker product.\\
$\Delta$: & the mapping from a $Jn$-dimensional row vector $\mathbf{c} = [\mathbf{c}_0 ~ \mathbf{c}_1 ~ \ldots ~ \mathbf{c}_{n-1}]$ over $\mathbb{F}_q$ to a $J\times n$ matrix $[\mathbf{c}_0^\mathrm{T}~ \mathbf{c}_1^\mathrm{T} ~ \ldots ~ \mathbf{c}_{n-1}^\mathrm{T}]$ over $\mathbb{F}_q$. \\ 
$\mathbf{A}^{\circ j}$: & the Hadamard power, that is, at every same location, the entry $a$ in $\mathbf{A}$ becomes $a^j$ in $\mathbf{A}^{\circ j}$. \\
$[\mathbf{A}_{e}]_{e \in \mathcal{E}}$: &the column-wise juxtaposition of matrices $\mathbf{A}_{e}$ with $e$ orderly chosen from a set $\mathcal{E}$.\\
$[\mathbf{A}_{i, j}]_{1 \leq i \leq M, 1 \leq j \leq N}$: & the $M \times N$ block matrix, in which every block $\mathbf{A}_{i, j}$ is the block entry with row and column indexed by $i$ and $j$, respectively. \\
$\beta$: & a primitive $L^{th}$ root of unity in $\mathbb{F}_{q^{m_L}}$.\\
$\mathcal{B}$: & an arbitrary basis of $\mathbb{F}_{q^N}$ over $\mathbb{F}_q$. \\
$\mathbf{C}_L$: &the $L\times L$ cyclic permutation matrix.\\
$\mathrm{Diag}([\mathbf{A}_i]_{0 \leq i \leq l-1}])$: & the $l \times l$ block diagonal matrix with diagonal entries equal to $\mathbf{A}_0, \ldots, \mathbf{A}_{l-1}$. \\
$\mathbb{F}_{q^N}^{N \times n}$:& the set of all $N \times n$ matrices over $\mathbb{F}_{q^N}$.\\
$\mathbb{F}_{q^N}^{n}$:& the set of all $n$-dimensional row vectors over $\mathbb{F}_{q^N}$.\\
$\mathbf{G}_L$:&the $J\times L$ matrix defined in \eqref{eqn:GH_def}.\\
$\mathbf{H}_L$:&the $L\times J$ matrix defined in \eqref{eqn:GH_def}.\\
$J$:& $|\mathcal{J}|$, equals to the number of positive integers less than $L$ that are coprime to $L$. \\
$L$: & a positive integer subject to $\gcd(q, L) = 1$.\\
$L_{\mathbb{F}_q}$: & $L \mod q$.\\
$m_L$: & the multiplicative order of $q$ modulo $L$, \emph{i.e.,} $q^{m_L}=1 \mod L$.\\
\end{tabular}
\end{flushleft}
\begin{flushleft}
\begin{tabular}{p{0.17\linewidth} p{0.73\linewidth}}
$\mathcal{J}$:& the set defined in \eqref{J}.\\
$\mathbf{M}_\mathbf{o}(\mathbf{v})$: &the $n \times N$ matrix over $\mathbb{F}_{q^N}$ defined as $[\mathbf{v}^\mathrm{T}~(\mathbf{v}^\mathrm{T})^{\circ q}~\ldots~(\mathbf{v}^\mathrm{T})^{\circ q^{N-1}}]$, where $\mathbf{v}$ is an $n$-dimensional row vector over $\mathbb{F}_{q^N}$.\\
$\mathbf{M}(\mathbf{u})$: &the $N \times n$ matrix $\mathbf{M}(\mathbf{u})=\mathbf{M}_\mathbf{o}(\mathcal{B})\mathbf{L}_\mathbf{u}$ over $\mathbb{F}_q$, where $\mathbf{M}_\mathbf{o}(\mathcal{B})$ is an $N \times N$ matrix over $\mathbb{F}_{q^N}$ defined in \eqref{eqn:MO_B} and $\mathbf{L}_\mathbf{u}$ is an $N \times n$ matrix over $\mathbb{F}_{q^N}$ defined in \eqref{eqn:L}.\\
$\widetilde{\mathcal{M}}$:&the $(J\times n, q^{Jk}, d)$ rank-metric code defined in \eqref{eqn:generalization of Gabidulin codes}. \\
$\mathbf{\Psi}_{k\times n}$:&the $k\times n$ block matrix defined in \eqref{eqn:Psi_def}.\\
$\tau(x)$:&the polynomial defined in \eqref{def: tau(x)}.\\
$\mathbf{u}$: & a $k$-dimensional row vector over $\mathbb{F}_{q^N}$. \\
$\mathbf{u}_j$: & the $(j+1)^{st}$ column of $\mathbf{U}$ defined in \eqref{U_def}.\\
$\mathbf{u}'_j$: & the $(j+1)^{st}$ column of $\mathbf{U}'$ defined in \eqref{U_def}.\\
$\mathbf{V}_L$: & the $L \times L$ Vandermonde matrix over $\mathbb{F}_{q^{m_L}}$ defined in \eqref{eqn:VL}.\\
$\widetilde{\mathbf{V}}_L$: &the $L \times L$ Vandermonde matrix over $\mathbb{F}_{q^{m_L}}$ defined in \eqref{eqn:VL tilde}.\\
\end{tabular}
\end{flushleft}


\begin{thebibliography}{99}
\bibitem{Delsarte1978}
P. Delsarte, ``Bilinear forms over a finite field, with applications to coding theory,'' \emph{Journal of Combinatorial Theory, Series A}, vol. 25, no. 3, pp. 226--241, 1978.

\bibitem{cryptosystems}
E. M. Gabidulin, ``Public-key cryptosystems based on linear codes,'' \emph{Codes and Cyphers}, vol. 16, no. 3, pp. 17--31, 1995.

\bibitem{Roth}
R. M. Roth, ``Maximum-rank array codes and their application to crisscross error correction,'' \emph{IEEE Transactions on Information Theory}, vol. 37, no. 2, pp. 328--336, 1991.

\bibitem{Silberstein}
N. Silberstein, A. S. Rawat and S. Vishwanath, ``Error resilience in distributed storage via rank-metric codes,'' in \emph{Proc. 50th Annual Allerton Conference on Communication, Control, and Computing (Allerton)}, pp. 1150--1157, 2012.


\bibitem{Koetter_RNC}
R. Koetter and F. R. Kschischang, ``Coding for errors and erasures in random network coding,'' \emph{IEEE Transactions on Information Theory},  vol. 54, no. 8, pp. 3579--3591, 2008.

\bibitem{Subspace Codes_Etzion}
T. Etzion and N. Silberstein, ``Codes and designs related to lifted MRD codes,'' \emph{IEEE Transactions on Information Theory}, vol. 59, no. 2, pp. 1004--1017, 2013. 

\bibitem{Subspace Codes_Chen_1}
L. Xu and H. Chen, ``New constant-dimension subspace codes from maximum rank distance codes,'' \emph{IEEE Transactions on Information Theory}, vol. 64, no. 9, pp. 6315--6319, 2018.

\bibitem{Subspace Codes_Chen_2}
H. Chen, X. He, J. Weng and L. Xu, ``New constructions of subspace codes using subsets of MRD codes in several blocks,'' \emph{IEEE Transactions on Information Theory}, vol. 66, no. 9, pp. 5317--5321, 2020.


\bibitem{Gabidulin}
E. M. Gabidulin, ``Theory of codes with maximum rank distance,'' \emph{Problemy Peredachi Informatsii}, vol. 21, no. 1, pp. 3--16, 1985.

\bibitem{GG}
A. Kshevetskiy and E. M. Gabidulin, ``The new construction of rank codes,'' in \emph{Proc. IEEE International Symposium on Information Theory (ISIT)}, pp. 2105--2108, 2005.

\bibitem{TG}
J. Sheekey, ``A new family of linear maximum rank distance codes,'' \emph{Advances in Mathematics of Communications}, vol. 10, no. 3, pp. 475--488, 2016.

\bibitem{GTG}
G. Lunardon, R. Trombetti and Y. Zhou, ``Generalized twisted Gabidulin codes,'' \emph{Journal of Combinatorial Theory, Series A}, 
vol. 159, pp. 79--106, 2018.

\bibitem{Otal K 1}
K. Otal and F. Özbudak, ``Additive rank metric codes,'' \emph{IEEE Transactions on Information Theory}, vol. 63, no. 1, pp. 164--168, 2016.

\bibitem{Sheekey_2019}
J. Sheekey, ``MRD codes: Constructions and connections'', \emph{Combinatorics and Finite Fields: Difference Sets, Polynomials, Pseudorandomness and Applications}, vol. 23, Walter de Gruyterpp, 2019.

\bibitem{Yildiz}
H. Yildiz and B. Hassibi, ``Gabidulin codes with support constrained generator matrices'', \emph{IEEE Transactions on Information Theory}, vol. 66, no. 6, pp. 3638--3649, 2019.

\bibitem{Li_ITW}
R. Li  and  F. -W. Fu , ``On the new rank metric codes related to Gabidulin codes'', in \emph{Proc. IEEE Information Theory Workshop (ITW)}, pp. 597--602, 2024.

\bibitem{non-Gabidulin}
A-L. Horlemann-Trautmann and K. Marshall, ``New criteria for MRD and Gabidulin codes and some rank-metric code constructions,'' \emph{Advances in Mathematics of Communications}, vol. 11, no. 3, pp. 533--548, 2017.

\bibitem{begin}
K. Otal and F. Özbudak, ``Some new non-additive maximum rank distance codes,'' \emph{Finite Fields and Their Applications}, vol. 50, pp. 293--303, 2018.

\bibitem{}
A. Cossidente, G. Marino and F. Pavese, ``Non-linear maximum rank distance codes,'' \emph{Designs, Codes and Cryptography}, vol. 79, no. 3, pp. 597--609, 2016.

\bibitem{}
B. Csajbók,  G. Marino, O. Polverino and C. Zanella, ``A new family of MRD-codes,'' \emph{Linear Algebra and its Applications}, vol. 548, pp. 203--220, 2018.


\bibitem{}
A. Ravagnani, ``Rank-metric codes and their duality theory,'' \emph{Designs, codes and Cryptography}, vol. 80, no. 1, pp. 197--216, 2016.

\bibitem{end}
E. Gorla and A. Ravagnani, ``Codes endowed with the rank metric,'' \emph{Network Coding and Subspace Designs}, pp. 3--23, 2018.


\bibitem{Tang_LNC_TIT}
H. Tang, Q. T. Sun, Z. Li, X. Yang and K. Long, ``Circular-shift linear network coding,'' \emph{IEEE Transactions on Information Theory}, vol. 65, no. 1, pp. 65--80, 2019.

\bibitem{Tang_Sun_Circular-shift_LNC_TCOM}
Q. T. Sun, H. Tang, Z. Li, X. Yang, and K. Long, ``Circular-shift linear network codes with arbitrary odd block lengths,'' \emph{IEEE
Transactions on Communications}, vol. 67, no. 4, pp. 2660--2672, 2019.

\bibitem{Jin_TIT}
S. Jin, Z. Zhai, Q. T. Sun, H. Zhang and Z. Li, ``Circular-shift-based vector linear network coding and its application to array codes,'' \emph{IEEE Transactions on Information Theory}, vol. 71, no. 12, pp. 9413--9431, 2025.


\bibitem{Rank-metric codes and applications}
E. M. Gabidulin, \emph{Rank-metric Codes and Applications}. Moscow Institute of Physics and Technology, State University, Dolgoprudny, Russia, 2011.

\bibitem{rank codes}
E. M. Gabidulin and V. Sidorenko, \emph{Rank Codes}. TUM Press, 2021.

\bibitem{Jin-arXiv}
S. Jin, Z. Zhai, Q. T. Sun and Z. Li, ``Circular-shift-based vector linear network coding and its application to array codes,'' arXiv:2412.17067, 2024.

\bibitem{tang2020circular}
H.~Tang, Q.~T. Sun, L.~Wang and T.~Yang, ``On circular-shift linear
solvability of multicast networks,'' \emph{IEEE Communications Letters},
vol.~24, no.~5, pp. 1000--1004, 2020.

\bibitem{Lidl:Finite_Field}
R. Lidl and H. Niederreiter, \emph{Finite Fields}. Cambridge University Press, 1997.

\bibitem{Su_20_TCOM}
R. Su, Q. T. Sun and Z. Zhang, ``Delay-complexity trade-off of random linear network coding in wireless broadcast,'' \emph{IEEE Transactions on Communications}, vol. 68, no. 9, pp. 5606--5618, 2020.

\bibitem{1988_TIT}
R. M. Tanner, ``A transform theory for a class of group-invariant codes,'' \emph{IEEE Transactions on Information Theory}, vol. 34, no. 4, pp. 725--775, 1988.

\bibitem{quasi-cyclic}
K. Lally and P. Fitzpatrick, ``Algebraic structure of quasicyclic codes,'' \emph{Discrete Applied Mathematics}, vol. 111, no. 1-2, pp. 157--175, 2001.


\bibitem{1-generator quasi-cyclic}
G. E. Seguin, ``A class of 1-generator quasi-cyclic codes,'' \emph{IEEE Transactions on Information Theory}, vol. 50, no. 8, pp. 1745--1753, 2004. 

\bibitem{Huang quasi-cyclic}
Q. Huang, Q. Diao, S. Lin and K. Abdel-Ghaffar, ``Cyclic and quasicyclic LDPC codes on constrained parity-check matrices and their
trapping sets,'' \emph{IEEE Transactions on Information Theory}, vol. 58, no. 5, pp. 2648--2671, 2012.

\bibitem{Blaum-Evenodd-ToC95}
M.~Blaum, J.~Brady, J.~Bruck and J.~Menon, ``EVENODD: an efficient scheme for
  tolerating double disk failures in raid architectures,'' \emph{IEEE
  Transactions on Computers}, vol.~44, no.~2, pp. 192--202, 1995.

\bibitem{RDP_2004}
P.~Corbett, B.~English, A.~Goel, T.~Grcanac, S.~Kleiman, J.~Leong and S.~Sankar, ``Row-diagonal parity for double disk failure correction,'' in \emph{Proc. The 3rd USENIX Conference on File and Storage
  Technologies}, pp. 1--14, 2004. 

\bibitem{Lv_2023_TIT}
J. Lv, W. Fang, X. Chen, J. Yang and S. -T. Xia, ``New constructions of q-ary MDS array codes with multiple parities and their effective decoding,'' \emph{IEEE Transactions on Information Theory}, vol. 69, no. 11, pp. 7082--7098, 2023.

\bibitem{Yu_2024_TCom}
L. Yu, Y. S. Han, J. Yuan and Z. Zhang, ``Variant codes based on a special polynomial ring and their fast computations,'' \emph{IEEE Transactions on Communications}, vol. 72, no. 9, pp. 5255--5267, 2024.

\bibitem{Zhai_TCom}
Z. Zhai, S. Jin, Q. T. Sun, S. Liu, X. Chen and Z. Li, ``New construction of MDS array codes and explicit characterization of decoding matrices,'' \emph{IEEE Transactions on Communications}, vol. 73, no. 8, pp. 5592--5606, 2025.


\bibitem{Hou_Basic}
H. Hou, K. W. Shum, M. Chen and H. Li, ``BASIC codes: low-complexity regenerating codes for distributed storage systems,'' \emph{IEEE Transactions on Information Theory}, vol. 62, no. 6, pp. 3053--3069, 2016.



\bibitem{matrix-rep}
W. P. Wardlaw, ``Matrix representation of finite fields,'' \emph{Mathematics Magazine}, vol. 67, no. 4, pp. 289--293, 1994.

\bibitem{Rank-metric codes and their applications}
H. Bartz, L. Holzbaur, H. Liu, S. Puchinger, J. Renner and A. Wachter-Zeh, ``Rank-metric codes and their applications,'' \emph{Foundations and Trends in Communications and Information Theory}, vol. 19, no. 3, pp. 390--546, 2022.
\end{thebibliography}
\end{document}